\documentclass[12pt,a4paper,titlepage,makeidx,reqno]{amsart}
\makeindex
\usepackage{amssymb}
\usepackage{amsfonts}
\usepackage{latexsym}
\usepackage{amsthm}
\usepackage{amsmath}
\usepackage{graphicx}

\usepackage{bm}

\newcommand\comment[1]{}

\begin{document}

%%%%%%%%%% Some definitions %%%%%%%%%%

%%%%%%%% Equations, theorems %%%%%%%%%
\renewcommand{\theequation}{\arabic{section}.\arabic{equation}}
\theoremstyle{plain}
\newtheorem{theorem}{\bf Theorem}[section]
\newtheorem{lemma}[theorem]{\bf Lemma}
\newtheorem{corollary}[theorem]{\bf Corollary}
\newtheorem{proposition}[theorem]{\bf Proposition}
\newtheorem{definition}[theorem]{\bf Definition}
\newtheorem{remark}[theorem]{\bf Remark}
%\theoremstyle{remark}
%\newtheorem{remark}[theorem]{\bf Remark}

%%%%% Alphabet %%%%%
\def\a{\alpha}  \def\cA{{\mathcal A}}     \def\bA{{\bf A}}  \def\mA{{\mathscr A}}
\def\b{\beta}   \def\cB{{\mathcal B}}     \def\bB{{\bf B}}  \def\mB{{\mathscr B}}
\def\g{\gamma}  \def\cC{{\mathcal C}}     \def\bC{{\bf C}}  \def\mC{{\mathscr C}}
\def\G{\Gamma}  \def\cD{{\mathcal D}}     \def\bD{{\bf D}}  \def\mD{{\mathscr D}}
\def\d{\delta}  \def\cE{{\mathcal E}}     \def\bE{{\bf E}}  \def\mE{{\mathscr E}}
\def\D{\Delta}  \def\cF{{\mathcal F}}     \def\bF{{\bf F}}  \def\mF{{\mathscr F}}
\def\c{\chi}    \def\cG{{\mathcal G}}     \def\bG{{\bf G}}  \def\mG{{\mathscr G}}
\def\z{\zeta}   \def\cH{{\mathcal H}}     \def\bH{{\bf H}}  \def\mH{{\mathscr H}}
\def\e{\eta}    \def\cI{{\mathcal I}}     \def\bI{{\bf I}}  \def\mI{{\mathscr I}}
\def\p{\psi}    \def\cJ{{\mathcal J}}     \def\bJ{{\bf J}}  \def\mJ{{\mathscr J}}
\def\vT{\Theta} \def\cK{{\mathcal K}}     \def\bK{{\bf K}}  \def\mK{{\mathscr K}}
\def\k{\kappa}  \def\cL{{\mathcal L}}     \def\bL{{\bf L}}  \def\mL{{\mathscr L}}
\def\l{\lambda} \def\cM{{\mathcal M}}     \def\bM{{\bf M}}  \def\mM{{\mathscr M}}
\def\L{\Lambda} \def\cN{{\mathcal N}}     \def\bN{{\bf N}}  \def\mN{{\mathscr N}}
\def\m{\mu}     \def\cO{{\mathcal O}}     \def\bO{{\bf O}}  \def\mO{{\mathscr O}}
\def\n{\nu}     \def\cP{{\mathcal P}}     \def\bP{{\bf P}}  \def\mP{{\mathscr P}}
\def\r{\rho}    \def\cQ{{\mathcal Q}}     \def\bQ{{\bf Q}}  \def\mQ{{\mathscr Q}}
\def\s{\sigma}  \def\cR{{\mathcal R}}     \def\bR{{\bf R}}  \def\mR{{\mathscr R}}
                \def\cS{{\mathcal S}}     \def\bS{{\bf S}}  \def\mS{{\mathscr S}}
\def\t{\tau}    \def\cT{{\mathcal T}}     \def\bT{{\bf T}}  \def\mT{{\mathscr T}}
\def\f{\phi}    \def\cU{{\mathcal U}}     \def\bU{{\bf U}}  \def\mU{{\mathscr U}}
\def\F{\Phi}    \def\cV{{\mathcal V}}     \def\bV{{\bf V}}  \def\mV{{\mathscr V}}
\def\P{\Psi}    \def\cW{{\mathcal W}}     \def\bW{{\bf W}}  \def\mW{{\mathscr W}}
\def\o{\omega}  \def\cX{{\mathcal X}}     \def\bX{{\bf X}}  \def\mX{{\mathscr X}}
\def\x{\xi}     \def\cY{{\mathcal Y}}     \def\bY{{\bf Y}}  \def\mY{{\mathscr Y}}
\def\X{\Xi}     \def\cZ{{\mathcal Z}}     \def\bZ{{\bf Z}}  \def\mZ{{\mathscr Z}}
\def\O{\Omega}

\def\ve{\varepsilon}   \def\vt{\vartheta}    \def\vp{\varphi}    \def\vk{\varkappa}

\def\Z{{\mathbb Z}}    \def\R{{\mathbb R}}   \def\C{{\mathbb C}}
\def\T{{\mathbb T}}    \def\N{{\mathbb N}}   \def\dD{{\mathbb D}}
\def\H{{\mathbb H}}    \def\P{{\mathbb P}}   \def\E{{\mathbb E}}

\def\rP{\mathrm{P}}    \def\rQ{\mathrm{Q}}
\def\rA{\mathrm{A}}    \def\rC{\mathrm{C}}

%%%%% Arrows %%%%%

\def\la{\leftarrow}              \def\ra{\rightarrow}            \def\Ra{\Rightarrow}
\def\ua{\uparrow}                \def\da{\downarrow}
\def\lra{\leftrightarrow}        \def\Lra{\Leftrightarrow}
\def\rra{\rightrightarrows}

%%%%% Typography %%%%%

\def\lt{\biggl}                  \def\rt{\biggr}
\def\ol{\overline}               \def\wt{\widetilde}

%%%%% Math signs %%%%%

\let\ge\geqslant                 \let\le\leqslant
\def\lan{\langle}                \def\ran{\rangle}
\def\/{\over}                    \def\iy{\infty}
\def\sm{\setminus}               \def\es{\emptyset}
\def\ss{\subset}                 \def\ts{\times}
\def\pa{\partial}                \def\os{\oplus}
\def\om{\ominus}                 \def\ev{\equiv}
\def\iint{\int\!\!\!\int}        \def\iintt{\mathop{\int\!\!\int\!\!\dots\!\!\int}\limits}
\def\el2{\ell^{\,2}}             \def\1{1\!\!1}
\def\sh{\sharp}

%%%%% Math operations %%%%%

\def\Area{\mathop{\mathrm{Area}}\nolimits}
\def\arg{\mathop{\mathrm{arg}}\nolimits}
\def\const{\mathop{\mathrm{const}}\nolimits}
\def\det{\mathop{\mathrm{det}}\nolimits}
\def\diag{\mathop{\mathrm{diag}}\nolimits}
\def\diam{\mathop{\mathrm{diam}}\nolimits}
\def\dim{\mathop{\mathrm{dim}}\nolimits}
\def\dist{\mathop{\mathrm{dist}}\nolimits}
\def\Im{\mathop{\mathrm{Im}}\nolimits}
\def\Int{\mathop{\mathrm{Int}}\nolimits}
\def\Ker{\mathop{\mathrm{Ker}}\nolimits}
\def\Lip{\mathop{\mathrm{Lip}}\nolimits}
\def\rank{\mathop{\mathrm{rank}}\limits}
\def\Ran{\mathop{\mathrm{Ran}}\nolimits}
\def\Re{\mathop{\mathrm{Re}}\nolimits}
\def\Res{\mathop{\mathrm{Res}}\nolimits}
\def\res{\mathop{\mathrm{res}}\limits}
\def\sign{\mathop{\mathrm{sign}}\nolimits}
\def\supp{\mathop{\mathrm{supp}}\nolimits}
\def\Tr{\mathop{\mathrm{Tr}}\nolimits}

\def\wind{\mathop{\mathrm{winding}}\nolimits}

\mathchardef\gRe="023C \mathchardef\gIm="023D

\newcommand\textfrac[2]{{\textstyle\frac{#1}{#2}}}

%%%%%%%%%%% Some special things %%%%%%%%%%%%

\def\SLE{\mathrm{SLE}}
\def\DS{\diamondsuit}
\def\CaraTo{\mathop{\longrightarrow}\limits^{\mathrm{Cara}}}
\def\rsa{\rightsquigarrow}

%\newcommand\rwe[2]{{\tau_{#1#2}}}  %% real weight of the edge
%\newcommand\cwe[2]{{\nu_{#1#2}}}   %% complex weight of the edge
%\newcommand\wds[1]{m^\d_\DS(#1)}          %% weight of the \DS-vertex
%\newcommand\wl[1]{m^\d_\L(#1)}            %% weight of the \L-vertex
%\newcommand\wg[1]{m^\d_\G(#1)}            %% weights of the \G-vertex

%%%%%%%%%%% End of definitions %%%%%%%%%%

\def\mesh{\delta}
\newcommand\weightG[1]{\mu^\mesh_\G(#1)}
\newcommand\weightE[2]{{\mu_{#1#2}}}
\newcommand\weightDS[1]{\mu^\mesh_\DS(#1)}
\newcommand\weightL[1]{\mu^\mesh_\L(#1)}
\newcommand\dhm[4]{\o^{#1}(#2;#3;#4)}
\renewcommand\hm[3]{\o(#1;#2;#3)}
\def\dpa{\partial^\mesh}
\def\dopa{\overline{\partial}\vphantom{\partial}^\mesh}
\def\sps{($\spadesuit$)}

\def\Pr{\mathrm{Proj}}

\title[Universality in the 2{D} {I}sing model]
{Universality in the 2{D} {I}sing model\\
 and conformal invariance of fermionic observables}

\date{\today}
\author[Dmitry Chelkak]{Dmitry Chelkak$^\mathrm{a,c}$}

\author[Stanislav Smirnov]{Stanislav Smirnov$^\mathrm{b,c}$}

\thanks{\textsc{${}^\mathrm{A}$ St.Petersburg Department of Steklov Mathematical Institute (PDMI RAS).
Fontanka~27, 191023 St.Petersburg, Russia.}}

\thanks{\textsc{${}^\mathrm{B}$ Section de Math\'ematiques, Universit\'e de Gen\`eve.
2-4 rue du Li\`evre, Case postale~64, 1211 Gen\`eve 4, Suisse.}}

\thanks{\textsc{${}^\mathrm{C}$ Chebyshev Laboratory, Department of Mathematics and
Mechanics, Saint-Petersburg State University, 14th Line, 29b, 199178 Saint-Petersburg,
Russia.}}

\thanks{{\it E-mail addresses:} \texttt{dchelkak@pdmi.ras.ru, Stanislav.Smirnov@unige.ch}}

\begin{abstract}
It is widely believed that the celebrated 2D Ising model at criticality has a  universal and
conformally invariant scaling limit, which is used in deriving many of its properties.
However, no mathematical proof has ever been given, and even physics arguments support (a
priori weaker) M\"obius invariance.
We introduce discrete holomorphic fermions
for the 2D Ising model at criticality on a large family of planar graphs.
We show that on bounded domains with appropriate boundary conditions,
those have universal and conformally invariant scaling limits,
thus proving the universality and conformal invariance conjectures.
\end{abstract}



\maketitle

\tableofcontents
\newpage

\section{Introduction}

\subsection{Universality and conformal invariance in the Ising model}

\subsubsection{Historical background}
The celebrated Lenz-Ising model is one of the simplest systems exhibiting
an order--disorder transition.
It was introduced by Lenz in \cite{lenz},
and his student Ising  proved \cite{ising} in his PhD thesis the absence of
phase transition in dimension one, wrongly conjecturing  the same picture in higher dimensions.
This belief was widely shared, and motivated Heisenberg to introduce his model
\cite{heisenberg1928theorie}.
However, some years later Peierls \cite{peierls} used estimates on the length of
interfaces between spin clusters to disprove the conjecture, showing a phase transition
in the two dimensional case.
After Kramers and Wannier \cite{kramers-wannier-i}
derived the value of the critical temperature and
Onsager \cite{onsager-i} analyzed behavior of the partition function
for the Ising model on the two-dimensional square lattice,
it became an archetypical example of the phase transition in lattice models and
in statistical mechanics in general,
see \cite{niss-ising-history-i,niss-ising-history-ii}
for the history of its rise to prominence.

Over the last six decades, thousands of papers were written about the Ising model,
with most of the literature, including this paper, restricted to the two dimensional case
(similar behavior is expected in three dimensions, but for now the complete description
remains out of reach).
The partition function and other parameters were computed exactly
%(albeit often non-rigorously)
in several different ways,
usually on the square lattice or other regular graphs.
It is thus customary to say that \emph{the 2D Ising model is exactly solvable},
though one should remark that
most of the derivations are non-rigorous, and moreover
many quantities cannot be derived by traditional methods.

Arrival of the \emph{renormalization group} formalism
(see \cite{fisher-renorm-basis} for a historical exposition)
led to an even better physical understanding, albeit still non-rigorous.
It suggests that block-spin renormalization transformation
(coarse-graining, i.e., replacing a block of neighboring sites by one)
corresponds to appropriately changing the scale and the temperature.
The Kramers-Wannier critical point arises then as a fixed point of the renormalization
transformations,
with the usual picture of stable and unstable directions.

In particular, under simple rescaling  the Ising model at the critical temperature should converge to
a \emph{scaling limit} -- a ``continuous'' version of the originally discrete Ising model,
which corresponds to a quantum field theory.
This leads to the idea of \emph{universality}:
the Ising models on different regular lattices or even more general planar graphs
belong to the same renormalization space, with a unique critical point,
and so at criticality the scaling limit and the scaling dimensions
of the Ising model should be independent of the lattice
(while the critical temperature depends on it).
Being unique, the scaling limit at the critical point is \emph{translation} and \emph{scale invariant},
which allows to deduce some information about correlations \cite{patashinskii-pokrovskii-1966,kadanoff-1966}.
By additionally postulating  invariance under inversions, one obtains
\emph{M\"obius invariance}, i.e. invariance under
global conformal transformations of the plane,
which allows \cite{polyakov-1970} to deduce more.
In seminal papers  \cite{bpz-nuclear,bpz-jsp}
Belavin, Polyakov and Zamolodchikov suggested %presented further heuristic arguments for the
much stronger \emph{full conformal invariance}
(under all conformal transformations of subregions),
thus generating an explosion of activity in
conformal field theory,
which allowed to explain non-rigorously many phenomena,
see \cite{cft-collection} for a collection
of the founding papers of the subject.
Note that in the physics literature there is sometimes confusion between the two notions,
with M\"obius invariance often called conformal invariance, though
the latter is a much stronger property.

Over the last 25 years our physical understanding of the 2D critical lattice models
%and the Ising model in particular,
%and corresponding quantum field theories
has greatly improved,
and the universality and conformal invariance
are widely accepted by the physics community.
However, supporting arguments are largely non-rigorous and some
even lack physical motivation.
This is especially awkward in the case of the Ising model,
which indeed admits many exact calculations.

\subsubsection{Our results}
The goal of this paper is to construct lattice holomorphic fermions
and to show that they have a universal conformally invariant scaling limit.
We give unambiguous (and mathematically rigorous)
arguments for the existence of the scaling limit,
its universality and conformal invariance
for some observables for the 2D Ising model at criticality,
and provide the framework to establish
the same for all observables.
By conformal invariance we mean not the M\"obius invariance,
but rather the \emph{full conformal invariance},
or invariance under conformal transformations of subregions
of $\C$.
This is a much stronger property, since
conformal transformations form an
infinite dimensional pseudogroup, unlike the M\"obius ones.
Working in subregions necessarily  leads us to consider the Ising model
in domains with appropriate boundary conditions.

At present we cannot make rigorous the renormalization
approach, but we hope that the knowledge gained
will help to do this in the future.
Rather, we use the integrable structure
to construct \emph{discrete holomorphic fermions}
in the Ising model.
For simplicity we work with discrete holomorphic functions,
defined e.g. on the graph edges,
which when multiplied by the $\sqrt{dz}$ field become
fermions or spinors.
Those functions turn out to be discrete holomorphic solutions of a discrete version
of the Riemann-Hilbert boundary value problem,
and we develop appropriate tools to show that
they converge to their continuous counterparts,
much as Courant, Friedrichs and Lewy have done
in \cite{courant-friedrichs-lewy} for
the Dirichlet problem.
The continuous versions of our boundary value problems are $\sqrt{dz}$-covariant,
and conformal invariance and universality then follow,
since different discrete conformal structures converge to the same
universal limit.

Starting from these observables, one can construct new ones, describe interfaces by the
Schramm's SLE~curves, and prove and improve many predictions originating in physics. Moreover,
our techniques work off criticality, and lead to \emph{massive} field theories and SLEs.
Several possible developments will be the subject of our future work
\cite{chelkak-smirnov-spin,smirnov-fk2,kemppainen-smirnov-fk3,hongler-smirnov-density}.

We will work with the family of \emph{isoradial graphs} or equivalently \emph{rhombic lattices}.
The latter were introduced by Duffin \cite{duffin-rhombic} in late sixties as (perhaps) the
largest family of graphs for which the Cauchy-Riemann operator admits a nice discretization.
They reappeared recently %as isoradial graphs
in the work of Mercat \cite{mercat-2001} and
Kenyon \cite{kenyon-operators}, as isoradial graphs
-- possibly the largest family of graphs were
%2D statistical mechanical models (notably
the Ising and dimer models enjoy the same integrability properties
as on the square lattice:
in particular, the critical point is well defined,
with weights depending only on the local structure.
More recently, Boutilier and de~Tili\`ere
\cite{boutillier-tiliere-periodic,boutillier-tiliere-locality}
used the Fisher representation of the Ising model by dimers and Kenyon's techniques
to calculate, among other things,  free energy for the Ising model on isoradial graphs.
While their work is closely related to ours (we can too use the Fisher representation
instead of the vertex operators to construct holomorphic fermions),
they work in the full plane and so do not address
conformal invariance.
Note that earlier eight vertex and Ising models
were considered by Baxter \cite{baxter-ptrsl-1978}
on $Z$-invariant graphs, arising from planar line arrangements.
Those graphs are topologically the same as the isoradial graphs,
and the choice of weights coincides with ours,
so quantities like partition function would coincide.
Kenyon and Schlenker \cite{kenyon-schlenker} have shown that
such graphs admit isoradial embeddings, but those change the conformal structure,
and one does not expect conformal invariance for the Ising model on
general $Z$-invariant graphs.

So there are two reasons for our choice of this particular family: firstly it seems to be the
largest family where the Ising model we are about to study is nicely defined, and secondly
(and perhaps not coincidentally) it seems to be the largest family of graphs where our main
tools, the discrete complex analysis, works well. It is thus natural to consider this family
of graphs in the context of \emph{conformal invariance} and \emph{universality} of the 2D
Ising model scaling limits.

The fermion we construct for the random cluster representation of the Ising model on domains
with two marked boundary points is roughly speaking given by the probability that the
interface joining those points passes through a given edge, corrected by a complex weight. The
fermion was proposed in \cite{smirnov-icm,smirnov-fk1} for the square lattice (see also
independent \cite{riva-cardy-hol} for its physical connections, albeit without discussion of
the boundary problem and  covariance). The fermion for the spin representation is somewhat
more difficult to construct, it corresponds to the partition function of the Ising model with
a $\sqrt{z}$ monodromy at a given edge, again corrected by a complex weight. We describe it in
terms of interfaces, but alternatively one can use a product of order and disorder operators
at neighboring site and dual site, or work with the inverse Kasteleyn's matrix for the
Fisher's dimer representation. It was introduced in \cite{smirnov-icm}, although similar
objects appeared earlier in Kadanoff and Ceva \cite{kadanoff-ceva} (without complex weight and
boundary problem discussions) and in Mercat \cite{mercat-2001} (again without discussion of
boundary problem and covariance).

Complex analysis on isoradial graphs is more complicated then on the square grid,
and less is known a priori about the Ising model there.
As a result parts of our paper are quite technical, so we would recommend reading
the much easier square lattice proofs \cite{smirnov-fk1,chelkak-smirnov-spin},
as well as the general exposition \cite{smirnov-icm,smirnov-icm2010} first.

\subsubsection{Other lattice models}
%\subsubsection{Universality and conformal invariance in other lattice models}

Over the last decade, conformal invariance of the scaling limit
was established for a number of critical lattice models.
An up-to-date introduction can be found in \cite{smirnov-icm},
so we will only touch the question of universality here.

Spectacular results of Kenyon on conformal invariance
of the dimer model, see e.g. \cite{kenyon-conformal,kenyon-gff},
%the references in his lecture course \cite{kenyon-lectures},
were originally obtained on the square lattice.
Some were extended to the isoradial case by de Tili\`ere
\cite{tiliere-isoradial-dimers}, but the questions of boundary conditions
and hence conformal invariance were not addressed yet.

Kenyon's dimer results had corollaries \cite{kenyon-long}
for the Uniform Spanning Tree (and the Loop Erased Random Walk).
Those used the Temperley
bijection between dimer and tree configurations on two coupled graphs,
so they would extend to the situations where boundary conditions can be addressed
and Temperley bijection exists.

Lawler, Schramm and Werner used in \cite{lsw-ust} simpler observables
to establish conformal invariance of the scaling limit of the UST
interfaces and the LERW curves,
and to identify them with Schramm's SLE curves.
In both cases one can obtain observables using
the Random Walk,
and for the UST one can use the Kirchhoff circuit laws
to obtain discrete holomorphic quantities.
The original paper deals with the square lattice only,
but it easily generalizes whenever boundary conditions
can be addressed.

In all those cases we have to deal with convergence of solutions of
the Dirichlet or Neumann boundary value problems to their continuous counterparts.
While this is a standard topic on regular lattices, there are technical difficulties
on general graphs; moreover functions are unbounded (e.g. observable for the LERW
is given by the Poisson kernel), so controlling their norm is far from trivial.
Tools developed by us in \cite{chelkak-smirnov-dca} for use in the current paper
however resolve most of such difficulties.

Situation is somewhat easier  with the observables for the Harmonic Explorer and Discrete
Gaussian Free Field, as discussed by Schramm and Sheffield
\cite{schramm-sheffield-harmonic,schramm-sheffield-dgff} -- both are harmonic and solving
Dirichlet problem with bounded boundary values, so the generalization from the original
triangular lattice is straight-forward. Note though that the key difficulty in the DGFF case
is to establish the martingale property of the observable.

Unlike the observables above, the one used for percolation in
\cite{smirnov-cras,smirnov-perc} is very specific to the triangular lattice,
so the question of universality is far from being resolved.

All the observables introduced so far
(except for the fermions from this paper)
are essentially bosonic,
either invariant under  conformal transformations $\varphi$ or
changing like ``pre-pre-Schwarzian'' forms,
i.e. by an addition of $\const\cdot\,\varphi'$.
They all satisfy Dirichlet or Neumann boundary conditions,
when establishing convergence is a classical subject,
dating back to Courant, Friedrichs and Lewy \cite{courant-friedrichs-lewy},
albeit in the non-bounded case one meets serious difficulties.

In the Ising case we work with fermions, hence
the Riemann-Hilbert boundary value problem
(or rather its homogeneous version due to Riemann).
Such problems turn out to be much more complicated
already on regular lattices:
near rough boundaries (which arise naturally since interfaces are fractal)
our observables blow up fast.
When working on general graphs, the main problem remains, but the tools
become quite limited.

We believe that further progress in other models requires the study
of holomorphic parafermions \cite{smirnov-icm}, so we expect
even more need to address the Riemann boundary value problems in the future.

\subsection{Setup and main results}
Throughout the paper, we work with isoradial graphs or, equivalently, rhombic lattices. A
planar graph $\G$ embedded in $\C$ is called $\mesh$-\emph{isoradial} if each face is
inscribed into a circle of a common radius $\mesh$. If all circle centers are inside the
corresponding faces, then one can naturally embed the dual graph $\G^*$ in $\C$ isoradially
with the same $\mesh$, taking the circle centers as vertices of $\G^*$. The name \emph{rhombic
lattice} is due to the fact that all quadrilateral faces of the corresponding bipartite graph
$\L$ (having $\G\cup\G^*$ as vertices and radii of the circles as edges) are rhombi with sides
of length $\mesh$. We denote the set of rhombi centers by $\DS$ (example of an isoradial graph is drawn in Fig.~\ref{Fig:IsoGraph}A). We also require the following mild assumption:
\centerline{\emph{the rhombi angles are uniformly bounded away from} $0$ \emph{and} $\pi$ }
(in other words, all these angles belong to $[\eta,\pi\!-\!\eta]$ for some fixed $\eta>0$).
Below we often use the notation \emph{const} for absolute positive constants that don't depend
on the mesh $\mesh$ or the graph structure but, in principle, may depend on $\eta$. We also
use the notation $f\asymp g$ which means that a double-sided estimate $\const_1\cdot f\le g\le
\const_2\cdot g$ holds true for some $\const_{1,2}>0$ which are independent of $\mesh$.

\begin{figure}
\centering{
\begin{minipage}[b]{0.52\textwidth}
\centering{\includegraphics[width=\textwidth]{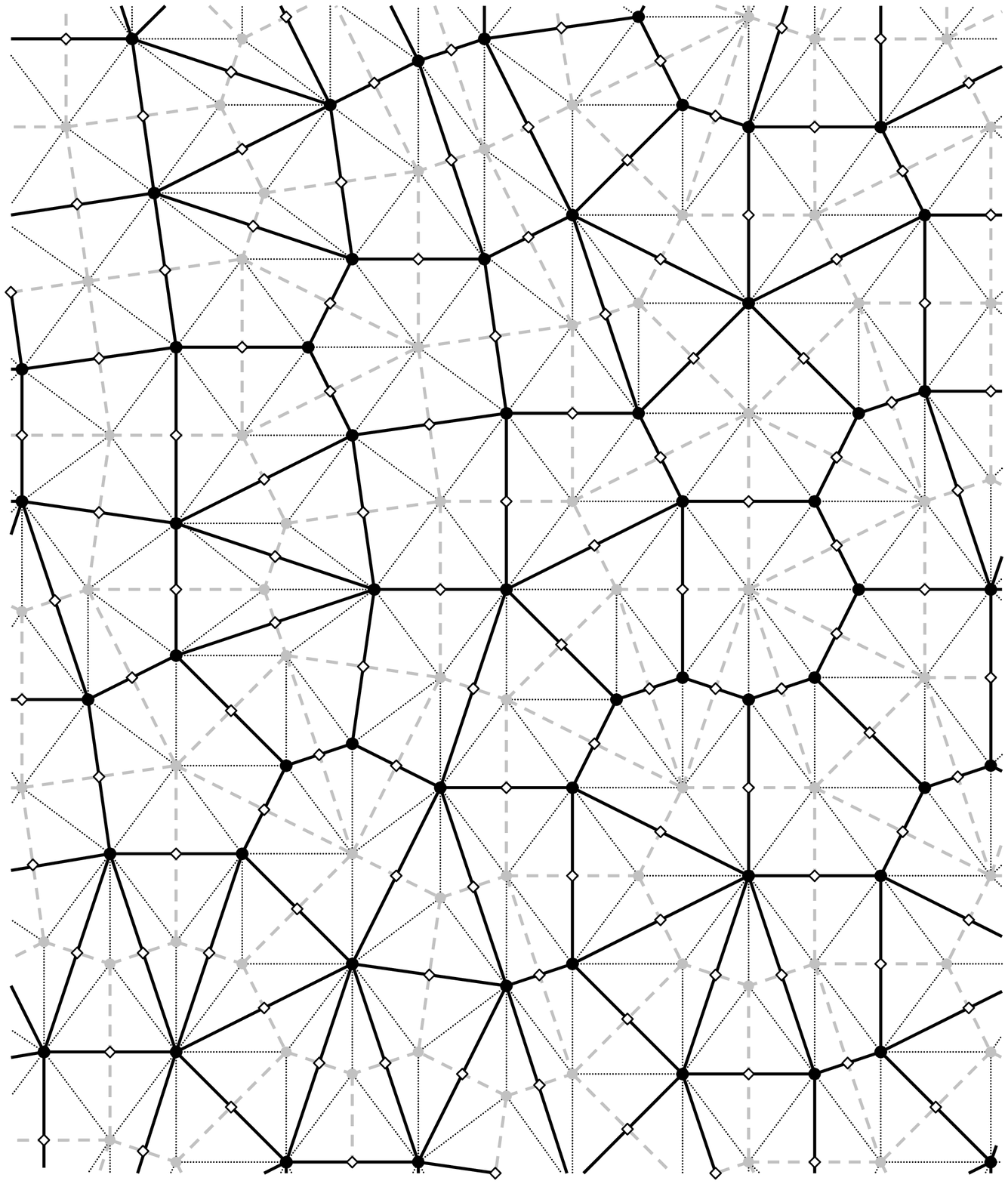}}

\bigskip
\bigskip

\centering{\textsc{(A)}}
\end{minipage}
\hskip 0.1\textwidth
\begin{minipage}[b]{0.32\textwidth}
\centering{\includegraphics[width=\textwidth]{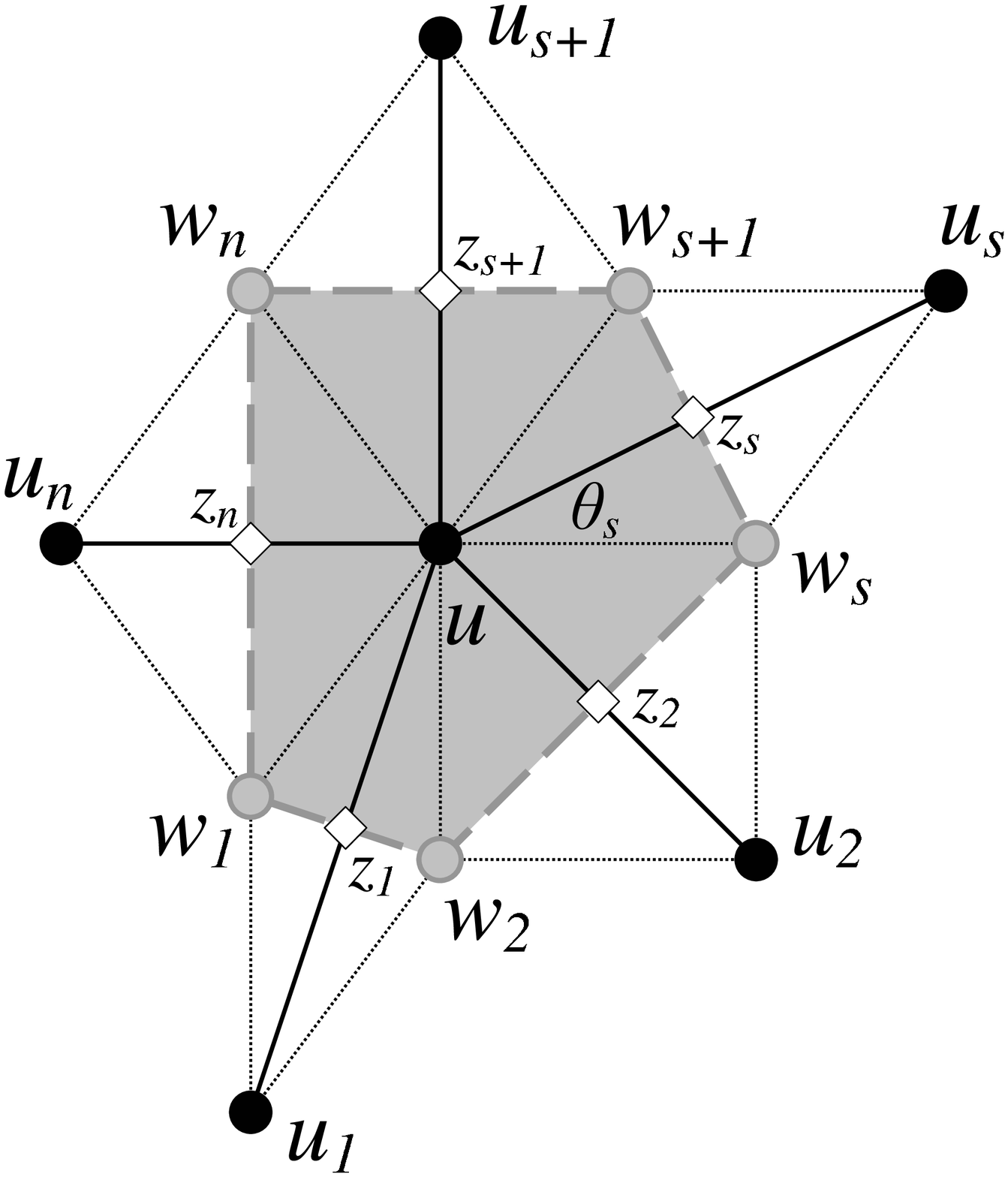}}

\centering{\textsc{(B)}}

\bigskip
\bigskip

\centering{\includegraphics[width=0.7\textwidth]{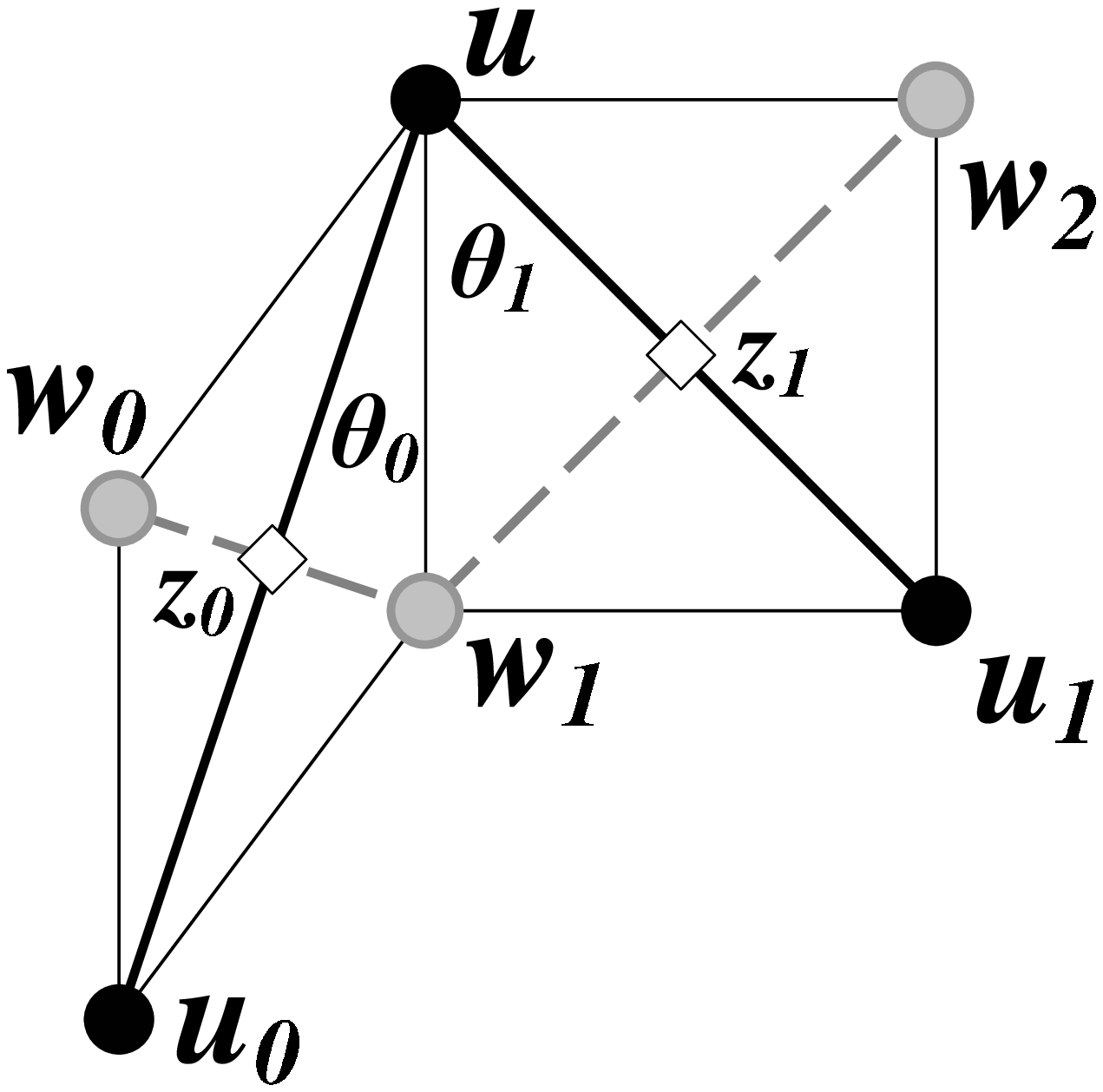}}

\centering{\textsc{(C)}

}

\end{minipage}}
\caption{\label{Fig:IsoGraph} \textsc{(A)} Example of an isoradial graph $\G$ (black vertices,
solid lines), its dual isoradial graph $\G^*$ (gray vertices, dashed lines), the corresponding
rhombic lattice or quad-graph (vertices $\L=\G\cup\G^*$, thin lines) and the set $\DS=\L^*$
(rhombi centers, white diamond-shaped vertices). \label{Fig:Notations} \textsc{(B)}
Local notation near $u\in\G$, with neighbors of $u$ enumerated counterclockwise
by $1,2\dots,s,s+1,\dots,n$. The weight $\weightG{u}$ is equal to the shaded polygon area.
\textsc{(C)} Definition of s-holomorphic functions: $F(z_0)$ and $F(z_1)$ have the same
projections on the direction $[i(w\!-\!u)]^{-\frac{1}{2}}$. Thus, we have one \emph{real}
identity for each pair of neighboring $z_0$, $z_1$.}
\end{figure}

\smallskip

It is known that one can define the (critical) Ising model on $\G^*$ so that
\renewcommand{\theenumi}{\alph{enumi}}
\begin{enumerate}
\item  the~interaction constants $J_{w_iw_j}$ are local (namely, depend on the lengths of edges
connecting $w_{i,j}\in\G^*$, $w_i\sim w_j$, only) and
\item the model is invariant under the
star-triangle transform.
\end{enumerate}%
Such invariance is widely recognized as the crucial sign of the integrability. Note that the
star-triangle transform preserves the isoradial graph/rhombic lattice structure. Moreover,
isoradial graphs form the largest family of planar graphs (embedded into $\C$) satisfying
these properties (see \cite{costa-santos-2006} and references therein). At the same time,
discrete holomorphic functions on isoradial graphs provide the simplest example of a discrete
integrable system in the so-called ``consistency approach'' to the (discrete) integrable
systems theory (see \cite{bobenko-suris-book}).

Recently, it was understood (see \cite{mercat-2001,smirnov-icm,riva-cardy-hol}) that some
objects coming from the theoretical physics approach to the Ising model (namely, products of
order and disorder operators with appropriate complex weights) can be considered and dealt
with as discrete holomorphic \emph{ functions} (see Sect.~\ref{SectSHol} for further
discussion). These functions (which we call \emph{basic observables} or \emph{holomorphic
fermions}) provide a powerful tool for rigorous proofs of several results concerning the
conformal invariance of the critical 2D Ising model. Implementing the program proposed  and
started in \cite{smirnov-icm,smirnov-fk1} for the square grid, in this paper we mainly focus
our attention on the topologically simplest case, when the model is defined in the
simply-connected (discrete) domain $\O^\mesh$ having two marked boundary points
$a^\mesh,b^\mesh$ (but see Section~\ref{SectFK4points} for more involved setup). We have to
mark some boundary points so that the conformal modulus is non-trivial, allowing us to
construct conformal invariants.

We will work with two representations of the Ising model:
the usual spin, as well as the random cluster (Fortuin-Kasteleyn).
The observables are similar, but do not directly follow from each other, and require slightly different approaches.
In either case there is an interface (between spin clusters or random clusters)
-- a discrete curve $\g^\mesh$ running from $a^\mesh$ to $b^\mesh$ inside $\O^\mesh$
(see Section~\ref{SectFKIsing},~\ref{SectSpinIsing} for precise definitions).
In both cases the basic observables are \emph{martingales} with respect to
(filtration induced by) the interface
grown progressively from $a^\mesh$, which opens the way
to identify its scaling limit as a Schramm's SLE curve, cf. \cite{smirnov-icm}.

The interface in the random cluster representation can in principle pass through some point
twice, but with our setup we move apart those passages, so that the curve becomes simple and
when arriving at the intersection it is always clear how to proceed. This setup is unique
where the martingale property holds, so there is only one conformally invariant way to address
this problem. Note that the resulting scaling limit, the $\SLE(16/3)$ curve, will have double
points. A similar ambiguity arises in the spin model (when, e.g. a vertex is surrounded by
four spins ``${-}{+}{-}{+}$''), but regardless of the way to address it (e.g. deterministic,
like always turning right, or probabilistic, like tossing a coin every time) the martingale
property always holds, and so the $\SLE(3)$ is the scaling limit. The latter is almost surely
simple, so we conclude that the double points in the discrete case produce only very small
loops, disappearing in the scaling limit.

The first two results of our paper %(Theorems \ref{ThmFKConvergence} and \ref{ThmSpinConvergence})
say that, in both representations, the holomorphic fermions are uniformly close to their
continuous \emph{conformally invariant} counterparts, independently of the structure of
$\G^\mesh$ (or $\DS^\mesh$) and the shape of $\O^\mesh$ (in particular, we don't use any
smoothness assumptions concerning the boundary). Namely, we prove
the following two theorems, formulated in detail
as Theorems \ref{ThmFKConvergence} and  \ref{ThmSpinConvergence}:

{\renewcommand{\thetheorem}{\Alph{theorem}}
%\begin{theorem}[FK-Ising fermion, see detailed formulation in Theorem \ref{ThmFKConvergence}]
\begin{theorem}[FK-Ising fermion]
Let discrete domains $(\O^\mesh;a^\mesh,b^\mesh)$ with two marked boundary points
$a^\mesh,b^\mesh$ approximate some continuous domain $(\O;a,b)$ as $\mesh\to 0$. Then,
uniformly on compact subsets of $\O$ and independently of the structure of $\DS^\mesh$,
\[
F^\mesh(z)\rightrightarrows \sqrt{\Phi'(z)},
\]
where $F^\mesh(z)=F^\mesh(z;\O^\mesh,a^\mesh,b^\mesh)$ is the discrete holomorphic fermion and
$\Phi$ denotes the conformal mapping from $\O$ onto the strip $\R\ts(0,1)$ such that $a,b$ are
mapped to $\mp\infty$.
\end{theorem}

and

%\begin{theorem}[spin-Ising fermion, see detailed formulation in Theorem \ref{ThmSpinConvergence}]
\begin{theorem}[spin-Ising fermion]
Let discrete domains $(\O^\mesh;a^\mesh,b^\mesh)$ approximate some continuous domain
$(\O;a,b)$ as $\mesh\to 0$. Then, uniformly on compact subsets of $\O$ and independently of
the structure of $\G^\mesh$,
\[
F^\mesh(z)\rightrightarrows \sqrt{\Psi'(z)},
\]
where $F^\mesh=F^\mesh(z;\O^\mesh,a^\mesh,b^\mesh)$ is the discrete holomorphic fermion and
$\Psi:\O\to\C_+$ is the conformal mapping such that $a$ and $b$ are mapped to $\infty$ and
$0$, appropriately normalized at $b$.
\end{theorem}

Because of the aforementioned martingale property,
these results are sufficient to prove the \emph{convergence of interfaces to
conformally invariant Schramm's SLE curves}
(in our case, $\SLE(3)$ for the spin representation and $\SLE(16/3)$ for the FK representation)
in the weak topology given by the convergence of driving forces in the Loewner equation,
cf. \cite{lsw-ust,smirnov-icm}. %~pp.~959--962).
A priori this is very far from establishing convergence of curve themselves, but we use
techniques of \cite{kemppainen-smirnov-curves} to prove this stronger convergence in
\cite{chelkak-smirnov-spin,smirnov-fk2}.

The third result shows how our techniques can be used to find the (conformally invariant)
limit of some macroscopic quantities, ``staying on the discrete level'', i.e. without
consideration of the limiting curves. Namely, we prove a \emph{crossing probability formula}
for the critical FK-Ising model on isoradial graphs, analogous to Cardy's formula
\cite{smirnov-cras,smirnov-perc} for critical percolation and formulated in detail as
Theorem~\ref{ThmFK4points}:
\begin{theorem}[FK-Ising crossing probability]
%\begin{theorem}[FK-Ising crossing probability, see detailed formulation in Theorem \ref{ThmFK4points}]
Let discrete domains
$(\O^\mesh;a^\mesh,b^\mesh,c^\mesh,d^\mesh)$
with %four  marked counterclockwise boundary points and
alternating (wired/free/wired/free) boundary conditions on four sides
approximate some continuous topological quadrilateral $(\O;a,b,c,d)$ as $\mesh\to 0$.
Then the probability of an FK cluster crossing
between two wired sides has a scaling limit,
which depends only on the conformal modulus of the limiting quadrilateral,
and is given for the half-plane by
\begin{equation}
\label{CrossP} p(\H;0,1\!-\!u,1,\infty)= %\Pi(\H;[1\!-\!u,1]\lra[\infty,0]) =
\frac{\sqrt{1-\sqrt{1\!-\!u}}}{\sqrt{1\!-\!\sqrt{u}}+\sqrt{1\!-\!\sqrt{1\!-\!u}}}\,,
~~u\in[0,1].
%\lt[1+\frac{\sin(\frac{\pi}{4}-\frac{\phi}{2})}{\sin(\frac{\phi}{2})}\rt]^{-1}=.
\end{equation}
\end{theorem}
\noindent The version of this formula for multiple SLEs was derived by Bauer, Bernard and
Kyt\"ol\"a in \cite{bauer-bernard-kytola-multiple}, see page 1160, their notation for the
modulus related to ours by $x=1-u$. Besides being of an independent interest, this result
together with \cite{kemppainen-smirnov-curves} is needed to improve the
topology of convergence of FK-Ising interfaces. % to $\SLE(16/3)$.
Curiously, the (macroscopic) answer for a unit disc
$(\dD;-e^{i\phi},e^{-i\phi},e^{i\phi},-e^{-i\phi})$ formally coincides with the relative
weights corresponding to two possible crossings inside (microscopic) rhombi (see
Fig.~\ref{Fig:FKDomain}A) in the critical model (see Remark~\ref{RemFK4pointsDisc}). }

\subsection{Organization of the paper}
We begin with the definition of Fortuin-Kasteleyn (random cluster) and spin representations of
the critical Ising model on isoradial graphs in Section~\ref{SectModel}.
From the outset we work with critical interactions, but in principle one can introduce
a temperature parameter, which would lead to \emph{massive holomorphic fermions}.
We also introduce the
basic discrete holomorphic observables (holomorphic fermions) satisfying the martingale
property with respect to the growing interface and, essentially, show
that they satisfy discrete version of the Cauchy-Riemann equation
(Proposition~\ref{PropFKfermionDef} and Proposition~\ref{PropSpinFermionHol})
using some simple combinatorial bijections between the sets of configurations. Actually, we
show that our observables satisfy the stronger %(than the usual discrete holomorphicity)
``two-points'' equation which we call
 \emph{spin} or \emph{strong holomorphicity},
or simply  \emph{s-holomorphicity}.

We discuss the properties of s-holomorphic functions in Section~\ref{SectSHolFunct}. The main
results are:
\renewcommand{\theenumi}{\alph{enumi}}
\begin{enumerate}
\item The (rather miraculous) possibility to define naturally the discrete version of
$h(z)=\Im\int (f(z))^2dz$, see Proposition \ref{PropHDef}. Note that the square $(f(z))^2$ of
a discrete holomorphic function $f(z)$ is not discrete holomorphic anymore,
%so the integral itself is not well-defined, at least in a straightforward way
but unexpectedly it satisfies ``half'' of the Cauchy-Riemann equations,
making its imaginary part a closed form with a well-defined integral;
\item The sub- and super-harmonicity of $h$ on the original isoradial graph $\G$ and
its dual $\G^*$, respectively, and the a priori comparability of the components $h\big|_\G$
and $h\big|_{\G^*}$ which allows one to deal with $h$ as with a harmonic function: e.g.
nonnegative $h$'s satisfy a version of the Harnack Lemma (see Section~\ref{SectHarnackH});
\item The uniform (w.r.t. $\mesh$ and the structure of the isoradial graph/rhombic lattice $\G,\G^*/\DS$) boundedness
and, moreover, uniform Lipschitzness of s-holomorphic functions inside their domains of
definition $\O^\mesh$, with the constants depending on $M=\max_{v\in\O^\mesh}|h(v)|$ and the
distance $d=\dist(z;\pa\O^\mesh)$ only (Theorem~\ref{FboundH}, these results should be
considered as discrete analogous of the standard estimates from the classical complex
analysis);
\item The combinatorial trick (see Section~\ref{SectBModTrick}) that allows us to
transform the discrete version of the Riemann-type boundary condition $f(\z)\parallel
(\t(\z))^{-\frac{1}{2}}$ into the Dirichlet condition for $h\big|_{\pa\O}$ on both
$\G$~and~$\G^*$, thus completely avoiding the reference to Onsager's magnetization estimate
used in \cite{smirnov-icm,smirnov-fk1} to control the difference $h\big|_\G-h\big|_{\G^*}$ on the
boundary.
\end{enumerate}

We prove the (uniform) convergence of the basic observable in the FK-Ising model to its
continuous counterpart in Section~\ref{SectFKobservable}. The main result here is
Theorem~\ref{ThmFKConvergence}, which is the technically simplest of our main theorems, so the
reader should consider the proof as a basic example of the application of our techniques.
Besides the results from \cite{chelkak-smirnov-dca} and the previous Sections, the important
idea (exactly as in \cite{smirnov-fk1}) is to use some compactness arguments (in the set of
all simply-connected domains equipped with the Carath\'eodory topology) in order to derive the
uniform (w.r.t. to the shape of $\O^\mesh$ and the structure of $\DS^\mesh$) convergence from
the ``pointwise'' one.

In Section~\ref{SectSpinObservable} we prove analogous convergence result for the holomorphic
fermion defined for the spin representation of the critical Ising model
(Theorem~\ref{ThmSpinConvergence}). There are two differences from the preceding Section: the
unboundedness of the (discrete) integral $h=\Im\int (f(z))^2dz$ (this prevents us from the
immediate use of compactness arguments) and the need to consider the normalization of our
observable at the target point $b^\mesh$ (this is crucial for the martingale property). In
order to handle the normalization at $b^\mesh$, we assume that our domains $\O^\mesh$ contain
a (macroscopic) rectangle near $b^\mesh$ and their boundaries $\pa\O^\mesh$ approximate the
corresponding straight segment as $\mesh\to 0$.
Making this technical assumption,
we don't lose much generality,
since the growing interface, though fractal in the limit, doesn't
change the shape of the domain near $b^\mesh$. Then, we use a version of the boundary Harnack
principle (Proposition~\ref{PropBoundHar}) in order to control the values of $h$ in the bulk
through the fixed value $f(b^\mesh)$. Another important technical ingredient is the universal
(w.r.t. to the structure of $\DS^\mesh$) multiplicative normalization of our observable.
Loosely speaking, we define it using the value at $b^\mesh$ of the discrete holomorphic
fermion in the discrete half-plane (see Theorem~\ref{ThmCFexists} for further details).

Section~\ref{SectFK4points} is devoted to the crossing probability formula
for the FK-Ising model on discrete quadrilaterals $(\O^\mesh;a^\mesh,b^\mesh,c^\mesh,d^\mesh)$
(Theorem~\ref{ThmFK4points}, see also Remark~\ref{RemFK4pointsDisc}). The main idea here is to
construct some discrete holomorphic in $\O^\mesh$ function whose boundary values reflect the
conformal modulus of the quadrilateral. Namely, in our construction, discrete functions
$h^\mesh=\Im\int(f^\mesh(z))^2dz$ approximate the imaginary part of the conformal mapping from
$\O^\mesh$ onto the slit strip $[\R\ts(0,1)]\setminus (-\infty+i\vk;i\vk]$ such that $a^\mesh$
is mapped to the ``lower''~$-\infty$, $b^\mesh$ to~$+\infty$; $c^\mesh$ to the
``upper''~$-\infty$ and $d^\mesh$ to the tip~$i\vk$. The respective crossing probabilities are
in the $1$-to-$1$ correspondence with values $\vk^\mesh$ which approximate $\vk$ as $\mesh\to
0$. Since $\vk$ is uniquely determined by the limit of conformal moduli of
$(\O^\mesh;a^\mesh,b^\mesh,c^\mesh,d^\mesh)$, we obtain (\ref{CrossP}) (see further details in
Section~\ref{SectFK4points}). Finally, Appendix contains several auxiliary lemmas: estimates
of the discrete harmonic measure, discrete version of the Cauchy formula, and technical
estimates of the Green function in the disc. We refer the reader interested in a more detailed
presentation of the discrete complex analysis on isoradial graphs to our paper
\cite{chelkak-smirnov-dca}.

\medskip

\noindent {\bf Acknowledgments.} This research was supported by the Swiss N.S.F., by the
European Research Council AG CONFRA, and by the Chebyshev
Laboratory (Department of Mathematics and Mechanics, Saint-Petersburg State University) under
the grant 11.G34.31.0026 of the Government of the Russian Federation. The first author was
partially funded by P.Deligne's 2004 Balzan prize in Mathematics and by the grant MK-7656.2010.1.

\section{Critical spin- and FK-Ising models on isoradial graphs.\\
 Basic observables (holomorphic fermions)}
\label{SectModel}

\subsection{Critical FK-Ising model}
\label{SectFKIsing}

\subsubsection{Loop representation of the model, holomorphic fermion, martingale property}

\begin{figure}
\centering{
\begin{minipage}[b]{0.3\textwidth}

\centering{\includegraphics[width=0.6\textwidth]{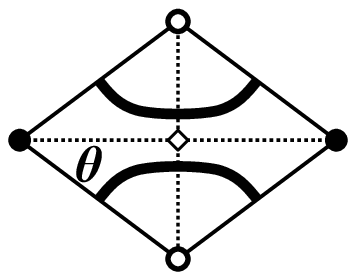}}
\smallskip
$\theta_\mathrm{config.}(z)=\theta$,
\smallskip

relative weight $=~\sin\tfrac{1}{2}\theta$

\bigskip \bigskip \medskip

\centering{\includegraphics[width=0.6\textwidth]{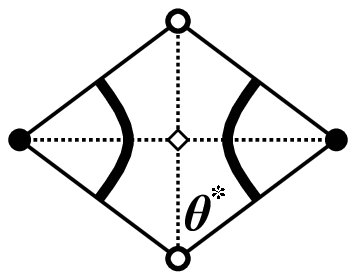}}

\smallskip

$\theta_\mathrm{config.}(z)=\theta^*=\tfrac{\pi}{2}-\theta$,

\smallskip

relative weight $=~\sin\tfrac{1}{2}\theta^*$

\bigskip \bigskip \medskip

\centering{\textsc{(A)}}
\end{minipage}
%\hskip 0.05\textwidth
\begin{minipage}[b]{0.64\textwidth}
\centering{\includegraphics[width=\textwidth]{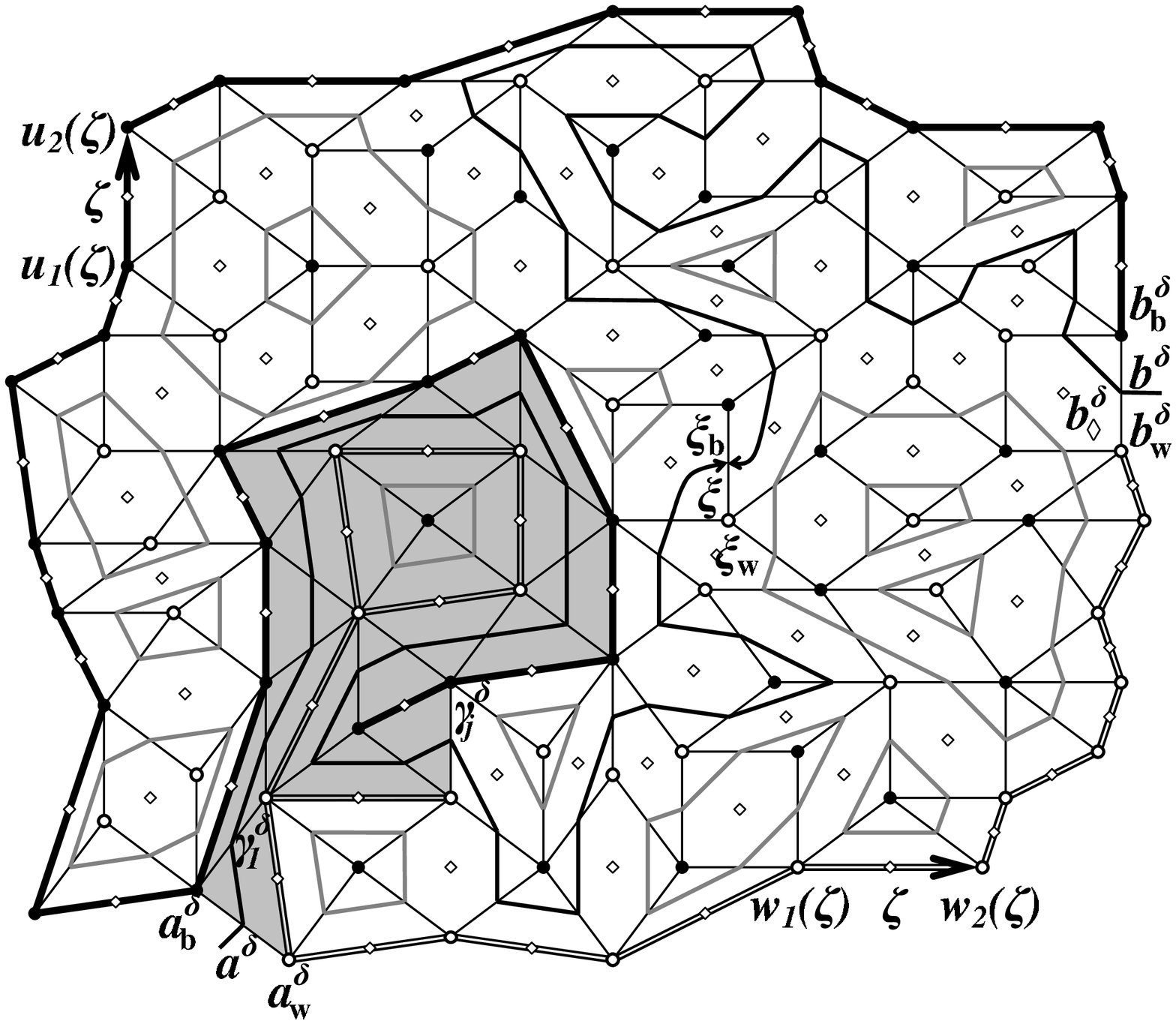}}

\centering{\textsc{(B)}}
\end{minipage}
} \caption{\label{Fig:FKDomain} \textsc{(A)}~Loop representation of the critical FK-Ising model on isoradial graphs: the
relative weights corresponding to two possible choices of connections inside the inner rhombus
$z$ (the partition function is given by (\ref{Z-FK})). \textsc{(B)} Discrete domain
$\O^\mesh_\DS$ with a sample configuration. Besides loops, there is an interface
$\g^\mesh$ connecting $a^\mesh$ to $b^\mesh$. Calculating the $\wind(\g^\mesh;\,b^\mesh\rsa
\xi)$, we draw $\g^\mesh$ so that it intersects the edge
$\xi=[\xi_\mathrm{b}\xi_\mathrm{w}]$ orthogonally. As $\g^\mesh$ grows, it separates some part
of $\O^\mesh_\DS$ (shaded) from $b^\mesh$. We denote by
$\O^\mesh_\DS\setminus[a^\mesh\g^\mesh_1..\g^\mesh_j]$ the connected component containing
$b^\mesh$ (unshaded).}
\end{figure}

We will work with a graph domain which can be thought of
as a discretization of a simply-connected planar domain with two marked boundary points.
Let \mbox{$\O^\mesh_\DS\!\ss\!\DS$} be a simply-connected discrete domain composed of inner
rhombi \mbox{$z\!\in\!\Int\O^\mesh_\DS$} and boundary half-rhombi $\z\in\pa\O^\mesh_\DS$, with
two marked boundary points $a^\mesh$, $b^\mesh$ and Dobrushin boundary conditions (see
Fig.~\ref{Fig:FKDomain}B): $\pa\O^\mesh_\DS$ consists of the ``white'' arc $a^\mesh_{\mathrm{w}}
b^\mesh_{\mathrm{w}}$, the ``black'' arc $b^\mesh_{\mathrm{b}}a^\mesh_{\mathrm{b}}$, and two edges
$[a^\mesh_{\mathrm{b}}a^\mesh_{\mathrm{w}}]$, $[b^\mesh_{\mathrm{b}}b^\mesh_{\mathrm{w}}]$ of
$\L$. Without loss of generality, we assume that
\begin{center}
\emph{$b^\mesh_{\mathrm{b}}-b^\mesh_{\mathrm{w}}=i\mesh$, i.e., the edge
$b^\mesh=[b^\mesh_\mathrm{b}b^\mesh_\mathrm{w}]$ is oriented vertically.}
\end{center}

For each inner rhombus $z\in\Int\O^\mesh_\DS$ we choose one of two possibilities to connect
its sides (see Fig.~\ref{Fig:FKDomain}A, there is only one choice for boundary half-rhombi),
thus obtaining the set of configurations (whose cardinality is
$2^{\mathrm{\#(\Int\O^\mesh_\DS)}}$). The {\bf partition function} of the critical FK-Ising
model is given by
\begin{equation}
\label{Z-FK} Z=\sum_{\mathrm{config.}}~{\textstyle
\sqrt{2}^{\,\mathrm{\#(loops)}}\prod_{z\in\Int\O^\mesh_\DS}\sin\frac{1}{2}\theta_{\mathrm{config.}}(z)},
\end{equation}
where $\theta_{\mathrm{config.}}(z)$ is equal to either $\theta$ or
$\theta^*=\frac{\pi}{2}-\theta$ depending on the choice of connections inside rhombus $z$ (see
Fig.~\ref{Fig:FKDomain}A).

We described the loop representation, since at criticality it is easier to work with, than the
usual random cluster one. The loops trace the perimeters of random clusters, and the curve
joining the two marked boundary points is the interface between a cluster and a dual cluster
wired on two opposite boundary arcs (the so-called Dobrushin boundary conditions).

Let $\xi=[\xi_\mathrm{b} \xi_\mathrm{w}]$ be some inner edge of $\O^\mesh_\DS$ (where
$\xi_\mathrm{b}\in\G$, $\xi_\mathrm{w}\in\G^*$).
Due to the boundary conditions chosen, each configuration consists of (a number of)
loops and one interface $\g^\mesh$ running from $a^\mesh$ to $b^\mesh$.
The {\bf holomorphic fermion} is defined as
\begin{equation}
\label{FKfermionXiDef} F^\mesh(\xi)=F^\mesh_{(\O^\mesh;\,a^\mesh,b^\mesh)}(\xi):=
(2\mesh)^{-\frac{1}{2}}\cdot \E\left[\,\c(\xi\!\in\!\g^\mesh)\cdot
e^{-\frac{i}{2}\wind(\g^\mesh;\,b^\mesh\rsa\xi)}\,\right],
\end{equation}
where $\c(\xi\!\in\!\g^\mesh)$ is the indicator function of the event that the interface intersects
$\xi$ and
\begin{equation}
\label{WindBZ=} \wind(\g^\mesh;\,b^\mesh\rsa\xi) =
\wind(\g^\mesh;\,a^\mesh\rsa\xi)-\wind(\g^\mesh;\,a^\mesh\rsa b^\mesh)
\end{equation}
denotes the total turn of $\g^\mesh$ measured (in radians) from $b^\mesh$ to~$\xi$. Note that,
for all configurations and edges $\xi$ one has (see Fig.~\ref{Fig:FKDomain}B),
\[
e^{-\frac{i}{2}\wind(\g^\mesh;\,b^\mesh\rsa\xi)}~\parallel\
[i(\xi_\mathrm{w}\!-\!\xi_\mathrm{b})]^{-\frac{1}{2}}\,.
\]
%, here and below the sign of the square root is, obviously, unimportant).

\begin{remark}[\bf Martingale property]
For each $\xi$,
$F^\mesh_{(\O^\mesh\setminus[a^\mesh\g^\mesh_1..\g^\mesh_j];\g^\mesh_j,b^\mesh)}(\xi)$ is a
martingale with respect to the growing interface $(a^\mesh=\g^\mesh_0, \g^\mesh_1,
 ..., \g^\mesh_j, ...)$ (till the stopping time when $\g^\mesh$ hits
$\xi$ or $\xi$~becomes separated from $b^\mesh$ by the interface, see
Fig.~\ref{Fig:FKDomain}B).
\end{remark}

\begin{proof}
Since the $\wind(\g^\mesh;\,b^\mesh\rsa\xi)$ doesn't depend on the beginning of the interface,
the claim immediately follows from the total probability formula.
\end{proof}

\subsubsection{Discrete boundary value problem for $F^\mesh$.}

We start with the extension of $F^\mesh$ to the centers of rhombi $z\in\O^\mesh_\DS$.
Actually, the (rather fortunate) opportunity to use the definition given below reflects the discrete
holomorphicity of $F^\mesh$ (see discussion in Section~\ref{SectSHol}).

\begin{proposition}
\label{PropFKfermionDef} Let $z\in\Int\O^\mesh_\DS$ be the center of some inner rhombus
$u_1w_1u_2w_2$. Then, there exists a complex number $F^\mesh(z)$ such that
\begin{equation}
\label{FKfermionDef}
F^\mesh([u_jw_k])=\Pr\left[\,F^\mesh(z)\,;\,[i(w_k\!-\!u_j)]^{-\frac{1}{2}}\right],\quad
j,k=1,2.
\end{equation}
\end{proposition}
\noindent The proposition essentially states that $F^\mesh$ is \emph{spin holomorphic} as
specified in Definition~\ref{DefSHol} below. By $\Pr[X;u]$ we denote the orthogonal projection
of the vector $X$ on the vector $u$, which is parallel to $u$ and equal to
\[
\Pr[X;u]=\Re\lt(X\frac{\ol{u}}{|u|}\rt)\frac{u}{|u|}=\frac{X}{2}+\frac{\ol{X}u^2}{2|u|^2}
\]
(here we consider complex numbers as vectors). Because of the latter rewriting, the choice of
the sign in the square root in \eqref{FKfermionDef} does not matter.

\begin{figure}
\centering{\includegraphics[width=0.7\textwidth]{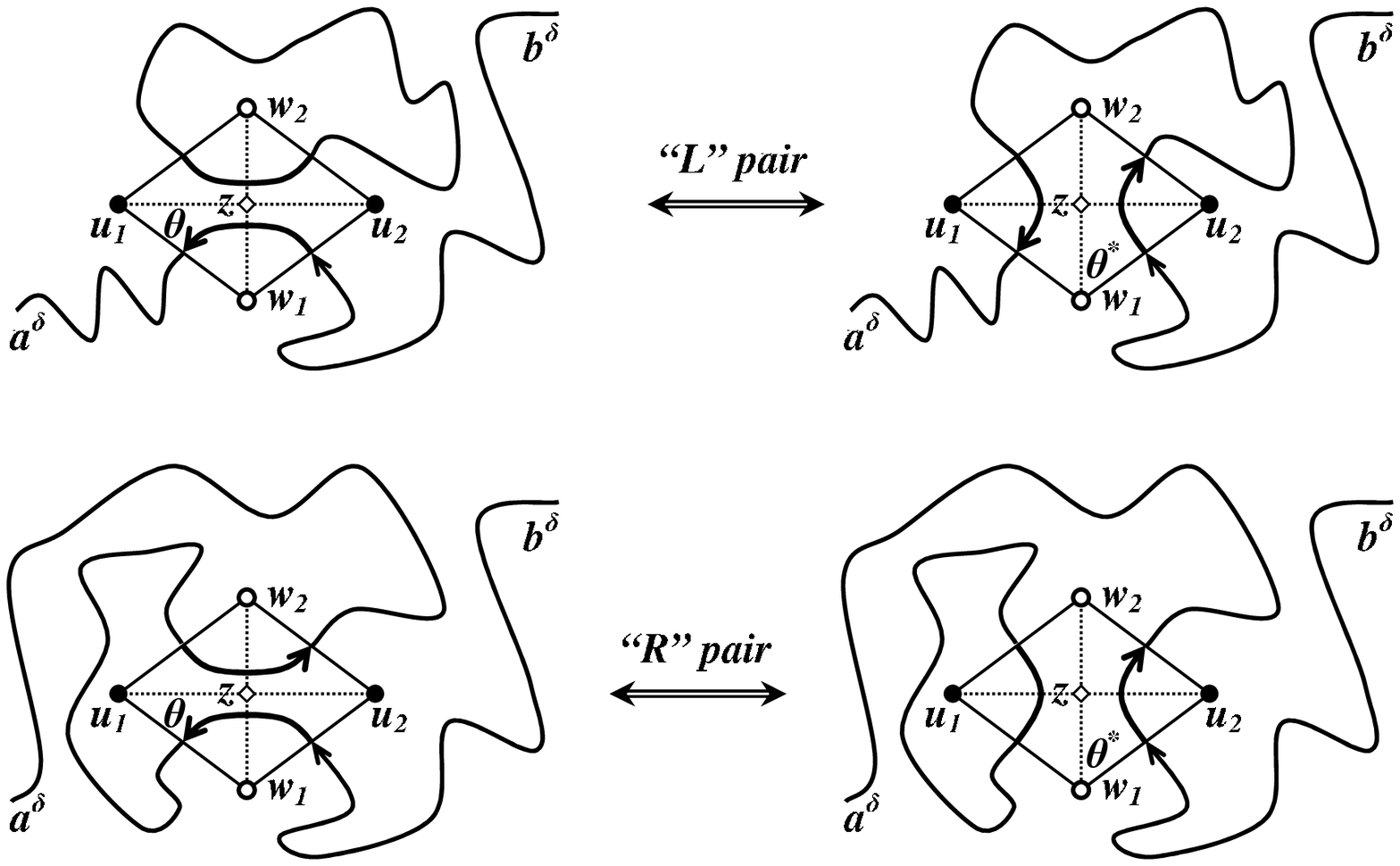}}
\caption{\label{Fig:FKBijection} Local rearrangement at $z$: the bijection between
configurations. Without loss of generality, we may assume that the (reversed, i.e., going
from $b^\mesh$) interface enters the rhombus $z$ through the edge $[u_2w_1]$ (the case
$[u_1w_2]$ is completely similar). There are two possibilities: either $\g^\mesh$ (finally)
leaves $z$ through $[u_1w_1]$ (``L'',~left turn) or through $[u_2w_2]$ (``R'',~right turn). In
view of (\ref{Z-FK}), the relative weights of configurations are
$\sqrt{2}\sin\tfrac{1}{2}\theta$, $\sin\tfrac{1}{2}\theta^*$ (``L'' pairs), and
$\sin\tfrac{1}{2}\theta$, $\sqrt{2}\sin\tfrac{1}{2}\theta^*$ (``R'' pairs).}
\end{figure}

\begin{proof}
As on the square grid (see \cite{smirnov-fk1}), the proof is based on the bijection between
configurations which is produced by their local rearrangement at $z$. It is sufficient to check
that the contributions to
$F^\mesh([u_jw_k])$'s
of each pair of configurations drawn in Fig.~\ref{Fig:FKBijection} are
the specified projections
of the same complex number. The relative contributions of configurations (up to the
same real factor coming from the structure of the configuration away from $z$) to the values of
$F$ on four edges around $z$ for the pairs ``L'' and ``R'' are given by

\begin{center}\begin{tabular}{||c||c|c|c|c||} \hline\hline
$\vphantom{\big|^|_|}$ & $e^{i\frac{\vp}{2}}\!\cdot F^\mesh([u_2w_1])$ &
$e^{i\frac{\vp}{2}}\!\cdot F^\mesh([u_2w_2])$ & $e^{i\frac{\vp}{2}}\!\cdot F^\mesh([u_1w_2])$
& $e^{i\frac{\vp}{2}}\!\cdot F^\mesh([u_1w_1])$ \cr \hline\hline
$\vphantom{\big|^|_|}\mathrm{L}$ & $\sqrt{2}\sin\frac{\theta}{2}+\sin\frac{\theta^*}{2}$ &
$e^{i\theta}\sin\frac{\theta^*}{2}$ & $-i\sin\frac{\theta^*}{2}$ &
$e^{-i\theta^*}[\sqrt{2}\sin\frac{\theta}{2}+\sin\frac{\theta^*}{2}]$ \cr \hline
$\vphantom{\big|^|_|}\mathrm{R}$ & $\sin\frac{\theta}{2}+\sqrt{2}\sin\frac{\theta^*}{2}$ &
$e^{i\theta}[\sin\frac{\theta}{2}+\sqrt{2}\sin\frac{\theta^*}{2}]$ & $i\sin\frac{\theta}{2}$ &
$e^{-i\theta^*}\sin\frac{\theta}{2}$ \cr \hline\hline
\end{tabular}\end{center}

\medskip

\noindent where $\vp$ denotes the total turn of the interface traced
from $b^\mesh$ to $[u_2w_1]$.
An easy trigonometric calculation using that $\theta+\theta^*=\pi/2$ shows
that
\[
\textstyle \left[\sqrt{2}\sin\frac{\theta}{2}+\sin\frac{\theta^*}{2}\right]
~-~i\sin\frac{\theta^*}{2} ~=~e^{i\theta}\sin\frac{\theta^*}{2}~+~
e^{-i\theta^*}\left[\sqrt{2}\sin\frac{\theta}{2}+\sin\frac{\theta^*}{2}\right]~.
\]
%\[
%\textstyle [\sqrt{2}\sin\frac{\theta}{2}+\sin\frac{\theta^*}{2}]\cdot \cos\theta -
%\sin\frac{\theta^*}{2}\cdot \sin\theta = \sin\frac{\theta^*}{2},
%\]
%\[
%\textstyle [\sqrt{2}\sin\frac{\theta}{2}+\sin\frac{\theta^*}{2}]\cdot \cos\theta^* +
%\sin\frac{\theta^*}{2}\cdot \sin\theta^* =
%[\sqrt{2}\sin\frac{\theta}{2}+\sin\frac{\theta^*}{2}],
%\]
Denoting the common value of the two sides by $e^{i\frac{\vp}{2}}\cdot F^\mesh(z)$ and
observing that $1\perp i$ and $e^{i\theta}\perp e^{-i\theta^*}$, we conclude that the first
row (``L'') describes the four projections of $e^{i\frac{\vp}{2}}\cdot F^\mesh(z)$ onto the
lines $\R$, $e^{i\theta}\R$, $i\R$ and $e^{-i\theta^*}\R$, respectively. Multiplying by the
common factor $e^{-i\frac{\vp}{2}}$ (which is always parallel to
$[i(w_1\!-\!u_2)]^{-\frac{1}{2}}$), we obtain the result for ``L'' pairs of configurations.
The interchanging of $\theta$ and $\theta^*$ yields the result for ``R'' pairs.
\end{proof}

\begin{remark} For $\z\in\pa\O^\mesh_\DS$ we define $F^\mesh(\z)$
so that (\ref{FKfermionDef}) holds true (in this case, only two projections are meaningful,
so $F^\mesh(\z)$ is easily and uniquely defined). Note that all interfaces passing through the
half-rhombus $\z$ intersect both its sides. Moreover, since the winding of the interface at
$\z$ is independent of the configuration chosen (and coincides with the winding of the
corresponding boundary arc) for topological reasons, we have
\begin{equation}
\label{FK-BC} F^\mesh(\z)\parallel (\t(\z))^{-\frac{1}{2}},\quad \z\in\pa \O^\mesh_\DS,
\end{equation}
where (see Fig.~\ref{Fig:FKDomain}B)
\[
\begin{array}{lll}
\t(\z)=w_2(\z)\!-\!w_1(\z), & \z\in (a^\mesh b^\mesh), & w_{1,2}(\z)\in\G^*, \cr
\t(\z)=u_2(\z)\!-\!u_1(\z), & \z\in (b^\mesh a^\mesh), & u_{1,2}(\z)\in\G,
\end{array}
\]
is the `` discrete tangent vector'' to $\pa\O^\mesh_\DS$ directed from $a^\mesh$ to $b^\mesh$
on both boundary arcs.
\end{remark}

Thus, we arrive at

\medskip

\noindent {\bf Discrete Riemann boundary value problem for $\bm{F^\mesh}$ (FK-case):}
\emph{The function $F^\mesh$ is defined in $\O^\mesh_\DS$ and for each pair of neighbors
$z_0,z_1\in\O^\mesh_\DS$, $z_0\sim z_1$, the
discrete holomorphicity condition holds:
\begin{equation}
\label{FKSHol} \Pr\left[F^\mesh(z_0)\,;[i(w\!-\!u)]^{-\frac{1}{2}}\right]=
\Pr\left[F^\mesh(z_1)\,;[i(w\!-\!u)]^{-\frac{1}{2}}\right].
\end{equation}
Moreover, $F^\mesh$ satisfies the boundary conditions (\ref{FK-BC}) and, since all interfaces
pass through $b^\mesh$, satisfies the normalization $F^\mesh(b^\mesh) = \Re
F^\mesh(b^\mesh_\DS) = (2\mesh)^{-\frac{1}{2}}$.}

\subsection{Critical spin-Ising model}
\label{SectSpinIsing}

\subsubsection{Definition of the model, holomorphic fermion, martingale property}

\begin{figure}
\centering{
\begin{minipage}[b]{0.7\textwidth}
\centering{\includegraphics[width=\textwidth]{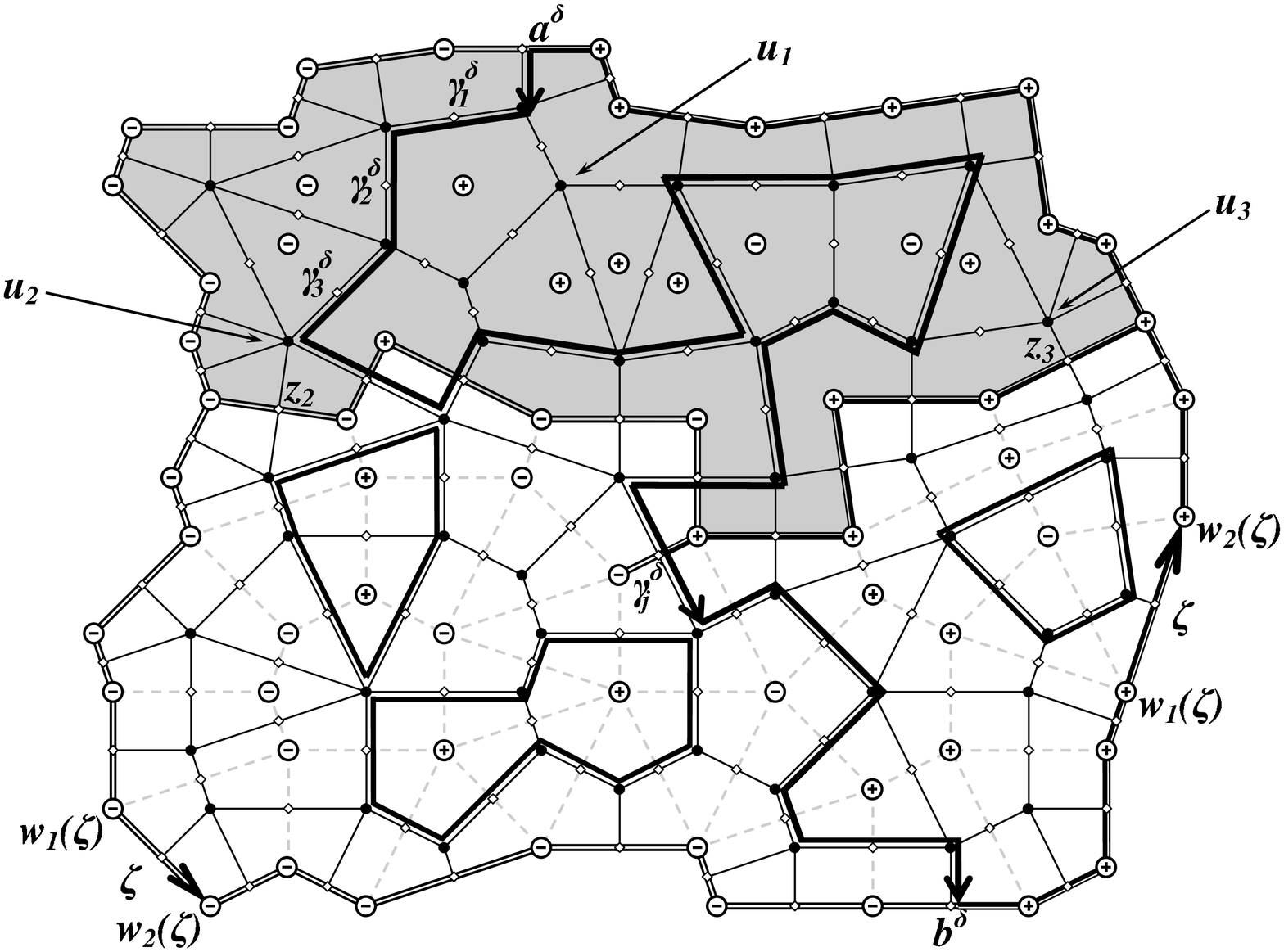}}

\centering{\textsc{(A)}}
\end{minipage}
\hskip 0.05\textwidth
\begin{minipage}[b]{0.2\textwidth}

\centering{\includegraphics[width=0.91\textwidth]{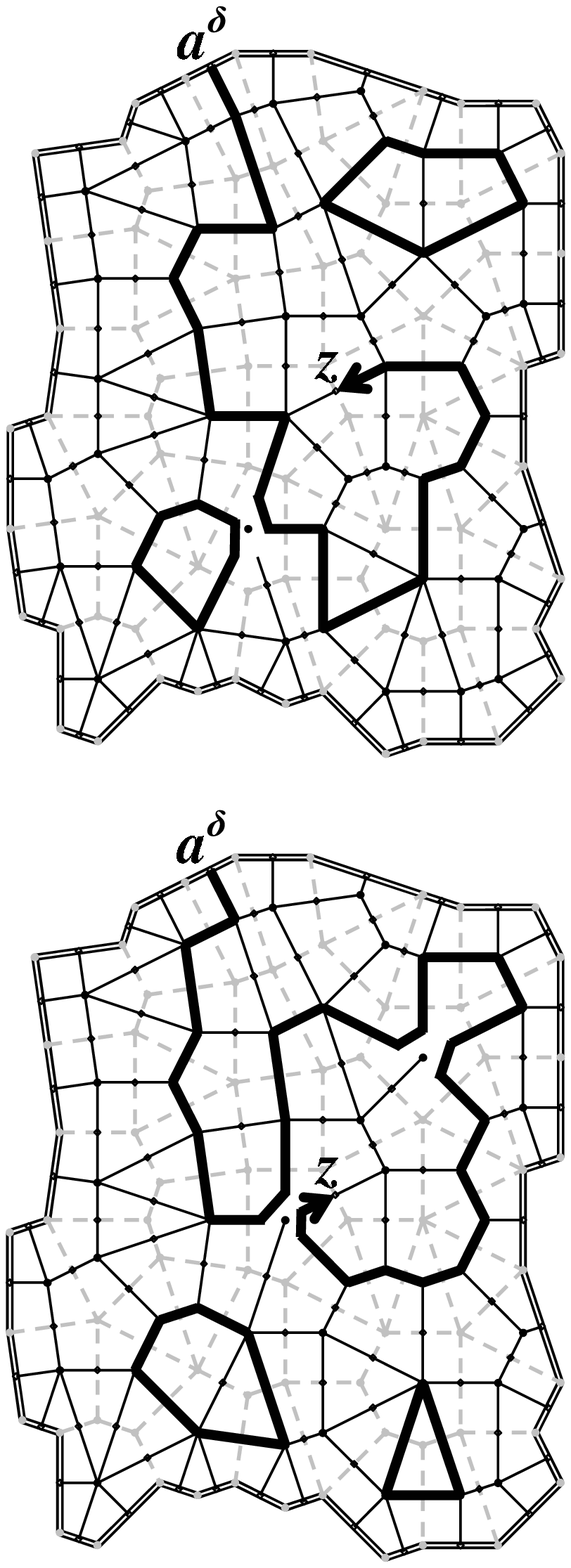}}

\bigskip

\centering{\textsc{(B)}}
\end{minipage}
} \caption{\label{Fig:SpinDomain} \textsc{(A)}~Ising model on isoradial graphs: discrete
domain $\O^\mesh_\DS$ and a sample configuration (the partition function is given by
(\ref{SpinZ=})). By our choice of the ``turning rule,'' loops and the interface $\g^\mesh$
separate clusters of ``$+$'' spins connected through edges and clusters of ``$-$'' spins
connected through vertices. To illustrate this, $\g^\mesh$ and loops are drawn slightly closer
to ``$+$'' spins. The component of $\O^\mesh$ not ``swallowed'' by the path
$[a^\mesh\g^\mesh_1..\g^\mesh_j]$ is unshaded. The vertex $u_1$ is shaded since it is not
connected to $b^\mesh$ anymore. The vertices $u_2$ and $u_3$ are shaded too since each of them
is connected to the bulk by a single edge which contradicts our definition of connected
discrete domains. Moreover, since the ``interface'' $a^\mesh\rsa\g_j^\mesh\rsa z$ could arrive
at each of these points only in a single way, our observable certainly satisfies boundary
condition (\ref{spin-BC}) at $z_2$ and $z_3$, thus not distinguishing them from the other
boundary~points. \textsc{(B)}~Two~samples of ``interface pictures'' composed from a number of
loops and a single interface $\g^\mesh$ running from $a^\mesh$ to $z$. To define the
$\wind(\g^\mesh;\,a^\mesh\rsa z)$ unambiguously, we draw $\g^\mesh$ so that, if there is a
choice, it turns to the left (for $z\in\pa\O^\mesh_\DS$, this corresponds to the
edge-connectivity of ``$+$'' clusters).}
\end{figure}

Let \mbox{$\O^\mesh_\DS\!\ss\!\DS$} be a simply-connected discrete domain composed of inner
rhombi \mbox{$z\!\in\!\Int\O^\mesh_\DS$} and boundary half-rhombi $\z\in\pa\O^\mesh_\DS$, with
two marked boundary points $a^\mesh,b^\mesh\in\pa\O^\mesh_\DS$, such that $\pa\O^\mesh_\DS$
contains only ``white'' vertices (see Fig.~\ref{Fig:SpinDomain}A) and there is no edge of
$\G^*$ breaking $\O^\mesh_\DS$ into two non-connected pieces.

To each inner ``white'' vertex $w\in\Int\O^\mesh_{\G^*}$, we assign the spin $\s(w)$ ($+$ or
$-$), thus obtaining the set of configurations (whose cardinality is
$2^{\mathrm{\#(\Int\O^\mesh_{\G^*})}}$). We also impose Dobrushin boundary conditions
assigning the $-$ spins to vertices on the boundary arc $(a^\mesh b^\mesh)$ and the $+$ spins
on the boundary arc $(b^\mesh a^\mesh)$ (see Fig.~\ref{Fig:SpinDomain}A). The {\bf partition
function} of the critical spin-Ising model is given by
\begin{equation}
\label{SpinZ=} \wt{Z}_{(\O^\mesh;\,a^\mesh\rsa b^\mesh)}=
[\sin\tfrac{1}{2}\theta(b^\mesh)]^{-1}\!\!\!\! \sum_{\mathrm{spin~config.}}~{\textstyle
\prod_{w_1\sim w_2: \s(w_1)\ne \s(w_2)} x_{w_1w_2},\quad x_{w_1w_2}=\tan
\frac{1}{2}\theta(z)},
\end{equation}
where $\theta(z)$ is the half-angle of the rhombus $u_1w_1u_2w_2$ having center at $z$
(i.e.,~$\tan\theta_{w_1w_2}\!=|w_2-w_1|/|u_2-u_1|$). The first factor $\sin^{-1}$ doesn't
depend on the configuration, and is introduced for technical reasons.

Due to Dobrushin boundary conditions, for each configuration, there is an interface $\g^\mesh$
running from $a^\mesh$ to $b^\mesh$ and separating $+$ spins from $-$ spins. If $\G$ is not a
trivalent graph, one needs to specify the algorithm of ``extracting $\g^\mesh$ from the
picture'', if it can be done in different ways. Below we assume that,
\begin{center}
\emph{if there is a choice, the ``interface'' takes the left-most possible route (see
Fig.~\ref{Fig:SpinDomain})}.
\end{center}
With this choice the interface separates clusters of ``$+$'' spins connected through edges and
clusters of ``$-$'' spins connected through vertices. Any other choice would do for the
martingale property and eventual conformal invariance, as discussed in the Introduction. For
example, one can toss a coin at each vertex to decide whether ``$+$'' or ``$-$'' spin clusters
connect through it. Note that, drawing all the edges separating $+$ spins from $-$ spins, one
can rewrite the partition function as a sum over all configurations $\varpi$ of edges which
consist of a single interface running from $a^\mesh$ to $b^\mesh$ and a number of loops.
Namely,
\begin{equation}
\label{SpinZ=1} \wt{Z}_{(\O^\mesh;\,a^\mesh\rsa b^\mesh)}=
[\sin\tfrac{1}{2}\theta(b^\mesh)]^{-1}\!\!\!\!\!\!
\sum_{\varpi=\{\mathrm{interface}+\mathrm{loops}\}}~{\textstyle \prod_{w_1\sim
w_2:[w_1;w_2]~\mathrm{intersects}~\o} x_{w_1w_2}}\,.
\end{equation}

For $z\in\O^\mesh_\DS$, the {\bf holomorphic fermion} is defined as
(cf.~(\ref{FKfermionXiDef}),~(\ref{WindBZ=}))
\begin{equation}
\label{SpinfermionDef} F^\mesh(z)=F^\mesh_{(\O^\mesh;\,a^\mesh,b^\mesh)}(z):=
\cF^\mesh(b^\mesh)\cdot \frac{\wt{Z}_{(\O^\mesh;\,a^\mesh\rsa z)}\cdot
\E_{(\O^\mesh;\,a^\mesh\rsa z)} e^{-\frac{i}{2}\wind(\g^\mesh;\,a^\mesh\rsa z)}}
{\wt{Z}_{(\O^\mesh;\,a^\mesh\rsa b^\mesh)}\cdot e^{-\frac{i}{2}\wind(a^\mesh\rsa b^\mesh)}}\,.
\end{equation}
As in (\ref{SpinZ=1}), $\wt{Z}_{(\O^\mesh;\,a^\mesh\rsa z^\mesh)}$ denotes the partition
function for the set of ``interfaces pictures'' containing (besides loops) one interface
$\g^\mesh$ running from $a^\mesh$ to $z$ (see Fig.~\ref{Fig:SpinDomain}B),
\begin{quotation}
\noindent \emph{for each configuration we count all weights corresponding to the drawn edges,
including $\tan\tfrac{1}{2}\theta(a^\mesh)$ and
$[\cos\frac{1}{2}\theta(z)]^{-1}=[\sin\frac{1}{2}\theta(z)]^{-1}\cdot\tan\frac{1}{2}\theta(z)$ for the first and the last.}
\end{quotation}
Then $\E_{(\O^\mesh;\,a^\mesh\rsa z^\mesh)}$ will stand for the expectation with respect to
the corresponding probability measure. Equivalently, one can take the partition function,
multiplying the weight of each configuration by the complex factor
$e^{-\frac{i}{2}\wind(\g^\mesh;\,a^\mesh\rsa z^\mesh)}$ (which, for $z\in\Int\O^\mesh_\DS$,
may be equal to one of the four different complex values $\a$, $i\a$, $-\a$, $-i\a$ depending on
the particular configuration). Finally,
\begin{quotation}
\emph{$\cF^\mesh(b^\mesh)\parallel (\t(b^\mesh))^{-1/2}$, where
$\t(b^\mesh)=w_2(b^\mesh)\!-\!w_1(b^\mesh)$, is a normalizing factor that depends only on
the structure of $\DS^\mesh$ near $b^\mesh$ (see Section~\ref{SectBoundHarnack}).}
\end{quotation}

\noindent Here and below, for $\z\in\pa\O^\mesh_\DS$, $\t(\z)=w_2(\z)\!-\!w_1(\z)$ denotes
``discrete tangent vector'' to $\pa\O^\mesh_\DS$ oriented counterclockwise (see
Fig.~\ref{Fig:SpinDomain}A for notation). For any $\z\in\pa\O^\mesh_\DS$, the
\mbox{$\wind(\g^\mesh;\,a^\mesh\rsa\z)$} is fixed due to topological reasons, and so
\begin{equation}
\label{spin-BC} F^\mesh(\z)\parallel (\t(\z))^{-\frac{1}{2}},\quad \z\in\pa
\O^\mesh_\DS\setminus\{a^\mesh\},\quad \mathit{while}\quad F^\mesh(a^\mesh)\parallel
i(\t(a^\mesh))^{-\frac{1}{2}}.
\end{equation}

\begin{remark}[\bf Martingale property]
For each $z\in\Int\O^\mesh_\DS$,
$F^\mesh_{(\O^\mesh\setminus[a^\mesh\g^\mesh_j];\g^\mesh_j,b^\mesh)}(z)$ is a martingale with
respect to the growing interface $(a^\mesh=\g^\mesh_0, \g^\mesh_1,
 ..., \g^\mesh_j, ...)$ (till the stopping time when $\g^\mesh$ hits
$\xi$ or $z$~becomes separated from $b^\mesh$ by the interface, see
Fig.~\ref{Fig:SpinDomain}A).
\end{remark}

\begin{proof}
It is sufficient to check that $F^\mesh$ has the martingale property when $\g^\mesh$ makes one
step. Let $a^\mesh_\mathrm{L},.., a^\mesh_\mathrm{R}$ denote all possibilities for the first
step. Then, %by definition of the model,
\begin{equation}
\label{Z=ZL+ZR} \wt{Z}_{(\O^\mesh;\,a^\mesh\rsa b^\mesh)} =
\tan\tfrac{1}{2}\theta(a^\mesh)\cdot\left[ \wt{Z}_{(\O^\mesh\setminus[a^\mesh
a^\mesh_\mathrm{L}];\,a^\mesh_\mathrm{L}\rsa b^\mesh)} + ... +
\wt{Z}_{(\O^\mesh\setminus[a^\mesh a^\mesh_\mathrm{R}];\,a^\mesh_\mathrm{R}\rsa
b^\mesh)}\right],
\end{equation}
and so
\[
{\P(\g^\mesh_1\!=\!a^\mesh_\mathrm{L})} = \frac{\wt{Z}_{(\O^\mesh\setminus[a^\mesh
a^\mesh_\mathrm{L}];\,a^\mesh_\mathrm{L}\rsa b^\mesh)}}
{[{\tan\tfrac{1}{2}\theta(a^\mesh)}]^{-1}\wt{Z}_{(\O^\mesh;\,a^\mesh\rsa b^\mesh)}}~,~\dots\
,~{\P(\g^\mesh_1\!=\!a^\mesh_\mathrm{R})}= \frac{\wt{Z}_{(\O^\mesh\setminus[a^\mesh
a^\mesh_\mathrm{R}];\,a^\mesh_\mathrm{R}\rsa b^\mesh)}}
{[{\tan\tfrac{1}{2}\theta(a^\mesh)}]^{-1}\wt{Z}_{(\O^\mesh;\,a^\mesh\rsa b^\mesh)}}\,.
\]
Taking into account that the difference
\[
\wind(a^\mesh\rsa b^\mesh) - \wind(a^\mesh_\mathrm{L}\rsa b^\mesh)= \wind(a^\mesh\to
a^\mesh_\mathrm{L}),
\] we obtain
\[
\P(\g^\mesh_1\!=\!a^\mesh_\mathrm{L})\cdot \frac{F^\mesh_{(\O^\mesh\setminus [a^\mesh
a^\mesh_\mathrm{L}];\,a^\mesh_\mathrm{L},b)}(z)}{F^\mesh_{(\O^\mesh;\,a^\mesh,b)}(z)}
 =
\frac{e^{-\frac{i}{2}\wind(a^\mesh\to a^\mesh_\mathrm{L})}}
{[\tan\tfrac{1}{2}\theta(a^\mesh)]^{-1}}\cdot \frac{\wt{Z}_{(\O^\mesh\setminus[a^\mesh
a^\mesh_\mathrm{L}];\,a^\mesh_\mathrm{L}\rsa z)}\cdot
\E %_{(\O^\mesh\setminus[a^\mesh a^\mesh_\mathrm{L}];\, a^\mesh_\mathrm{L}\rsa z)}
e^{-\frac{i}{2}\wind(\g^\mesh;\,a^\mesh_\mathrm{L}\rsa z)}}
{\wt{Z}_{(\O^\mesh;\,a^\mesh\rsa z)}\cdot \E%_{(\O^\mesh;\,a^\mesh\rsa z)}
e^{-\frac{i}{2}\wind(\g^\mesh;\,a^\mesh\rsa z)}}
\]
and so on. On the other hand, counting ``interface pictures'' depending on the first step as
in (\ref{Z=ZL+ZR}), we easily obtain
\[
\begin{array}{l}
[{\tan\tfrac{1}{2}\theta(a^\mesh)}]^{-1}\wt{Z}_{(\O^\mesh;\,a^\mesh\rsa z)}\cdot
\E%_{(\O^\mesh;\,a^\mesh\rsa z)}
e^{-\frac{i}{2}\wind(\g^\mesh;\,a^\mesh\rsa z)} \cr \qquad\quad =
e^{-\frac{i}{2}\wind(a^\mesh\to a^\mesh_\mathrm{L})}\cdot \wt{Z}_{(\O^\mesh\setminus[a^\mesh
a^\mesh_\mathrm{L}];\,a^\mesh_\mathrm{L}\rsa z)}\cdot
\E%_{(\O^\mesh\setminus[a^\mesh a^\mesh_\mathrm{L}];\, a^\mesh_\mathrm{L}\rsa z)}
e^{-\frac{i}{2}\wind(\g^\mesh;\,a^\mesh_\mathrm{L}\rsa z)}+\dots \vphantom{\big|^|_|}\cr
\qquad\quad +\,e^{-\frac{i}{2}\wind(a^\mesh\to a^\mesh_\mathrm{R})}\cdot
\wt{Z}_{(\O^\mesh\setminus[a^\mesh a^\mesh_\mathrm{R}];\,a^\mesh_\mathrm{R}\rsa z)}\cdot
\E%_{(\O^\mesh\setminus[a^\mesh a^\mesh_\mathrm{R}];\,a^\mesh_\mathrm{R}\rsa z)}
e^{-\frac{i}{2}\wind(\g^\mesh;\,a^\mesh_\mathrm{R}\rsa z)},
\end{array}
\]
which gives the result.
\end{proof}

\subsubsection{Discrete boundary value problem for $F^\mesh$.}

\begin{proposition}
\label{PropSpinFermionHol} For each pair of neighbors \mbox{$z_0, z_1\in\O^\mesh_\DS$}
separated by an edge $(w_1u)$,
we have
\begin{equation}
\label{SpinSHol} \Pr\left[F^\mesh(z_0)\,;[i(w_1\!-\!u)]^{-\frac{1}{2}}\right]=
\Pr\left[F^\mesh(z_1)\,;[i(w_1\!-\!u)]^{-\frac{1}{2}}\right].
\end{equation}
\end{proposition}
\noindent The proposition amounts to saying that $F^\mesh$ is \emph{spin holomorphic} as in
Definition~\ref{DefSHol} below (see Fig.~\ref{Fig:Notations}C for notation).

\begin{figure}
\noindent\begin{tabular}{ccccccc}
\includegraphics[width=0.16\textwidth]{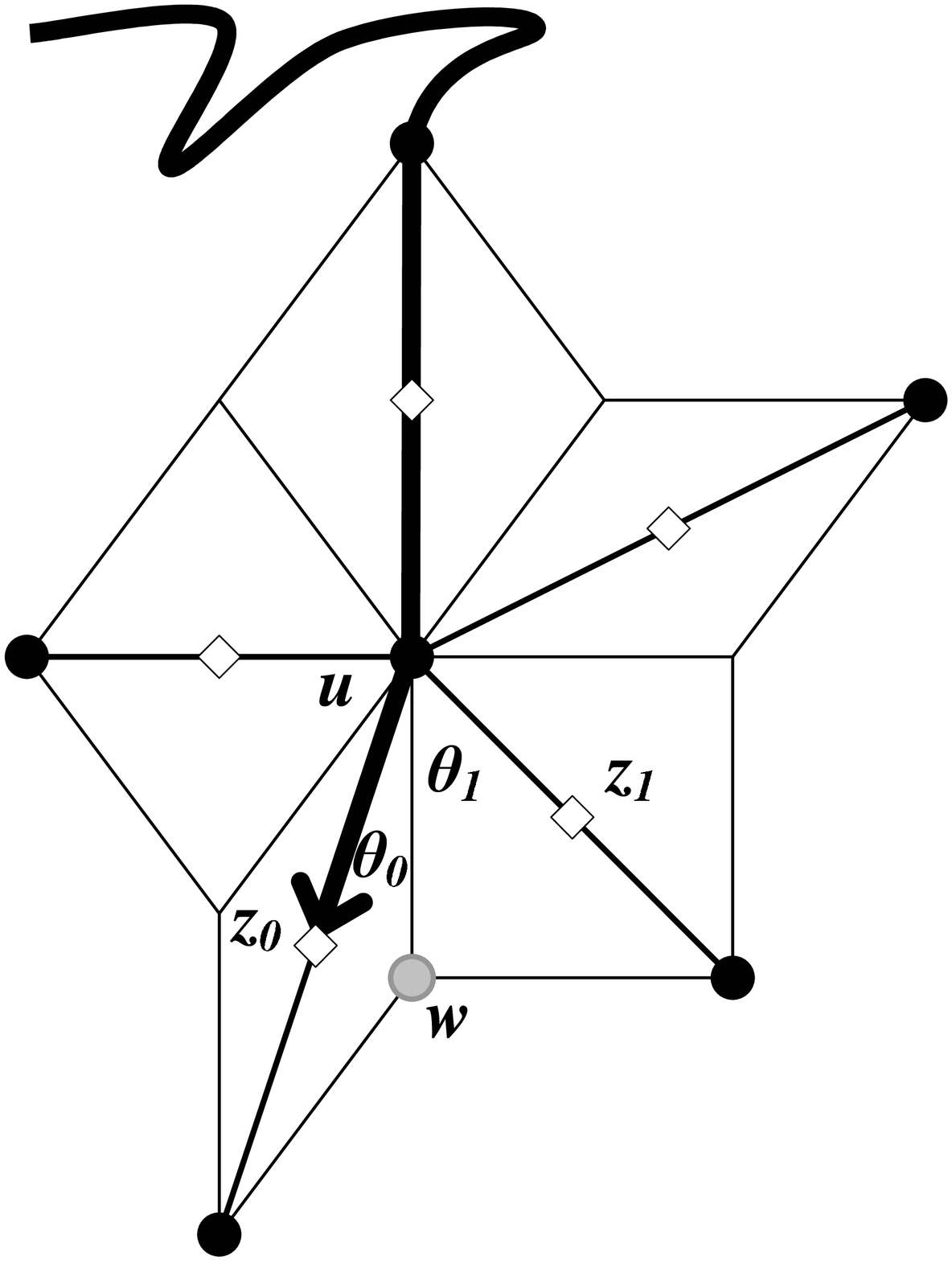}
&\begin{minipage}[b]{32pt} \noindent $\begin{array}{c}\textsc{I} \cr\displaystyle
\Longleftrightarrow \end{array}$
\[
\]
\end{minipage}&
\includegraphics[width=0.16\textwidth]{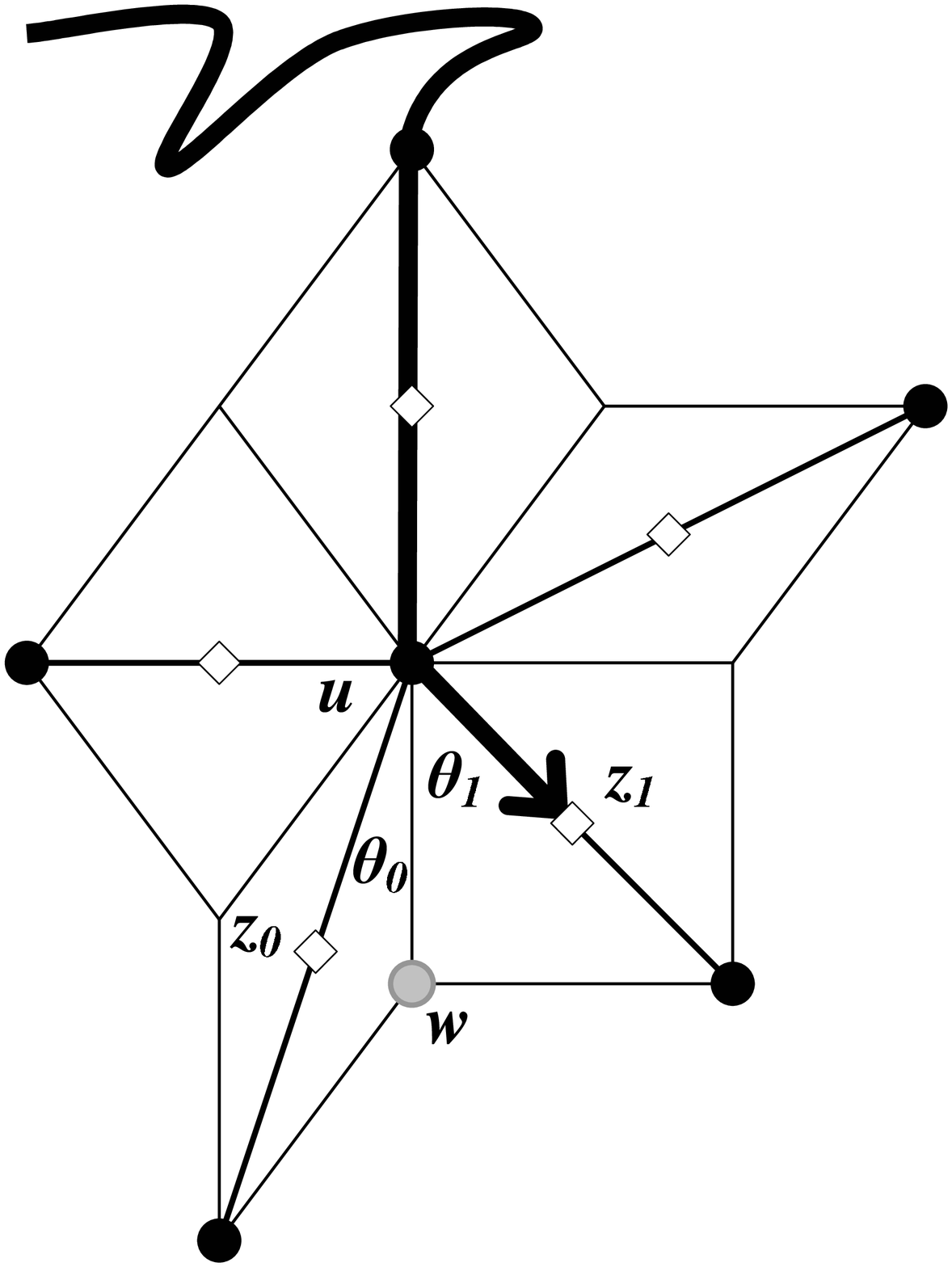} &\qquad\qquad&
\includegraphics[width=0.16\textwidth]{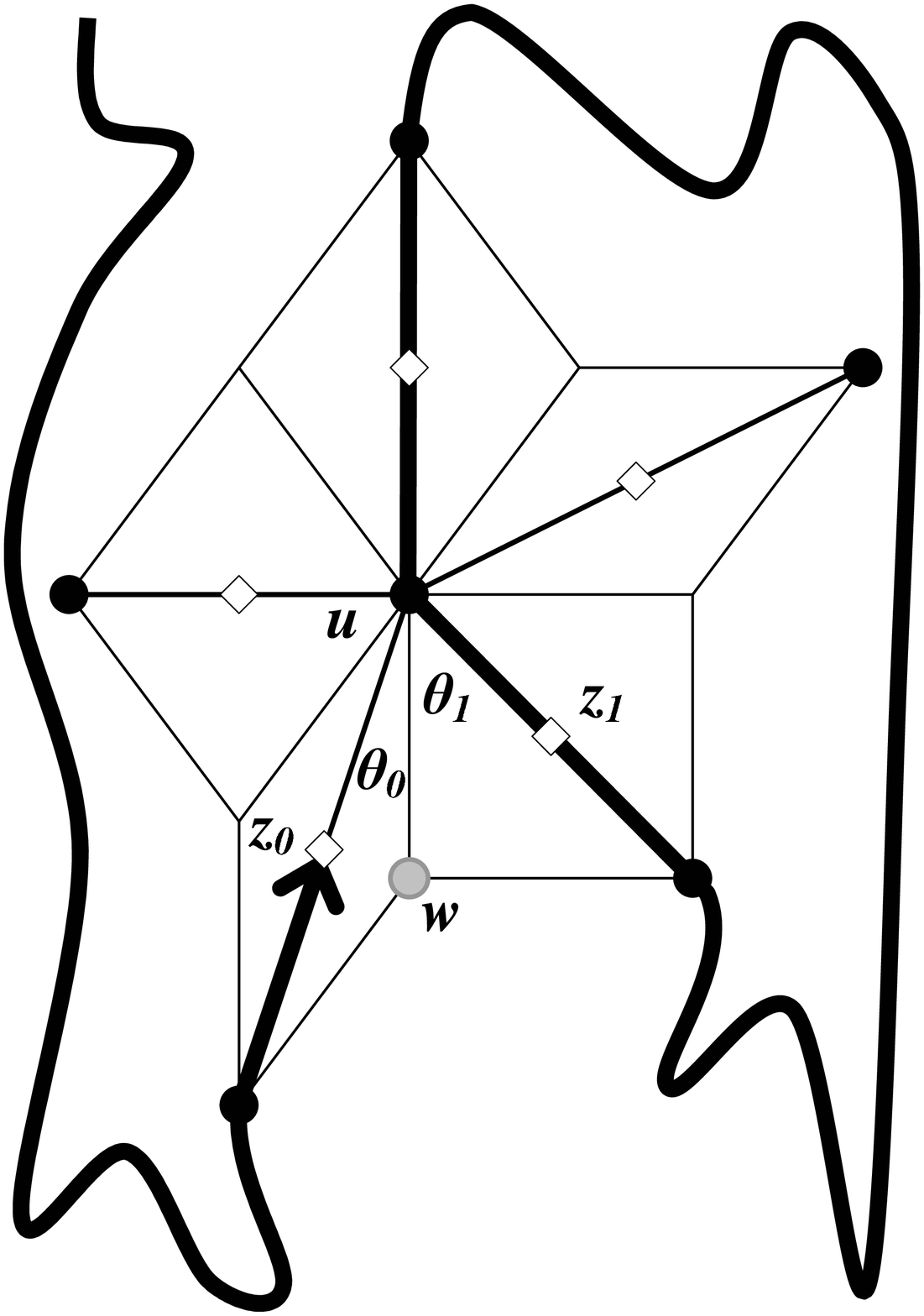}
&\begin{minipage}[b]{32pt} \noindent $\begin{array}{c}\textsc{IVa} \cr \displaystyle
\Longleftrightarrow \end{array}$
\[
\]
\end{minipage}&
\includegraphics[width=0.16\textwidth]{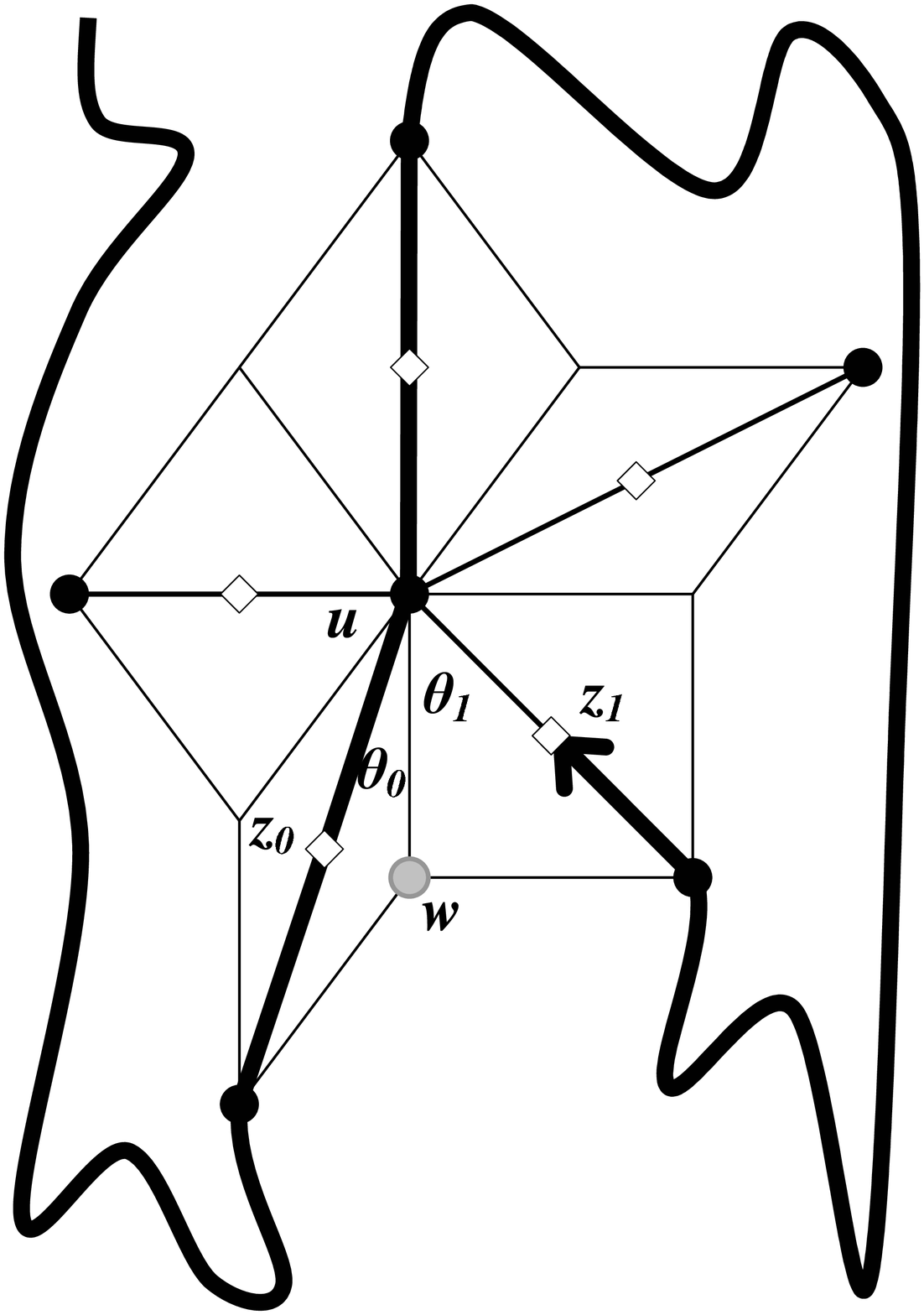} \cr\cr
\includegraphics[width=0.16\textwidth]{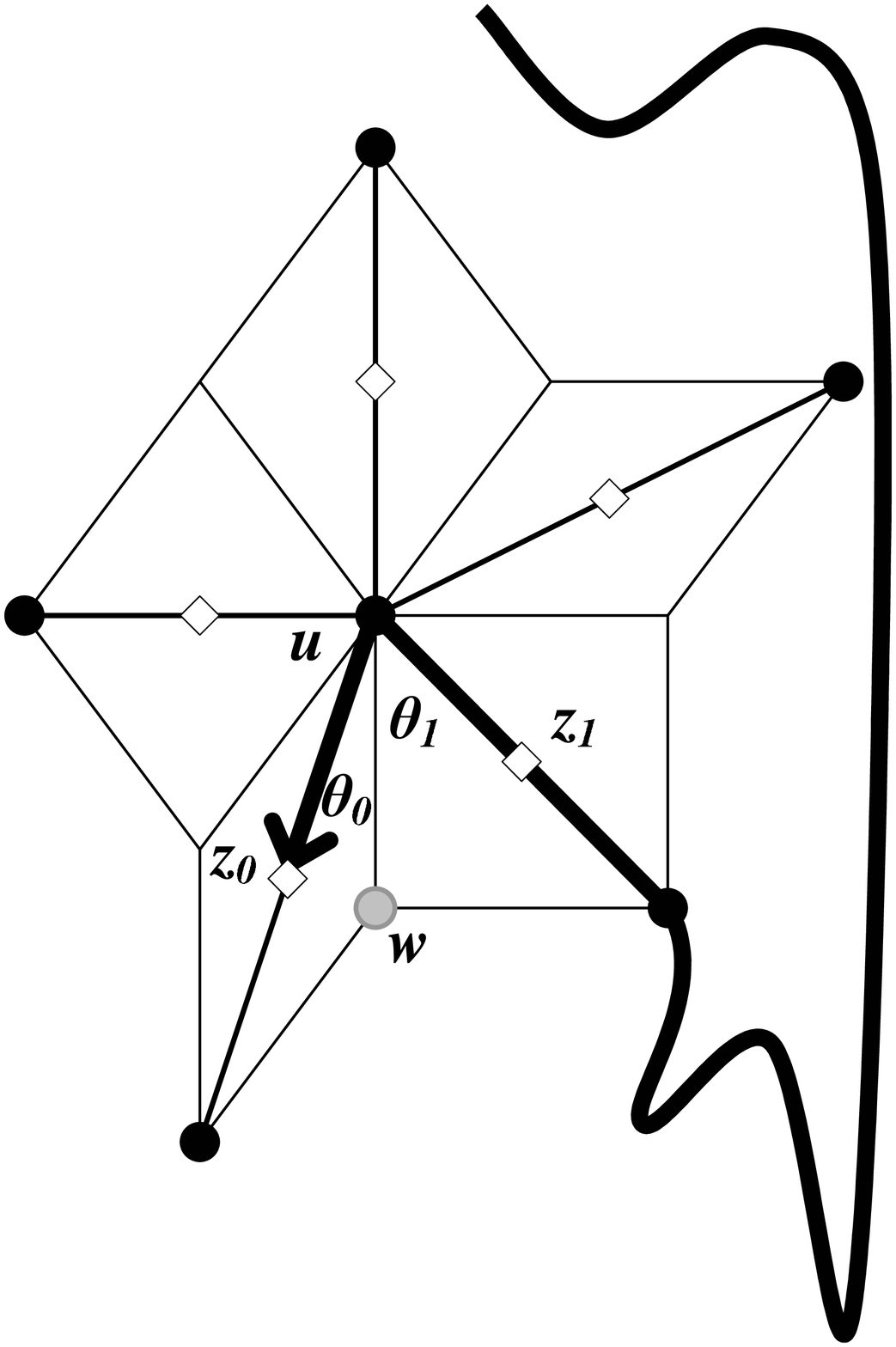}
&\begin{minipage}[b]{32pt} \noindent $\begin{array}{c}\textsc{II} \cr\displaystyle
\Longleftrightarrow \end{array}$
\[
\]
\end{minipage}&\includegraphics[width=0.16\textwidth]{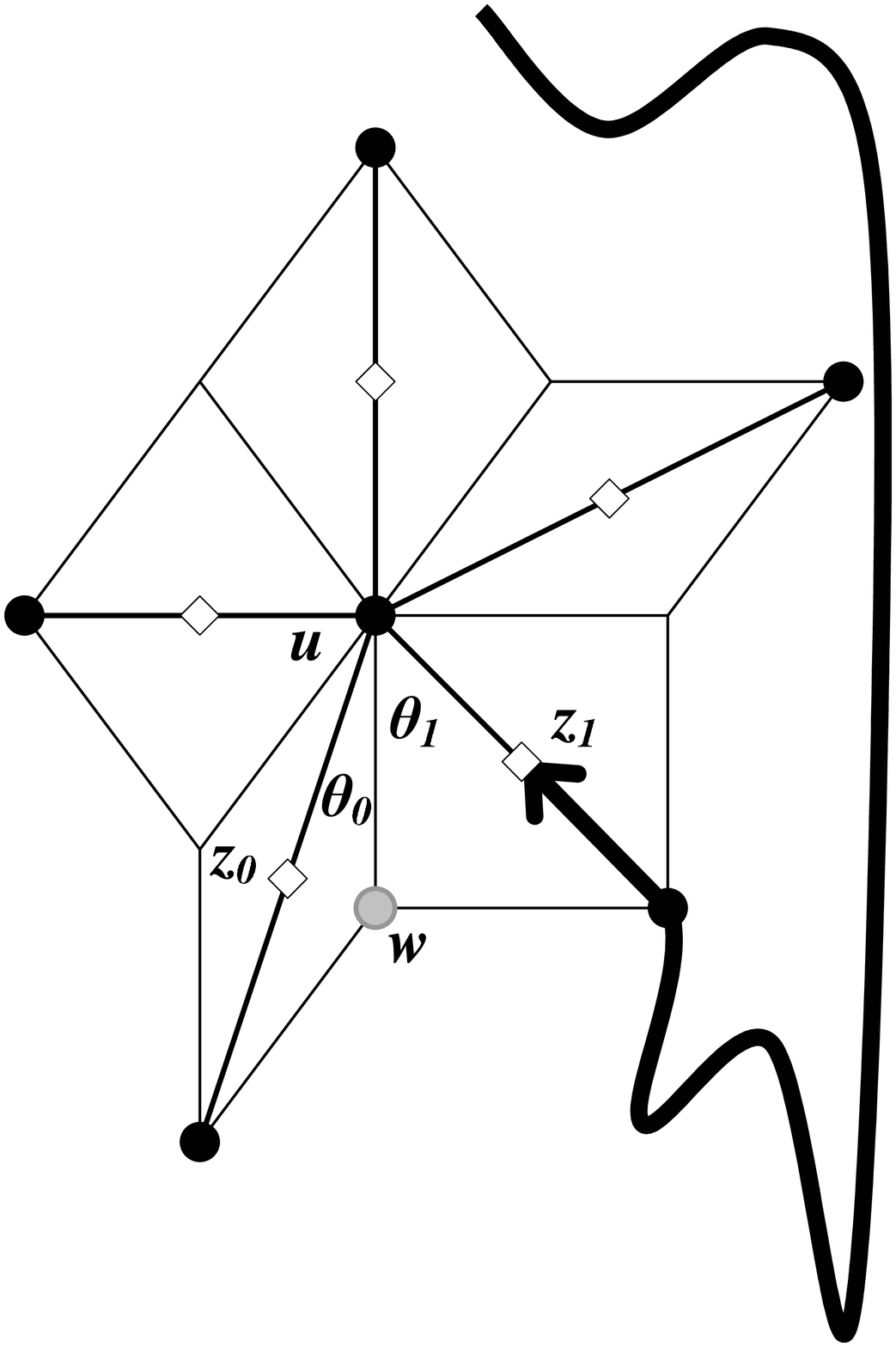} &&
\includegraphics[width=0.16\textwidth]{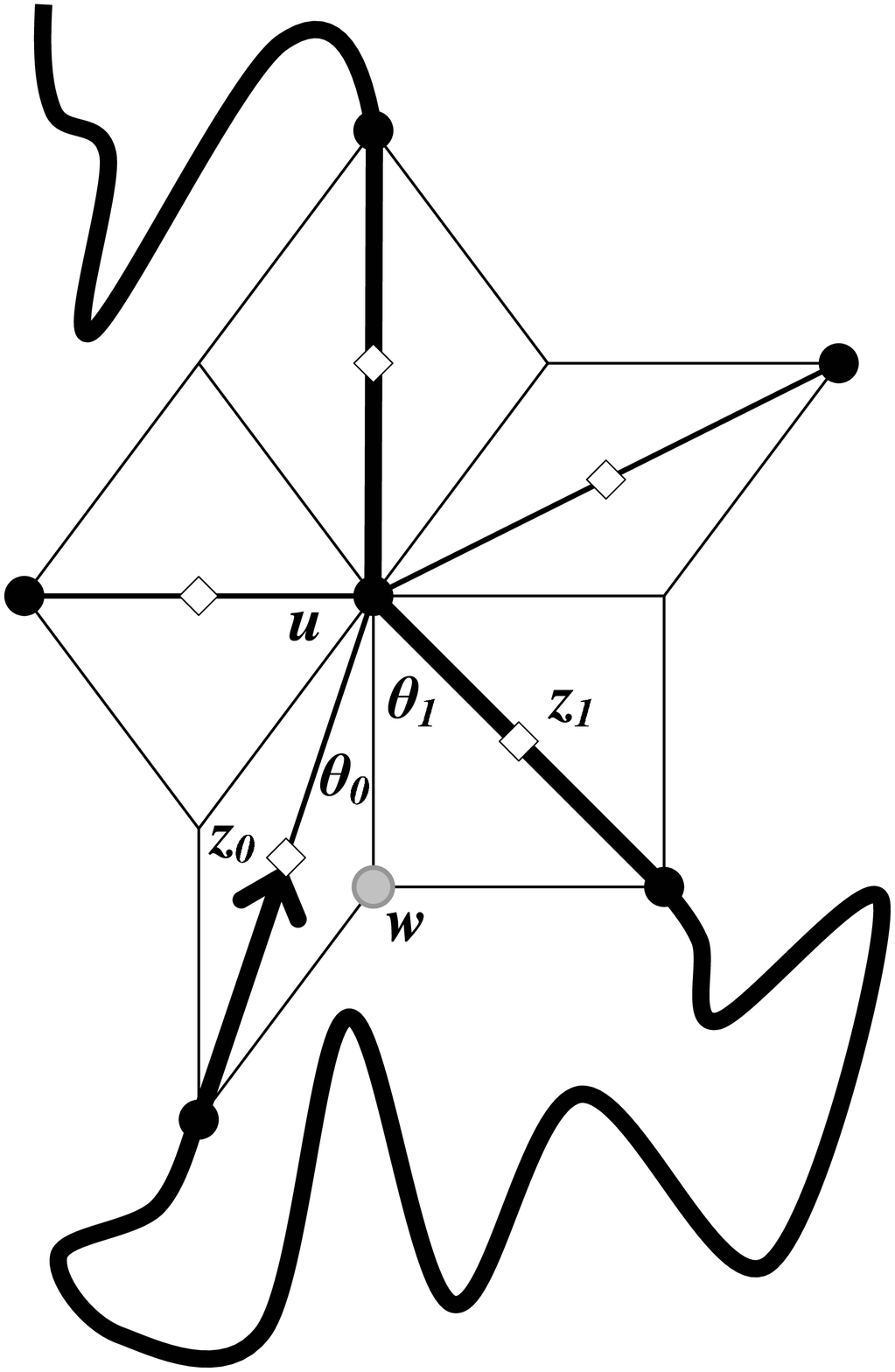}
&\begin{minipage}[b]{32pt} \noindent $\begin{array}{c}\textsc{IVb} \cr\displaystyle
\Longleftrightarrow \end{array}$
\[
\]
\end{minipage}&\includegraphics[width=0.16\textwidth]{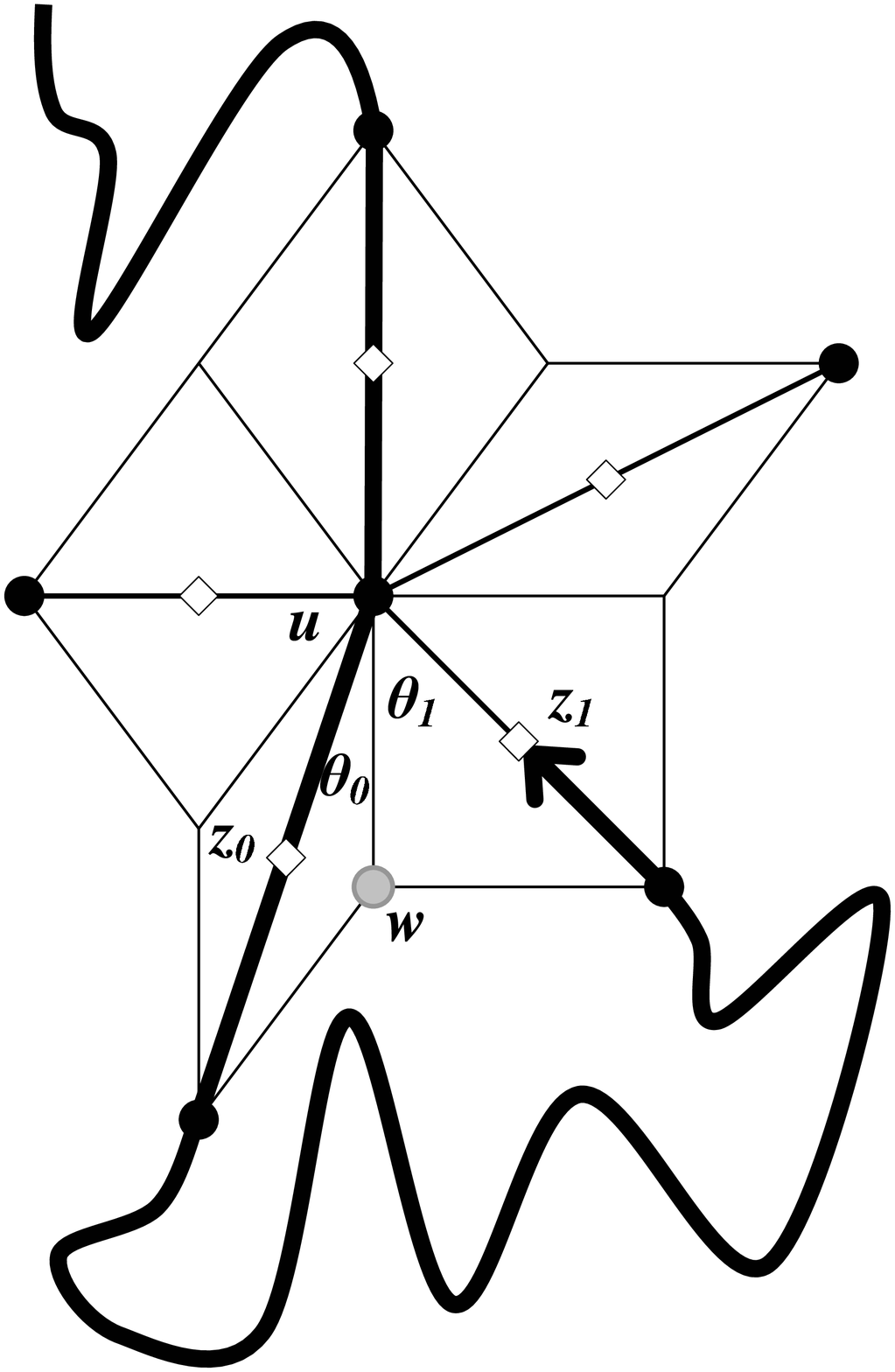} \cr\cr
\includegraphics[width=0.16\textwidth]{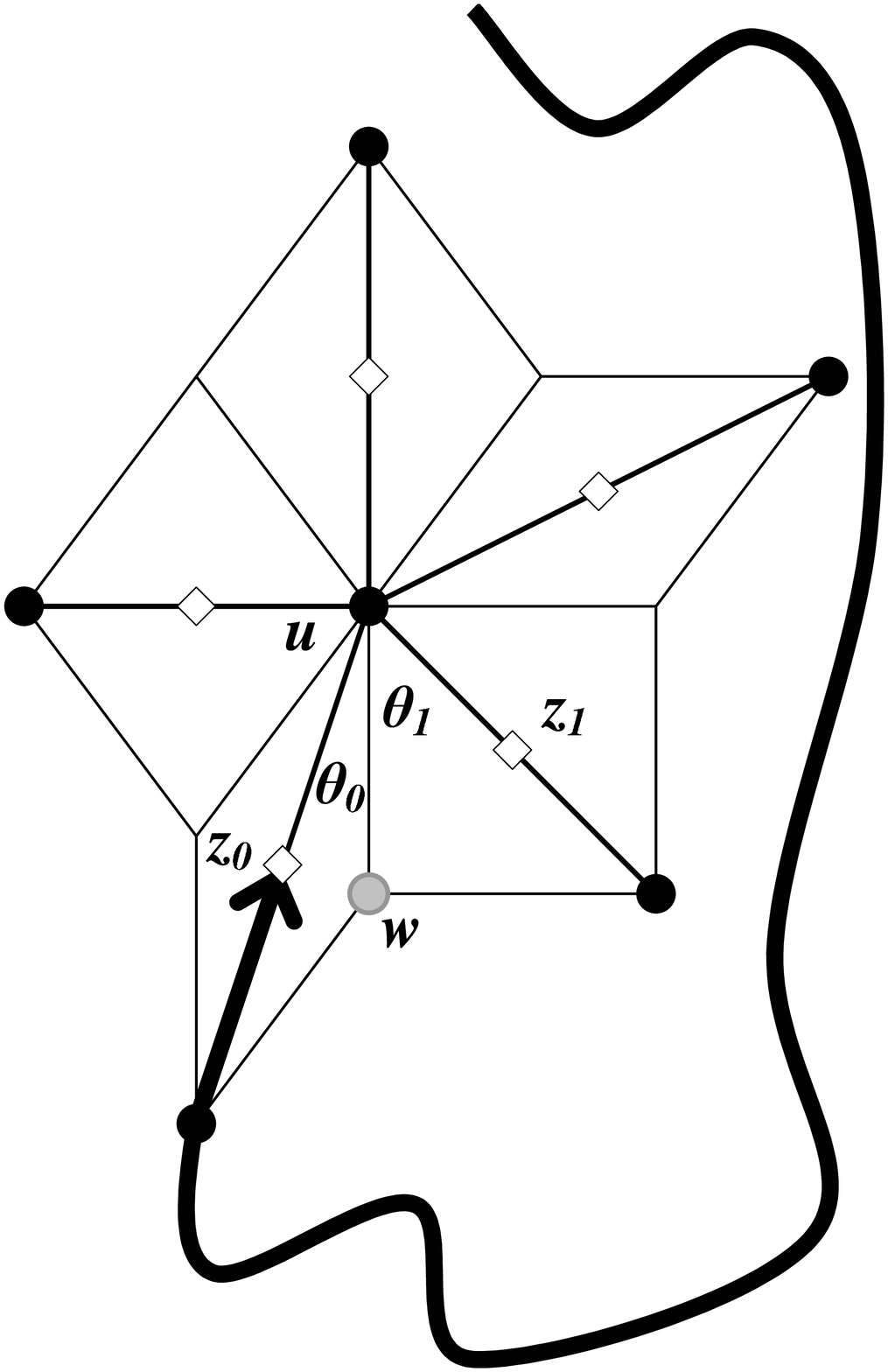}
&\begin{minipage}[b]{32pt} \noindent $\begin{array}{c}\textsc{III} \cr\displaystyle
\Longleftrightarrow \end{array}$
\[
\]
\end{minipage}&\includegraphics[width=0.16\textwidth]{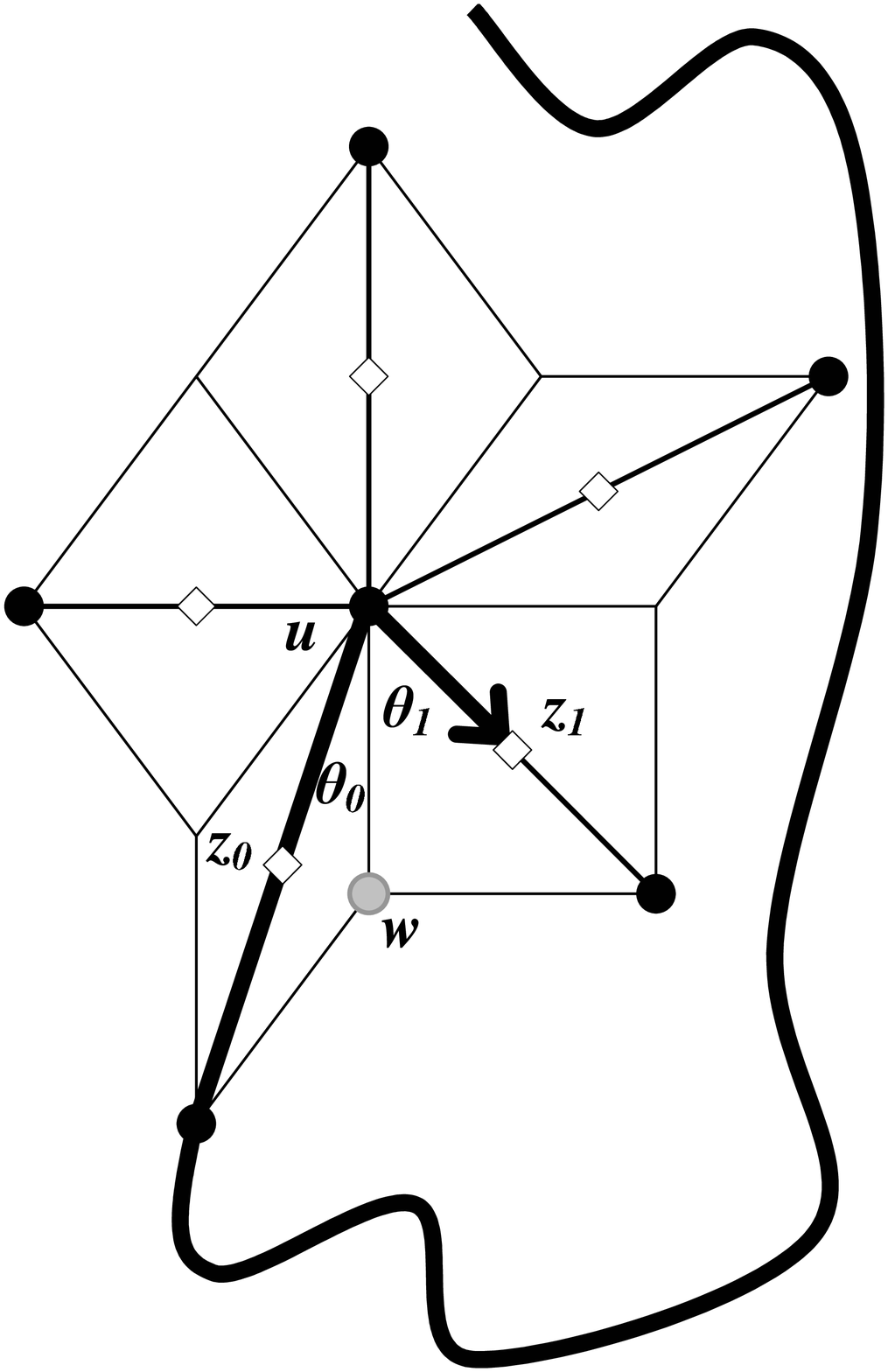} &&
\includegraphics[width=0.16\textwidth]{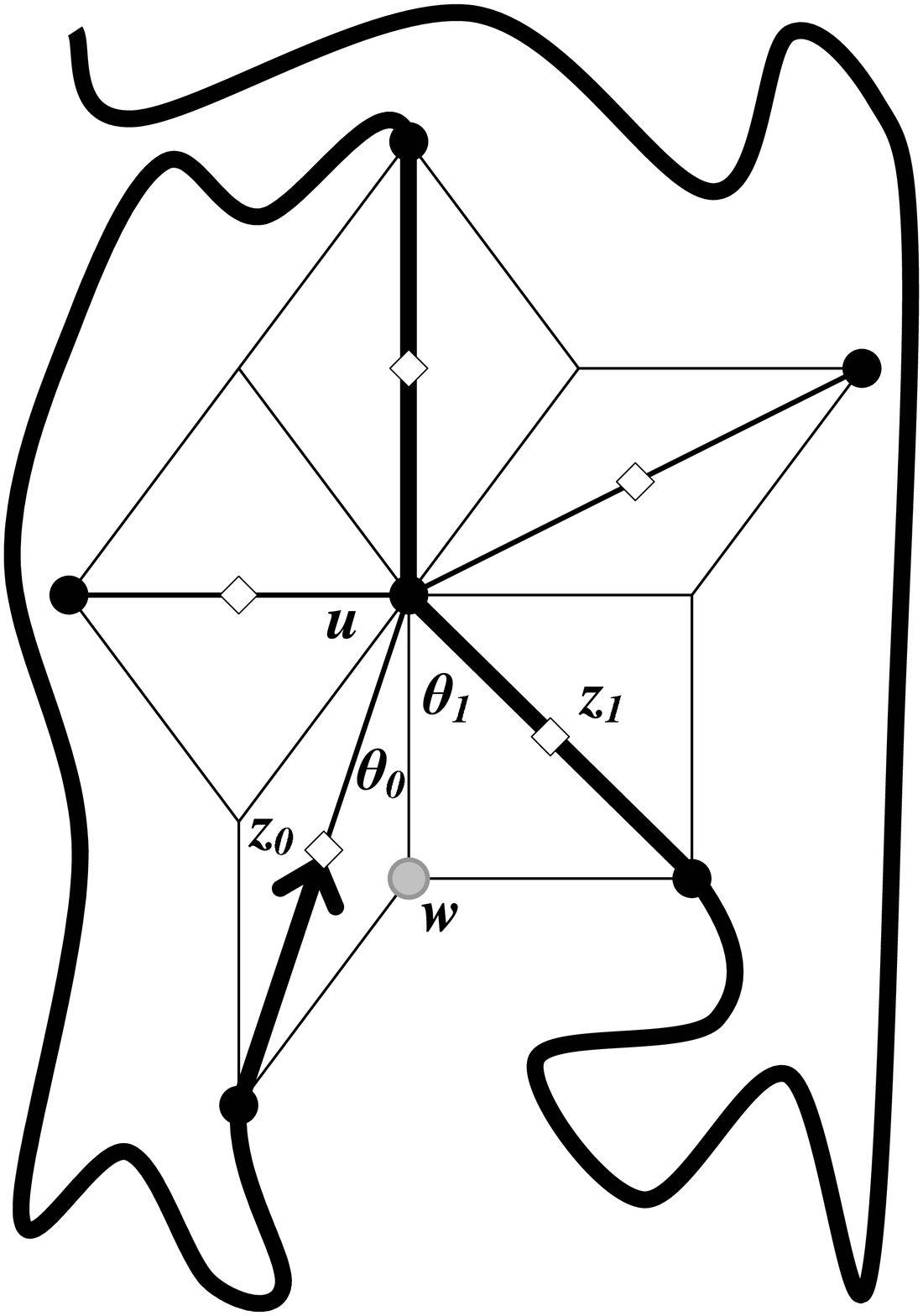}
&\begin{minipage}[b]{32pt} \noindent $\begin{array}{c}\textsc{IVc} \cr\displaystyle
\Longleftrightarrow \end{array}$
\[
\]
\end{minipage}&\includegraphics[width=0.16\textwidth]{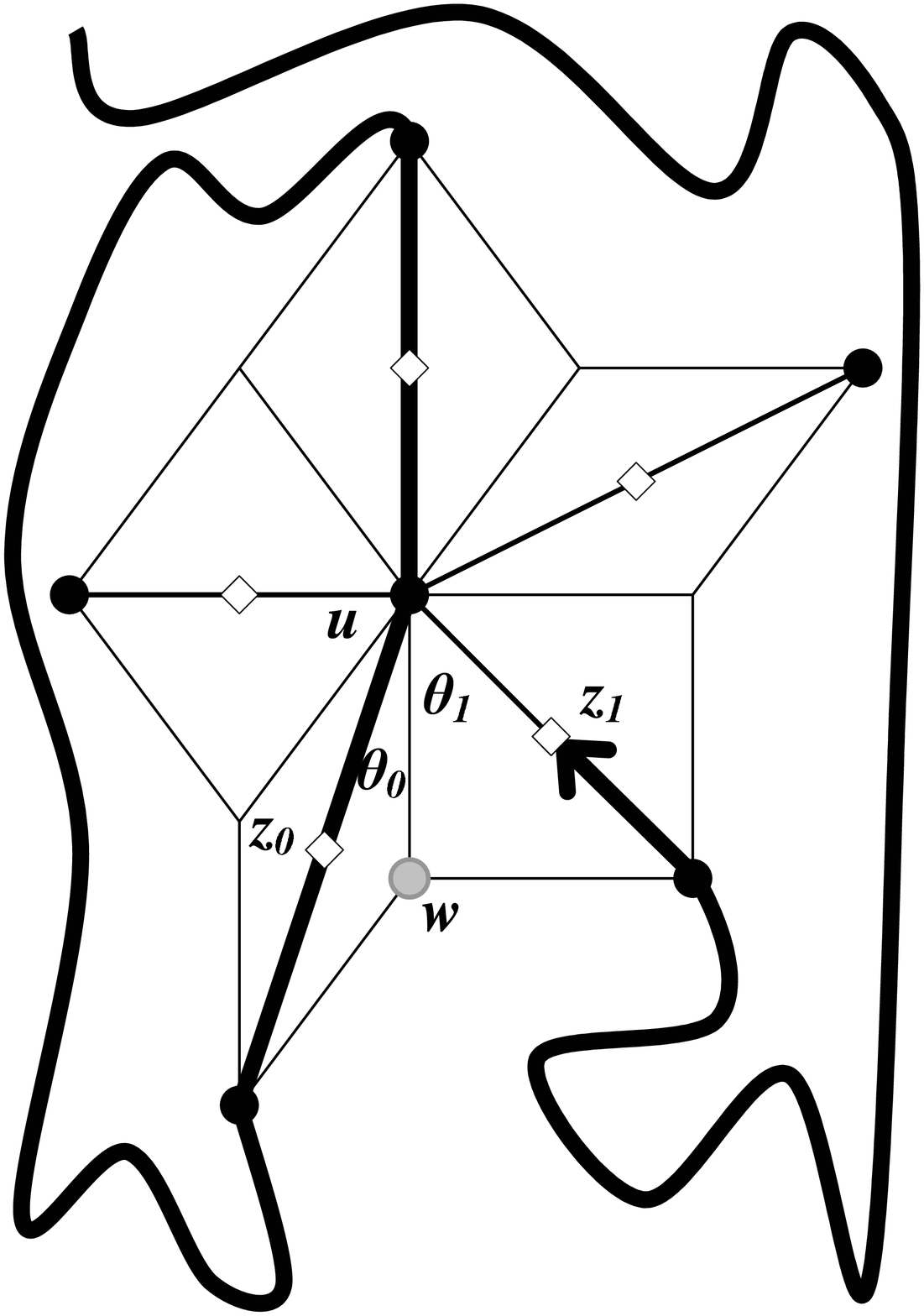} \cr
\end{tabular}

\caption{\label{Fig:SpinBijection} Bijection between the sets \mbox{$(\O^\mesh;\,a^\mesh\rsa
z_0)$} and \mbox{$(\O^\mesh;\,a^\mesh\rsa z_1)$} (local notations $z_{0,1}$, $\theta_{0,1}$,
$u$ and $w_1$ are given in Fig.~\ref{Fig:Notations}C). In cases \mbox{I-III}, the winding
$\vp_1=\wind(\g^\mesh_1;\,a^\mesh\rsa z_1)$ is unambiguously defined by
$\vp_0=\wind(\g^\mesh_0;\,a^\mesh\rsa z_0)$. In case IV, there are two possibilities: $\vp_1$
is equal to either $\vp_0-2\pi+\theta_0+\theta_1$ (IVa,~IVc) or $\vp_0+2\pi+\theta_0+\theta_1$
(IVb). }
\end{figure}

\begin{proof}
The proof is based on the bijection between the set of ``interface pictures''
\mbox{$(\O^\mesh;\,a^\mesh\rsa z_0)$} and the similar set $(\O^\mesh;\,a^\mesh\rsa z_1)$,
which is schematically drawn in Fig.~\ref{Fig:SpinBijection}. The relative contributions of
the corresponding pairs to $F^\mesh(z_0)$ and $F^\mesh(z_1)$ (up to the same real factor) are given in the following table:
\begin{center}\begin{tabular}{||c||c|c||} \hline\hline
$\vphantom{\big|^|_|}$ & $F^\mesh(z_0)$ & $F^\mesh(z_1)$ \cr \hline\hline
$\vphantom{\big|^|_|}\mathrm{I}$ & $[\cos\tfrac{1}{2}\theta_0]^{-1}\cdot e^{-\tfrac{i}{2}\vp}$
& $[\cos\tfrac{1}{2}\theta_1]^{-1}\cdot e^{-\frac{i}{2}(\vp+\theta_0+\theta_1)}$ \cr \hline
$\vphantom{\big|^|_|}\mathrm{II}$ & $[\cos\tfrac{1}{2}\theta_0]^{-1}\tan\tfrac{1}{2}\theta_1
\cdot e^{-\tfrac{i}{2}\vp}$ & $[\cos\tfrac{1}{2}\theta_1]^{-1}\cdot
e^{-\frac{i}{2}(\vp-\pi+\theta_0+\theta_1)}$ \cr \hline $\vphantom{\big|^|_|}\mathrm{III}$ &
$[\cos\tfrac{1}{2}\theta_0]^{-1}\cdot e^{-\tfrac{i}{2}\vp}$ &
$[\cos\tfrac{1}{2}\theta_1]^{-1}\tan\tfrac{1}{2}\theta_0\cdot
e^{-\frac{i}{2}(\vp-\pi+\theta_0+\theta_1)}$ \cr \hline $\vphantom{\big|^|_|}\mathrm{IV}$ &
$[\cos\tfrac{1}{2}\theta_0]^{-1}\tan\tfrac{1}{2}\theta_1 \cdot e^{-\tfrac{i}{2}\vp}$ &
$[\cos\tfrac{1}{2}\theta_1]^{-1}\tan\tfrac{1}{2}\theta_0\cdot e^{-\frac{i}{2}(\vp\pm 2\pi +
\theta_0+\theta_1)}$ \cr \hline\hline
\end{tabular}
\end{center}
where $\vp=\wind(\g^\mesh;\,a^\mesh\rsa z_0)-\wind(a^\mesh\rsa b^\mesh)+\arg\t(b^\mesh)$. Note
that
\[
\begin{array}{lll}
e^{-\frac{i}{2}\vp}\parallel [w_1\!-\!w_0]^{-\frac{1}{2}}, &
e^{-\frac{i}{2}(\vp+\theta_0+\theta_1)}\parallel [w_2\!-\!w_1]^{-\frac{1}{2}}, &\mathrm{in~cases\
I~\&~II,} \vphantom{\big|_|}\cr e^{-\frac{i}{2}\vp}\parallel [w_0\!-\!w_1]^{-\frac{1}{2}}, &
e^{-\frac{i}{2}(\vp+\theta_0+\theta_1)}\parallel [w_1\!-\!w_2]^{-\frac{1}{2}}, & \mathrm{in~cases\
III~\&~IV.} \vphantom{\big|^|}
\end{array}
\]
A simple trigonometric calculation then shows that
the (relative) contributions to \emph{both} projections
\mbox{$\Pr[F^\mesh(z_{j});[i(w\!-\!u)]^{-\frac{1}{2}}]$} for $j=0,\,1$
are equal to $1$,
$\tan\tfrac{1}{2}\theta_1$, $\tan\tfrac{1}{2}\theta_0$ and
$\tan\tfrac{1}{2}\theta_0\tan\tfrac{1}{2}\theta_1$ in cases I--IV, respectively.
\end{proof}

Summing it up, we arrive at

\medskip

\noindent {\bf Discrete Riemann boundary value problem for $\bm{F^\mesh}$ (spin-case):} \emph{
The function $F^\mesh$ is defined in $\O^\mesh_\DS$ so that
discrete holomorphicity (\ref{SpinSHol}) holds for every pair of neighbors \mbox{$z_0,z_1\in\O^\mesh_\DS$}. Furthermore, $F^\mesh$ satisfies the
boundary conditions (\ref{spin-BC}) and is normalized at $b^\mesh$.}

\section{S-holomorphic functions on isoradial graphs}
\label{SectSHolFunct} \setcounter{equation}{0}

\subsection{Preliminaries. Discrete harmonic and discrete holomorphic functions on isoradial
graphs} We start with basic definitions of the discrete complex analysis on isoradial graphs,
more details can be found in Appendix and our paper \cite{chelkak-smirnov-dca},
where a ``toolbox'' of
discrete versions of continuous results is provided.

Let $\G$ be an isoradial graph, and $H$ be defined on some vertices of $\G$. We define its
discrete Laplacian whenever possible by
\begin{equation}
\label{DMeshDef} [\D^\mesh H](u)~:=~\frac{1}{\weightG{u}}\sum_{u_s\sim u}
\tan\theta_s\cdot[H(u_s)\!-\!H(u)],
\end{equation}
where $\weightG{u}=\frac{1}{2}\mesh^2\sum_{u_s\sim u}\sin 2\theta_s$ (see
Fig.~\ref{Fig:Notations}B for notation). Function $H$ is called \emph{(discrete) harmonic} in
some discrete domain $\O^\mesh_\G$ if $\D^\mesh H =0$ at all interior vertices of
$\O^\mesh_\G$. It is worthwhile to point out that, on isoradial graphs, as in the continuous
setup, harmonic functions satisfy some (uniform w.r.t. $\mesh$ and the structure of $\DS$)
variant of the Harnack's Lemma (see Proposition~\ref{PropHarnack}).

Let $\dhm{\mesh}{u}{E}{\O^\mesh_\G}$ denote the \emph{harmonic measure} of $E\ss
\pa\O^\mesh_\G$ viewed from $u\in \Int \O^\mesh_\G$, i.e., the probability that
the random walk
generated by (\ref{DMeshDef}) (i.e. such that transition probabilities
at $u$ are proportional to $\tan\theta_s$'s) on $\G$ started from $u$ exits $\O^\mesh_\G$ through $E$.
As usual, $\o^\mesh$~is a probability measure on $\pa{\O^\mesh_\G}$ and a harmonic function of $u$. If
$\O^\mesh_\G$ is bounded, then we have
\[
H(u)=\sum_{a\in\pa\O^\mesh_\G} \dhm \mesh u {\{a\}} {\O^\mesh_\G} \cdot H(a),\quad
u\in\Int\O^\mesh_\G,
\]
for any discrete harmonic in $\O^\mesh_\G$ function $H$. Below we use some uniform (w.r.t.
$\mesh$ and the structure of $\G$) estimates of $\o^\mesh$
from \cite{chelkak-smirnov-dca}, which are quoted
in Appendix.

Let $H$ be defined on some part of $\G$ or $\G^*$ or $\L=\G\cup\G^*$ and $z$ be the center of
the rhombus $v_1v_2v_3v_4$. We set
\[
[\dpa H](z)~:=~\frac{1}{2}\lt[\frac{H(v_1)\!-\!H(v_3)}{v_1\!-\!v_3} +
\frac{H(v_2)\!-\!H(v_4)}{v_2\!-\!v_4}\rt],\qquad z\in\DS.
\]
Furthermore, let $F$ be defined on some subset of $\DS$.
We define its discrete $\dopa$-derivative by setting
\begin{equation}
\label{DopaDef} [\dopa F](u)= -\frac{i}{2\weightG{u}}\sum_{z_s\sim u}
(w_{s+1}\!-\!w_s)F(z_s),\qquad u\in\L=\G\cup\G^*
\end{equation}
(see Fig.~\ref{Fig:Notations}B for notation when $u\in\G$). Function $F$ is called
\emph{(discrete) holomorphic} in some discrete domain $\O^\mesh_\DS\ss\DS$ if $\dopa F =0$ at
all interior vertices. It is easy to check that $\D^\mesh=4\dopa\dpa$, and so $\dpa H$ is
holomorphic for any harmonic function $H$. Conversely, in simply connected domains, if $F$ is
holomorphic on $\DS$, then there exists a harmonic function $H=\int^\mesh F(z)d^\mesh z$ such
that $\dpa H =F$. Its components $H|_\G$ and $H|_{\G^*}$ are defined uniquely up to additive
constants by
\[
H(v_2)-H(v_1)=F\left({\textstyle\frac{1}{2}}(v_2\!+\!v_1)\right)\cdot (v_2-v_1),\quad v_2\sim
v_1,
\]
where both $v_1,v_2\in\G$ or both $v_1,v_2\in\G^*$, respectively.

It is most important that discrete holomorphic (on $\DS$) functions are
Lipschitz continuous in an appropriate sense, see Corollary \ref{LipHol}.
For the sake of the reader, we quote all
other necessary results in Appendix.

\subsection{S-holomorphic functions and the propagation equation for spinors}
\label{SectSHol}

In this section we investigate the notion of s-holomorphicity
which appears naturally for holomorphic fermions in the Ising model (see (\ref{FKSHol}),
(\ref{SpinSHol})).
We discuss its connections to  spinors
defined on the double-covering of $\DS$ edges (see \cite{mercat-2001}) and essentially equivalent to the introduction of disorder operators (see \cite{kadanoff-ceva,rajabpour-cardy} and the references therein).
We don't refer to this discussion (except Definition \ref{DefSHol} and elementary Lemma
\ref{SHolIsHol}) in the rest of our paper.

\begin{definition}
\label{DefSHol} Let $\O^\mesh_\DS\ss\DS$ be some discrete domain and
$F:\O^\mesh_\DS\to\C$. We
call function $F$ \mbox{{\bf strongly}} or \mbox{{\bf spin holomorphic}}, or \mbox{{\bf s-holomorphic}} for short,
if for each pair of neighbors
$z_0,z_1\in\O^\mesh_\DS$, $z_0\sim z_1$, the following  projections of
two values of $F$ are equal:
\begin{equation}
\label{SHolDef} \Pr\left[F(z_0)\,;[i(w_1\!-\!u)]^{-\frac{1}{2}}\right]=
\Pr\left[F(z_1)\,;[i(w_1\!-\!u)]^{-\frac{1}{2}}\right]
\end{equation}
or, equivalently,
\begin{equation}
\label{SHolDefb}
\ol{F(z_1)}\!-\!\ol{F(z_0)}= -i(w_1\!-\!u)\mesh^{-1}\cdot (F(z_1)\!-\!F(z_0)),
\end{equation}
where $(w_1u)$, $u\in\G$, $w_1\in\G^*$, is the common edge of rhombi $z_0$, $z_1$ (see
Fig.~\ref{Fig:Notations}C).
\end{definition}
Recall that orthogonal projection of $X$ on $u$ satisfies
\[
\Pr[X;u]=\Re\lt(X\frac{\ol{u}}{|u|}\rt)\frac{u}{|u|}=\frac{X}{2}+\frac{\ol{X}u^2}{2|u|^2},
\]
which we used above.

It's easy to check that
 the property to be s-holomorphic is stronger than the usual discrete
holomorphicity:
\begin{lemma}
\label{SHolIsHol} If $F:\O^\d_\DS\to\C$ is s-holomorphic, then $F$ is holomorphic in
$\O^\mesh_\DS$, i.e., $[\dopa F](v)=0$ for all $v\in\Int\O^\mesh_\L$.
\end{lemma}

\begin{proof}
Let $v\!=\!u\in\G$ (the case $v=w\in\G^*$ is essentially the same) and $F_s=F(z_s)$ (see
Fig.~\ref{Fig:Notations}B for notation). Then
\[
-i\sum_{s=1}^n (w_{s+1}\!-\!w_s)F_s = -i\sum_{s=1}^n (w_{s+1}\!-\!u)(F_s\!-\!F_{s+1}) =
-\mesh\cdot\sum_{s=1}^n (\ol{F}_{s}\!-\!\ol{F}_{s+1})=0.
\]
Thus, $[\dopa F](u)=0$.
\end{proof}

Conversely, in a simply-connected domain
every discrete holomorphic function can be decomposed
into the sum of two s-holomorphic functions (one multiplied by $i$):

\begin{lemma}
Let $\O^\mesh_\DS$ be a simply connected discrete domain and $F:\O^\mesh_\DS\to\C$ be a
discrete holomorphic function. Then there are (unique up to an additive constant)
s-holomorphic functions $F_1,F_2:\O^\mesh_\DS\to\C$ such that $F=F_1+iF_2$.
\end{lemma}

\begin{proof}
Let $\O^\mesh_\Upsilon$ denote the set of all oriented edges
$\xi=[\xi_\mathrm{b} \xi_\mathrm{w}]$ of the
rhombic lattice $\DS$ connecting neighboring vertices \mbox{$\xi_\mathrm{b}\in\O^\mesh_\G$,}
\mbox{$\xi_\mathrm{w}\in\O^\mesh_{\G^*}$.} For a function $F:\O^\mesh_\DS\to\C$, we
define its differential on edges
(more precisely, $d\,F$ a $1$-form on the edges of the dual graph, but there is no difference) by
\[
d\, F([uw_1]):=F(z_1)-F(z_0)~,~~d\,F:\O^\mesh_\Upsilon\to\C~,
\]
(see Fig.~\ref{Fig:Notations}C for notation). Then, a given antisymmetric function
$G$ defined anti-symmetrically on $\O^\mesh_\Upsilon$
is a differential %$G=d\,F$
of some discrete holomorphic function $F$
(uniquely defined on $\O^\mesh_\DS$ up to an additive constant) if for each
(black or white) vertex $u\in\O^\mesh_\L$ (see
Fig.~\ref{Fig:Notations}B for notation when $u\in\G$) two identities hold:
\begin{equation}
\label{Hol1form} \sum_{s=1}^n G([uw_s]) = 0 \quad \mathrm{and} \quad \sum_{s=1}^n
G([uw_s])(w_s\!-\!u)=0.
\end{equation}
Indeed, the first identity means that $G$ is an exact form and so a differential of some function $F$,
and the second ensures that $F$ is holomorphic by  \eqref{SHolIsHol}:
\[
\sum_{s=1}^n G([uw_s])(w_s\!-\!u) =  \sum_{s=1}^n (F(z_s)-F(z_{s-1}))(w_s\!-\!u)=
-\sum_{s=1}^n F(z_s)(w_{s+1}\!-\!w_s).
\]
Note that identities (\ref{Hol1form}) are invariant under the antilinear involution $G\mapsto G^\sharp$,
where
\[
G^\sharp([uw]):= \ol{G([uw])}\cdot i(\ol{w}\!-\!\ol{u})\mesh^{-1}
\]
On the other hand, from \eqref{SHolDefb} we see that $F$ is s-holomorphic iff $[d\,F]^\sharp = d\,F$. Thus,
the  functions $F_j$ defined by
$d\,F_1=\frac{1}{2}(d\,F+(d\,F)^\sharp)$ and $d\,F_2=\frac{1}{2i}(d\,F-(d\,F)^\sharp)$
are s-holomorphic and do the job.
Uniqueness easily follows from \eqref{SHolDef}:
If $0=F_1+iF_2$, and both functions are s-holomorphic,
then \eqref{SHolDef} implies that $F_1(z_0)=F_2(z_0)$,
and uniqueness follows.
\end{proof}

\begin{figure}
\centering{
\begin{minipage}[b]{0.4\textwidth}
\centering{\includegraphics[width=\textwidth]{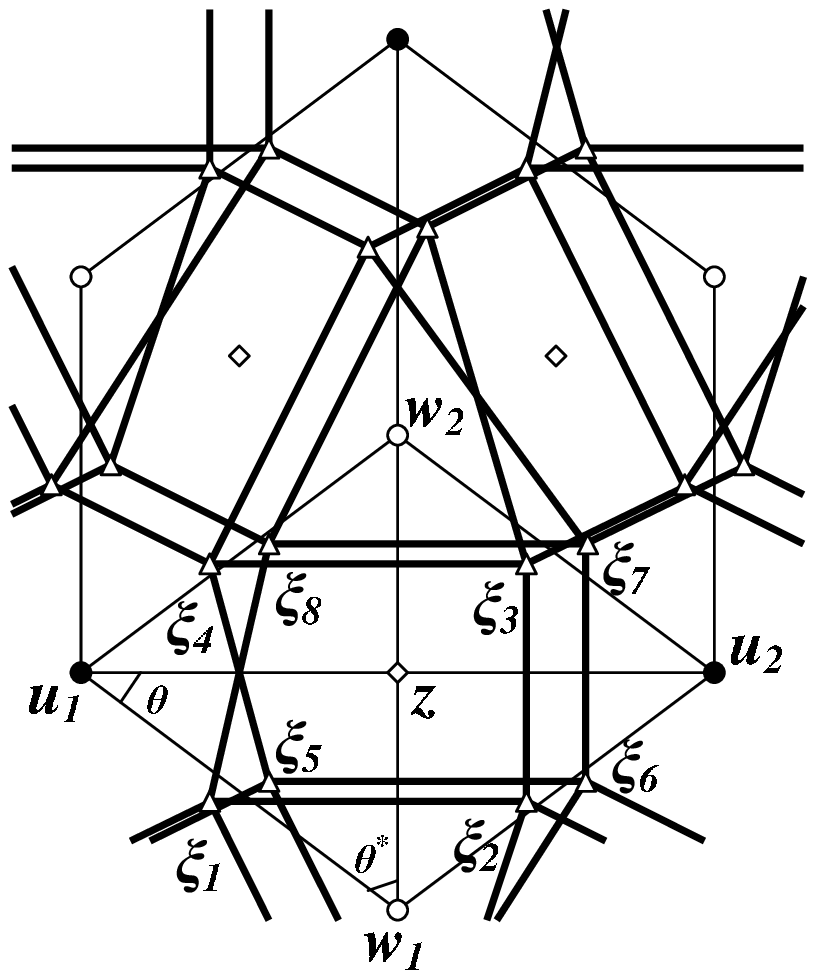}}
\end{minipage}
\hskip 0.15\textwidth
\begin{minipage}[b]{0.4\textwidth}
The spinor \mbox{$\cS(\xi):=[i(w_\mathrm{\xi}\!-\!u_\mathrm{\xi})]^{-\frac{1}{2}}$} is
naturally defined on $\widehat\Upsilon$ (``continuously'' around rhombi and vertices).
In~particular,

\bigskip

$\cS(\xi_1)=-\cS(\xi_5)=:\a$,

\medskip

$\cS(\xi_2)=-\cS(\xi_6)=e^{i\theta^*}\a$,

\medskip

$\cS(\xi_3)=-\cS(\xi_7)=i\a$,

\medskip

$\cS(\xi_4)=-\cS(\xi_8)=ie^{i\theta^*}\a$.

\bigskip \bigskip \bigskip

$\phantom{x}$
\end{minipage}}
\caption{\label{Fig:DoubleCovering} Double covering $\widehat\Upsilon$ of $\Upsilon$ ($=$
edges of $\DS$).}
\end{figure}

Following Ch.Mercat \cite{mercat-2001}, we denote  by $\widehat\Upsilon$ the \emph{
double-covering} of the set $\Upsilon$ of edges $\DS$ which is connected around each $z\in\DS$
and each $v\in\L$ (see \cite{mercat-2001}~p.209). A function $S$ defined on $\widehat\Upsilon$ is
called a \emph{spinor} if it changes the sign between the sheets. The simplest example is the
square root $[i(w\!-\!u)]^{-\frac{1}{2}}$ naturally defined on $\widehat\Upsilon$ (see
Fig.~\ref{Fig:DoubleCovering}). We say that a spinor $S$ satisfies the \emph{propagation
equation} (see \cite{mercat-2001} p.210 for historical remarks
and Fig.~\ref{Fig:DoubleCovering} for notation) if,
when walking around the edges $\xi_1,\dots,\xi_8$
of a doubly-covered rhombus,
for any three consecutive edges spinor values satisfy
\begin{equation}
\label{Propagation} S(\xi_{j+2})= (\cos\theta_j)^{-1}\cdot S(\xi_{j+1}) -\, \tan\theta_j\cdot
S(\xi_{j}),
\end{equation}
where $\theta_j$ denotes the half-angle ``between'' $\xi_{j}$ and $\xi_{j+1}$, i.e.,
$\theta_j=\theta$, if $j$ is odd, and $\theta_j=\frac{\pi}{2}-\theta$, if $j$ is even.

\begin{lemma}
Let $\O^\mesh_\DS$ be a simply connected discrete domain.
If a function $F:\O^\mesh_\DS\to\C$ is s-holomorphic then the real spinor
\begin{equation}
\label{SHol-Spinor} F_{\widehat\Upsilon}([uw])\,:=\,
\Re\left[F(z_0)\cdot[i(w\!-\!u)]^{\frac{1}{2}}\right] =
\Re\left[F(z_1)\cdot[i(w\!-\!u)]^{\frac{1}{2}}\right]
\end{equation}
satisfies the propagation equation (\ref{Propagation}). Conversely, if $F_{\widehat\Upsilon}$
is a real spinor satisfying (\ref{Propagation}), then there exists a unique  s-holomorphic
function $F$ such that (\ref{SHol-Spinor}) holds.
\end{lemma}

\begin{proof}
Note that (\ref{Propagation}) is nothing but the relation between the projections of the same
complex number onto three directions $\a=[i(w_1\!-\!u_1)]^{-\frac{1}{2}}$,
$e^{i(\frac{\pi}{2}-\theta)}\a=[i(w_2\!-\!u_1)]^{-\frac{1}{2}}$ and
$i\a=[i(w_2\!-\!u_2)]^{-\frac{1}{2}}$. Thus, s-holomorphicity of $F$ implies
(\ref{Propagation}). Conversely, starting with some real spinor $F_{\widehat\Upsilon}$
satisfying (\ref{Propagation}), one can construct a function $F$ such that
\[
\Pr[F(z)\,;\,[i(w\!-\!u)]^{-\frac{1}{2}}]=F_{\widehat\Upsilon}([uw])\cdot[i(w\!-\!u)]^{-\frac{1}{2}}
\]
(which is equivalent to (\ref{SHol-Spinor})) for any $z$, and this function is s-holomorphic
by the construction.
\end{proof}

\subsection{Integration of $F^2$ for s-holomorphic functions}
\label{SectIntF^2}

\begin{lemma}
If $F,\wt{F}:\O^\mesh_\DS\to\C$ are s-holomorphic, then $[\dopa (F\wt{F})](v)\in i\R$,
$v\in\O^\mesh_\L$.
\end{lemma}
\begin{proof}
Let $v\!=\!u\in\G$ (the case $v=w\in\G^*$ is essentially the same) and denote $F_s:=F(z_s)$ (see
Fig.~\ref{Fig:Notations}B for notation). Using \eqref{SHolDefb}, we infer that $\dopa F^2$ satisfies:
\[
-i\sum_{s=1}^n (w_{s+1}\!-\!w_s)F_s^2 %= -i\sum_{s=1}^n (w_{s+1}\!-\!u)(F_s^2\!-\!F_{s+1}^2)
= -\mesh\cdot\sum_{s=1}^n (F_s\!+\!F_{s+1})(\ol{F}_{s}\!-\!\ol{F}_{s+1}) = 2i\mesh\cdot
\Im\sum_{s=1}^n\ol{F}_sF_{s+1}\in i\R.
\]
Therefore, $4[\dopa (F\wt{F})](u)=[\dopa (F\!+\!\wt{F})^2](u)-[\dopa (F\!-\!\wt{F})^2](u)\in
i\R$.
\end{proof}

In the continuous setup, the condition $\Re[\ol{\pa} G](z) \equiv 0$ allows one to define the
function $\Im\int G(z)\,dz$ (i.e., in simply connected domains, the integral doesn't depend on the path). It is easy to check that the same holds in the discrete setup. Namely, if
$\O^\mesh_\DS$ is simply connected and $G:\O^\mesh_\DS\to\C$ is such that $\Re [\dopa G]
\equiv 0$ in $\O^\mesh_\DS$, then the discrete integral $H=\Im\int^\d G(z)d^\mesh z$ is
well-defined (i.e., doesn't depend on the path of integration) on both $\O^\mesh_\G\ss\G$ and
$\O^\mesh_{\G^*}\ss\G^*$ up to two (different for $\G$ and $\G^*$) additive constants.

It turns out that if $G=F^2$ for some s-holomorphic $F$, then $H=\Im\int^\d (F(z))^2 d^\mesh
z$ can be defined \emph{simultaneously} on $\G$ and $\G^*$ (up to \emph{one} additive constant)
in the following way:
\begin{equation}
\label{HDef} H(u)-H(w_1)~:=~
2\d\cdot\left|\Pr\left[F(z_j)\,;[i(w_1\!-\!u)]^{-\frac{1}{2}}\right]\right|^2,\quad u\sim w_1,
\end{equation}
where $(uw_1)$, $u\in\G$, $w_1\in\G^*$, is the common edge of two neighboring rhombi $z_0,z_1\in\DS$ (see
Fig.~\ref{Fig:Notations}C), and taking $j=0,1$ gives the same value.

\begin{proposition}
\label{PropHDef} Let $\O^\mesh_\DS$ be a simply connected discrete domain. If
$F:\O^\mesh_\DS\to\C$ is s-holomorphic, then

\smallskip

\noindent (i)\phantom{ii}
function $H:\O^\mesh_\L\to\C$
is well-defined (up to an additive constant) by (\ref{HDef});

\smallskip

\noindent (ii)\phantom{i} for any neighboring $v_1,v_2\in\O^\mesh_\G\ss\G$ or
$v_1,v_2\in\O^\mesh_{\G^*}\ss\G^*$ the identity
\[
H(v_2)-H(v_1)=\Im[(v_2\!-\!v_1)(F(\textfrac{1}{2}(v_1\!+\!v_2)))^2]
\]
holds (and so $H=\Im\int^\mesh (F(z))^2d^\mesh z$ on both $\G$ and $\G^*$);

\smallskip

\noindent (iii) $H$ is (discrete) subharmonic on $\G$ and superharmonic on $\G^*$, i.e.,
\[
[\D^\mesh H](u)\ge 0\quad \mathit{and}\quad [\D^\mesh H](w)\le 0
\]
for all $u\in\Int\O^\mesh_\G\ss\G$ and $w\in\Int\O^\mesh_{\G^*}\ss\G^*$;
\end{proposition}

\begin{proof}
(i),(ii) Let $z$ be the center of the rhombus $u_1w_1u_2w_2$. For $j=1,2$ we have
\begin{align*}
[H(u_2)&-H(w_j)]+[H(w_j)-H(u_1)]
\\
& = 2\left|\Re \left[[i(w_j\!-\!u_2)]^{1/2}F(z)\right]\right|^2- 2\left|\Re
\left[[i(w_j\!-\!u_1)]^{1/2}F(z)\right]\right|^2
\\
&={\textstyle\frac{1}{2}}\left[[i(w_j\!-\!u_2)]^{1/2}F(z)+
[-i(\ol{w_j}\!-\!\ol{u_2})]^{1/2}\ol{F(z)} \right]^2
\\
&\qquad\qquad - {\textstyle\frac{1}{2}}\left[ [i(w_j\!-\!u_1)]^{1/2}F(z)+
[-i(\ol{w_j}\!-\!\ol{u_1})]^{1/2}\ol{F(z)}\right]^2
\\
&={\textstyle\frac{1}{2}}\left[i(u_1\!-\!u_2)(F(z))^2 - i
(\ol{u}_1\!-\!\ol{u}_2)(\ol{F(z)})^{2}\right]=\Im\left[(u_2\!-\!u_1)(F(z))^2\right]\,.
\end{align*}
The computations for $H(w_2)\!-\!H(w_1)$ are similar.

\smallskip

\noindent (iii) Let $u\in\G$ and set $F_s:=F(z_s)$ (see Fig.~\ref{Fig:Notations}B for notation).
Denote
\[
t_s:=\Pr\left[F_s\,;[i(w_s\!-\!u)]^{-\frac{1}{2}}\right]=
\Pr\left[F_{s-1}\,;[i(w_s\!-\!u)]^{-\frac{1}{2}}\right].
\]
Knowing two projections of $F_s$ uniquely determines it's value:
\[
F_s=\frac{i(t_se^{-i\theta_s}-t_{s+1}e^{i\theta_s})}{\sin\theta_s}
\]
and so
\[
\weightG{u}\cdot[\D^\mesh H](u)= -\Im \lt[\sum_{s=1}^n\tan\theta_s\cdot\frac{
(t_se^{-i\theta_s}-t_{s+1}e^{i\theta_s})^2\cdot (u_s\!-\!u)}{\sin^2\theta_s}\rt]
\]
\[
=-2\mesh\cdot\Im \lt[ \sum_{s=1}^n \frac{
(t_s^2e^{-2i\theta_s}-2t_st_{s+1}+t_{s+1}^2e^{2i\theta_s})\cdot
e^{i\arg(u_s-u)}}{\sin\theta_s} \rt].
\]
%\[
%= -2\mesh\cdot\Im \lt[ \sum_{s=1}^n \frac{
%-i|t_s|^2e^{-i\theta_s}-i|t_{s+1}|^2e^{i\theta_s}-2t_st_{s+1}e^{i\arg(u_s-u)}}{\sin\theta_s}
%\rt]
%\]
Let $t_s=x_s\cdot \exp\left[i\arg\left([i(w_s\!-\!u)]^{-\frac{1}{2}}\right)\right]$, where
$x_s\in\R$ and the argument of the square root changes ``continuously''
when we move around $u$, taking $s=1,..,n$.
Then
\[
(2\mesh)^{-1}\weightG{u}\cdot[\D^\mesh H](u)= -\Im \lt[ -i\sum_{s=1}^n \frac{
x_s^2e^{-i\theta_s}\mp 2x_sx_{s+1}+x_{s+1}^2e^{i\theta_s}}{\sin\theta_s} \rt]
\]
\[
= \sum_{s=1}^n \frac{ \cos\theta_s\cdot(x_s^2+x_{s+1}^2)\mp
2x_sx_{s+1}}{\sin\theta_s}=:Q_{\theta_1\,;\,..\,;\,\theta_n}^{(n)}(x_1,..,x_n)\ge 0,
\]
where ``$\mp$'' is ``$-$'' for all terms except $x_nx_1$ (these signs come from our choice of
arguments). The non-negativity of the quadratic form $Q^{(n)}$ (for arbitrary
$\theta_1,..,\theta_n>0$ with $\theta_1\!+\!..\!+\!\theta_n=\pi$) can be easily shown by
induction. Indeed, the identity
\[
Q^{(n)}_{\theta_1\,;\,..\,;\,\theta_n}(x_1,..,x_n)-
Q^{(n-1)}_{\theta_1\,;\,..\,;\,\theta_{n-2}\,;\,\theta_{n-1}+\theta_n}(x_1,..,x_{n-1})
%\[
%= \frac{ \cos\theta_{n-1}\cdot(x_{n-1}^2+x_n^2)- 2x_{n-1}x_n}{\sin\theta_{n-1}} + \frac{
%\cos\theta_n\cdot(x_n^2+x_1^2)+ 2x_nx_1}{\sin\theta_n} - \frac{
%\cos\wt{\theta}_{n-1}\cdot(x_{n-1}^2+x_1^2)+ 2x_{n-1}x_1}{\sin\wt{\theta}_{n-1}}
%\]
=Q^{(3)}_{\pi-\theta_{n-1}-\theta_n\,;\,\theta_{n-1}\,;\,\theta_n}(x_1,x_{n-1},x_n),
\]
reduces the problem to the non-negativity of $Q^{(3)}$'s.
But, if $\a,\b,\g>0$ and $\a\!+\!\b\!+\!\g=\pi$, then
\[
Q^{(3)}_{\a\,;\,\b\,;\,\g}(x,y,z)=\lt[
\frac{\sin^{\frac{1}{2}}\b}{\sin^{\frac{1}{2}}\a\cdot\sin^{\frac{1}{2}}\g}\cdot x -
\frac{\sin^{\frac{1}{2}}\g}{\sin^{\frac{1}{2}}\a\cdot\sin^{\frac{1}{2}}\b}\cdot y +
\frac{\sin^{\frac{1}{2}}\a}{\sin^{\frac{1}{2}}\b\cdot\sin^{\frac{1}{2}}\g}\cdot z \rt]^2.
\]
This finishes the proof for $v=u\in\G$,
the case $v=w\in\G^*$ is similar. The opposite sign of $[\D^\mesh H](w)$ comes from the
invariance of definition (\ref{SHolDef}) under the (simultaneous) multiplication of $F$ by $i$
and the transposition of $\G$ and $\G^*$.
\end{proof}

\begin{remark}
As it was shown above, the quadratic form $Q^{(n)}_{\theta_1\,;\,..\,;\,\theta_n}(x_1,..,x_n)$
can be represented as a sum of $n-2$ perfect squares $Q^{(3)}$. Thus, its kernel is (real)
two-dimensional. Clearly, this corresponds to the case when all the (complex) values $F(z_s)$
are equal to each other, since in this case $[\D^\mesh H]=0$ by definition. Thus,
\begin{equation}
\label{DHasympGradF} \mesh\cdot|[\D^\mesh H](u)|\asymp
Q_{\theta_1\,;\,..\,;\,\theta_n}^{(n)}(x_1,..,x_n)
\asymp\sum_{s=1}^n|F(z_{s+1})\!-\!F(z_s)|^2,
\end{equation}
since both sides considered as real quadratic forms in $x_1,..,x_n$ (with coefficients of
order~$1$) are nonnegative and have the same two-dimensional kernel.
\end{remark}

\subsection{Harnack Lemma for the integral of $F^2$}
\label{SectHarnackH}

As it was shown above, using (\ref{HDef}), for any s-holomorphic function $F$, one can define
a function $H=\Im\int^\mesh(F^\mesh(z))^2d^\mesh z$ on both $\G$ and $\G^*$. Moreover,
$H\big|_\G$ is subharmonic while $H\big|_{\G^*}$ is superharmonic.
It turns out, that $H$, despite  not being a harmonic function,
 \emph{a priori} satisfies a version of the Harnack Lemma (cf.~Proposition~\ref{PropHarnack}(i)).
In Section~\ref{SectBoundHarnack} we also prove a
version of the boundary Harnack principle which compares the values of $H$ in the bulk with
its normal derivative on a straight part of the boundary.

We start by showing that $H\big|_\G$ and $H\big|_{\G^*}$ cannot
differ too much.

\begin{proposition}
\label{PropH-H} \noindent Let $w\in\Int\O^\mesh_{\G^*}$ be an inner face surrounded by inner
vertices \mbox{$u_s\in\Int\O^\mesh_\G$} and faces $w_s\in\O^\mesh_{\G^*}$, $s=1,..,n$.
% and points $z_s=\frac{1}{2}(w\!+\!w_s)=\frac{1}{2}(u_s\!+\!u_{s+1})\in\Int\O^\mesh_\DS$
If $H$ is defined by (\ref{HDef}) for some \mbox{s-holo}morphic function
$F:\O^\mesh_\DS\to\C$, then
\[
\max_{s=1,..n} H(u_s)-H(w)\le \const\cdot\left(H(w)-\min_{s=1,..,n}H(w_s)\right).
\]
\end{proposition}

\begin{remark} Since definition (\ref{SHolDef}) is invariant under the (simultaneous) multiplication of $F$ by $i$
and the transposition of $\G$ and $\G^*$, one also has
\[
H(u)-\min_{s=1,..n} H(w_s)\le \const\cdot(\max_{s=1,..,n}H(u_s)-H(u))
\]
for any inner vertex \mbox{$u\in\Int\O^\mesh_{\G}$} surrounded by
\mbox{$w_s\in\Int\O^\mesh_{\G^*}$} and \mbox{$u_s\in\O^\mesh_{\G}$}, $s=1,..,n$.
\end{remark}

\begin{proof}
By subtracting a constant, we may assume that $\min_{s=1,..,n}H(w_s)=0$. Since
$H\big|_{\G^*}$ is superharmonic at $w$, it is non-negative there and moreover
\[
H(w)\ge \const\cdot H(w_s) \quad
\mathrm{for\ all}\ s=1,..,n.
\]
Since $H\big|_\G(w_s)\ge0$, by \eqref{HDef}, it is non-negative at all
points of $\G$ which are neighbors of $w_s$'s.
From subharmonicity of $H\big|_\G$
we similarly deduce that $H(u_s)\le\const\cdot H(u_{s+1})$. Therefore,
\[
H(u_s)\asymp M:=\max_{s=1,..n} H(u_s)\quad \mathrm{for\ all}\ s=1,..,n.
\]
We need to prove that $K:=M/H(w)\le\const$. Assume the opposite, i.e. $K\gg 1$. Then, for any
$s=1,..,n$, one has
\[
\frac{H(u_s)\!-\!H(w_s)}{H(u_s)\!-\!H(w)}=1\!+\!O(K^{-1}).
\]
By (\ref{HDef}), these increments of $H$ are derived from projections of $F(z_s)$
on two directions, and so
\[
\frac{|\Pr[F(z_s)\,;\,[i(w_s-u_s)]^{-\frac{1}{2}}]|}{
|\Pr[F(z_s)\,;\,[i(w-u_s)]^{-\frac{1}{2}}]|} = 1\!+\!O(K^{-1}).
\]
Here $z_s$ is a center of the rhombus $wu_sw_su_{s+1}$, and we infer
\[
\arg F(z_s)~\mathrm{mod}~\pi~ =~\left[\mathrm{either}~
\arg[(w_s\!-\!w)^{-\frac{1}{2}}]~~\mathrm{or}~\arg[(w\!-\!w_s)^{-\frac{1}{2}}]\right]+
O(K^{-1}).
\]
Moreover, $\mesh\cdot |F(z_s)|^2\asymp H(u_s)\!-\!H(w)\asymp M$ for all $s=1,..,n$, if $K$ is
big enough. Since $F(z_s)$ and $F(z_{s+1})$ have very different arguments, using (\ref{DHasympGradF})
we prove the proposition:
\[
M\asymp \mesh\cdot \sum_{s=1}^n|F(z_{s+1})\!-\!F(z_s)|^2\asymp \mesh^2\cdot |[\D^\mesh
H](w)|\le \const\cdot H(w).\qedhere
\]
\end{proof}

\begin{remark}[\bf uniform comparability of $\bm{H^\mesh_\G}$ and $\bm{H^\mesh_{\G^*}}$]
\label{RemHGasympG*} Suppose $H\big|_{\G^*}\ge 0$ on $\pa\O^\mesh_{\G^*}$ and hence, due to its
superharmonicity, everywhere inside $\O^\mesh_{\G^*}$. Then, definition (\ref{HDef}) and
Proposition~\ref{PropH-H} give
\[
H^\mesh_{\G^*}(w)\le H^\mesh_\G(u)\le \const\cdot H^\mesh_{\G^*}(w)
\]
for all neighboring $u\sim w$ in $\O^\mesh$, where the constant is independent of $\O^\mesh$
and $\mesh$.
\end{remark}

\begin{proposition}[\bf Harnack Lemma for $\bm{\Im\int^\mesh (F(z))^2d^\mesh z}$]
\label{HarnackIntF^2} Take $v_0\in\L=\G\cup\G^*$ and let $F:B^\mesh_\DS(u_0,R)\to\C$ be an
s-holomorphic function. Define
\[
H:=\Im\int^\mesh (F(z))^2d^\mesh z:\O^\mesh_\L\to\R
\]
by (\ref{HDef}) so that $H\ge 0$ in $B^\mesh_{\L}(u_0,R)$. Then,
\[
H(v_1)\le\const\cdot H(v_0)\ \ \mathit{for\ any}\ \ v_1\in
B^\mesh_\L(u_0,{\textstyle\frac{1}{2}}r).
\]
\end{proposition}

\begin{proof} Due to Remark~\ref{RemHGasympG*}, we may assume that $v_0\in\G^*$ while $v_1\in\G$.
Set $M:=\max_{u\in B^\mesh_\G(v_0,{\frac{1}{2}}r)} H(u)$. Since the function $H\big|_\G$ is
subharmonic, one has
\[
M=H(v_1)\le H(v_2)\le H(v_3)\le \dots
\]
for some path of consecutive neighbors $K^\mesh_\G=\{v_1\sim v_2\sim v_3\sim\dots\}\ss\G$
running from $v_1$ to the boundary $\pa B^\mesh_\G(v_0,R)$.
Take a nearby path of consecutive
neighbors $K^\mesh_{\G^*}=\{w_1^\mesh\sim w_2^\mesh\sim w_3^\mesh\sim\dots\}\ss\G^*$ starting
in $\pa B^\mesh_{\G^*}(v_0,\frac{1}{2}R)$ and running to $\pa B^\mesh_{\G^*}(v_0,R)$.
By Remark~\ref{RemHGasympG*}, $H\big|_{\G^*}\ge \const\cdot M$ on $K^\mesh_{\G^*}$, and so
\[
H(v_0)\ge \const\cdot M\cdot \dhm\mesh{v_0}{K^\mesh_{\G^*}}{B^\mesh_{\G^*}(v_0,R)\setminus
K^\mesh_{\G^*}}\ge \const\cdot M,
\]
since $H\big|_{\G^*}$ is superharmonic and the discrete harmonic measure of the path
$K^\mesh_{\G^*}$ viewed from $v_0$ is uniformly bounded from below due to standard random walk
arguments (cf.~\cite{chelkak-smirnov-dca} Proposition 2.11).
\end{proof}

\subsection{Regularity of s-holomorphic functions}
\label{SectRegSHol}

\begin{theorem}
\label{FboundH} For a simply connected $\O^\mesh_\DS$ and an s-holomorphic $F:\O^\mesh_\DS\to\C$
define \mbox{$H=\int^\mesh (F(z))^2d^\mesh z:\O^\mesh_\L\to\C$} by
\eqref{HDef} in accordance with
 Proposition~\ref{PropHDef}. Let point $z_0\in\Int\O^\mesh_\DS$
be a definite distance from the boundary:
$d=\dist(z_0;\pa\O^\mesh_\DS)\ge \const\cdot\mesh$ and set $M=\max_{v\in\O^\mesh_\L}|H(v)|$. Then
\begin{equation}
\label{SHolBdd} |F(z_0)|\le \const\cdot\frac{M^{1/2}}{d^{1/2}}
\end{equation}
and, for any neighboring $z_1\sim z_0$,
\begin{equation}
\label{SHolLip} |F(z_1)-F(z_0)|\le \const\cdot \frac{M^{1/2}}{d^{3/2}}\cdot\mesh.
\end{equation}
\end{theorem}

\begin{remark}
Estimates (\ref{SHolBdd}) and (\ref{SHolLip}) have exactly the same form as if $H$ would be
harmonic. Due to (\ref{HDef}) and (\ref{DHasympGradF}), a posteriori this also means that the
subharmonic function $H|_\G$ and the superharmonic function $H|_{\G^*}$ should be uniformly
close to each other inside $\O^\mesh$, namely $H\big|_\G-H\big|_{\G^*}=O(\mesh M/d)$, and,
moreover, $|\D^\mesh H|=O(\mesh M/d^3)$.
\end{remark}

\noindent The proof consists of four steps:

\smallskip

\noindent \emph{Step 1.} {Let $B^\mesh_\G(z_0;r)\ss\G$ denote the discrete disc centered at
$z_0$ of radius $r$. Then the discrete $L_1$ norm (as defined below) of the Laplacian of $H$ satisfies
\begin{equation}
\label{x1} \left\|\D^\mesh H\right\|_{1\,;\,B^\mesh_\G(z_0;\frac{3}{4}d)}:=\sum\nolimits_{u\in
B^\mesh_\G(z_0;\frac{3}{4}d)} \left|[\D^\mesh H](u)\right|\weightG{u} \le \const\cdot M
\end{equation}
and the same estimate for $H$ restricted to $\G^*$ holds.}

\begin{proof}[Proof of Step 1.] Represent $H$ on $B^\mesh_\G(z_0;d)$ as a sum of a harmonic function
with the same boundary values and a subharmonic one:
\[
H|_\G=H_{\mathrm{harm}}+H_{\mathrm{sub}},\quad \D^\mesh H_{\mathrm{harm}}=0,~~\D^\mesh
H_{\mathrm{sub}}\ge 0\quad \mathrm{in}~~B^\mesh_\G(z_0;d)
\]
and
\[
H_{\mathrm{harm}}=H,~~H_{\mathrm{sub}}=0\quad \mathrm{on}~~\pa B^\mesh_\G(z_0;d).
\]
Note that the negative function $H_{\mathrm{sub}}$ satisfies
\[
H_{\mathrm{sub}}(\cdot)~=\!\!\sum_{u\in \Int B^\mesh_\G(z_0;d)} G(\cdot\,;u)[\D^\mesh
H_{\mathrm{sub}}](u)\weightG{u}~\le\!\sum_{u\in B^\mesh_\G(z_0;\frac{3}{4}d)}
G(\cdot\,;u)[\D^\mesh H_{\mathrm{sub}}](u)\weightG{u},
\]
since the Green's function $G(\cdot\,;u)$
in $B^\mesh_\G(z_0;d)$ is negative.
By the maximum principle, $|H_{\mathrm{harm}}|\le M$ and so $|H_{\mathrm{sub}}|\le 2M$ in
$B^\mesh_\G(z_0;d)$.
Therefore,
\[
\const\cdot Md^2 \ge \|H_{\mathrm{sub}}\|_{1\,;B^\mesh_\G(z_0;d)}\ge \|\D^\mesh
H_{\mathrm{sub}}\|_{1\,;B^\mesh_\G(z_0;\frac{3}{4}d)}~\cdot\!\! \min_{u\in
B^\mesh_\G(z_0;\frac{3}{4}d)}\|G(\cdot\,;u)\|_{1\,;B^\mesh_\G(z_0;d)}.
\]
Since $\D^\mesh H_{\mathrm{sub}}=\D^\mesh H$, the inequality
(\ref{x1}) follows from the (uniform) estimate
\[
\|G(\cdot\,;u)\|_{1\,;B^\mesh_\G(z_0;d)}\ge \const\cdot d^2\quad \mathrm{for~all}~~u\in
B^\mesh_\G(z_0;\textstyle{\frac{3}{4}}d)
\]
which we prove in Appendix (Lemma \ref{Gestimate}).
\end{proof}

\smallskip

\noindent\emph{Step 2.} {The estimate
\begin{equation}
\label{x2} \left\|F\right\|^2_{2\,;\,B^\mesh_\DS(z_0;\frac{1}{2}d)}:=\sum\nolimits_{z\in
B^\mesh_\DS(z_0;\frac{1}{2}d)} \left|F(z)\right|^2\weightDS{z} \le \const\cdot Md
\end{equation}
holds. }

\begin{proof}[Proof of Step 2.]
Since $H=\Im\int^\mesh(F(z))^2d^\mesh z$ on both $\G$ and $\G^*$, it is sufficient to prove
\[
\left\|\,\dpa [H|_{\G}]\,\right\|_{1\,;\,B^\mesh_\DS(z_0;\frac{1}{2}d)}\le \const\cdot Md
\]
and a similar estimate for $\dpa [H|_{\G^*}]$\,. Represent $H$ on $B^\mesh_\G(z_0;\textfrac{3}{4}d)$ as a sum:
\[
H|_\G=H_{\mathrm{harm}}+H_{\mathrm{sub}},\quad \D^\mesh H_{\mathrm{harm}}=0,~~\D^\mesh
H_{\mathrm{sub}}\ge 0\quad \mathrm{in}~~B^\mesh_\G(z_0;\textfrac{3}{4}d)
\]
and
\[
H_{\mathrm{harm}}=H,~~H_{\mathrm{sub}}=0\quad \mathrm{on}~~\pa
B^\mesh_\G(z_0;\textfrac{3}{4}d).
\]
It follows from the discrete Harnack's Lemma (see Corollary~\ref{CorHarnack}) and
the estimate $|H_{\mathrm{harm}}|\le M$ that
\[
|\dpa H_{\mathrm{harm}}(z)|\le \const \cdot {M}/{d}%\frac{M}{d}
\]
for all $z\in B^\mesh_\G(z_0;\textfrac{1}{2}d)$ and hence
\[
\left\|\,\dpa H_{\mathrm{harm}}\,\right\|_{1\,;\,B^\mesh_\DS(z_0;\frac{1}{2}d)}\le \const\cdot
Md.
\]
Furthermore,
\[
[\dpa H_{\mathrm{sub}}](z)~=\!\!\sum_{u\in \Int B^\mesh_\G(z_0;\frac{3}{4}d)} [\dpa
G](z\,;u)[\D^\mesh H](u)\weightG{u},
\]
where $G(\cdot\,;u)\le 0$ denotes the Green's function in $B^\mesh_\G(z_0;\textfrac{3}{4}d)$.
We infer that
\[
\left\|\,\dpa H_{\mathrm{sub}}\,\right\|_{1\,;\,B^\mesh_\DS(z_0;\frac{1}{2}d)} \le
\left\|\D^\mesh H\right\|_{1\,;\,B^\mesh_\G(z_0;\frac{3}{4}d)}\cdot \max_{u\in\Int
B^\mesh_\G(z_0;\frac{3}{4}d)}\left\|\,[\dpa
G](\cdot\,;u)\,\right\|_{1\,;\,B^\mesh_\DS(z_0;\frac{1}{2}d)}.
\]
The first factor is bounded by $\const\cdot M$ (Step 1) and the second, as in the continuous
setup, is bounded by $\const\cdot d$ (see Lemma \ref{Ggradestimate}),
which concludes the proof.
\end{proof}

\smallskip

\noindent\emph{Step 3.} {The uniform estimate
\begin{equation}
\label{x3} |F(z)| \le \const\cdot \frac{M^{1/2}}{d^{1/2}}
\end{equation}
holds for all $z\in B^\mesh_\DS(z_0;\frac{1}{4}d)$. }

\begin{proof}[Proof of Step 3.]
Applying the Cauchy formula (see Lemma \ref{CauchyFormula} (i)), to the discs
\[
\textstyle B^\mesh_\DS(z;5\mesh k) \ss B^\mesh_\DS(z_0;\frac{1}{2}d),\quad
k:~\frac{1}{8}d<5\mesh k<\frac{1}{4}d,~k\in\Z~,
\]
whose boundaries don't intersect each other and summing over all such $k$, we estimate
\[
\frac{d}{\mesh}\cdot |F(z)|\le \frac{\const}{d}\cdot
\frac{\|F\|_{1\,;\,B^\mesh_\DS(z_0;\frac{1}{2}d)}}{\mesh}\,.
\]
Thus, by the Cauchy-Schwarz inequality and \eqref{x2}
\[
|F(z)|\le \frac{\const}{d^2}\cdot \|F\|_{1\,;\,B^\mesh_\DS(z_0;\frac{1}{2}d)} \le
\frac{\const}{d}\cdot \|F\|_{2\,;\,B^\mesh_\DS(z_0;\frac{1}{2}d)}\le \const\cdot
\frac{M^{1/2}}{d^{1/2}}\,. \qedhere
\]
\end{proof}

\smallskip

\noindent\emph{Step 4.} {The estimate
\begin{equation}
\label{x4} |F(z_1)-F(z_0)|\le \const\cdot \frac{M^{1/2}}{d^{3/2}}\cdot\mesh.
\end{equation}
holds for all $z_1\sim z_0$. }

\begin{proof}[Proof of Step 4.]
It follows from the discrete Cauchy formula (see Lemma \ref{CauchyFormula} (ii)) applied to
the disc $B^\mesh_\DS(z_0;\frac{1}{4}d)$ and Step 3 for its boundary that there exist $A,B\in\C$ such that, for
any neighboring $z_1\sim z_0$ (see Fig.~\ref{Fig:Notations}C for notation), we have for both points the identity
\[
F(z_{0,1})=\Pr[A\,;\,\ol{u_{0,1}\!-\!u}]+\Pr[B\,;\,\ol{w_{0,2}\!-\!w_1}]+O(\epsilon)=
A+\Pr[B\!-\!A\,;\,\ol{w_{0,2}\!-\!w_1}]+O(\epsilon),
\]
where
\[
\epsilon= \frac{\max_{z\in\pa B^\mesh_\DS(z_0;\frac{1}{4}d)} |F(z)|\cdot \mesh}{d}\le
\frac{M^{1/2}\cdot \mesh}{d^{3/2}}\,.
\]
The definition of an s-holomorphic function stipulates that
\[
\Pr[F(z_0)-F(z_1)\,;\,[i(w_1\!-\!u)]^{-\frac{1}{2}}]=0
\]
(for all $z_1\sim z_0$), which implies $B\!-\!A=O(\epsilon)$, and hence
$F(z_{i})=A+O(\epsilon)$ for $i=0,1$.
\end{proof}

\subsection{The ``${(\t(z))^{-\frac{1}{2}}}$'' boundary condition and the ``boundary
modification trick''} \label{SectBModTrick}

Throughout the paper, we often deal with s-holomorphic in $\O^\mesh_\DS$ functions $F^\mesh$
satisfying the Riemann boundary condition
\begin{equation}
\label{RH-1/2bc} F(\z)\parallel (\t(z))^{-\frac{1}{2}}
\end{equation}
on, say, the ''white'' boundary arc $L^\mesh_{\G^*}\ss\pa\O^\mesh_\DS$ (see
Fig.~\ref{Fig:BoundaryTrick}A), where \mbox{$\t(z):=w_2(z)\!-\!w_1(z)$} is the ``discrete
tangent vector'' to $\pa\O^\mesh_\DS$ at $\z$. Being s-holomorphic, these functions posses
discrete primitives
\[
H^\mesh:= \Im \int^\mesh (F^\mesh(z))^2 d^\mesh z.
\]
Due to the boundary condition (\ref{RH-1/2bc}), an additive constant can be fixed so that
\[
H^\mesh_{\G^*}=0~~\mathrm{on}~L^\mesh_{\G^*}.
\]
The boundary condition for $H^\mesh_\G$ is more complicated. Fortunately, one can reformulate
it exactly in the same way, using the following {\bf ``boundary modification trick''}:
\begin{quotation}
\noindent {\it For each half-rhombus $u_{\mathrm{int}}w_1w_2$
touching the boundary arc $L^\mesh_{\G^*}$, we draw two new rhombi
 $u_{\mathrm{int}}w_1\wt{u}_1\wt{w}$ and
$u_{\mathrm{int}}\wt{w}\wt{u}_2w_2$ so that the corresponding angles
$\wt{\theta}_1=\wt{\theta}_2$ are equal to $\frac{1}{2}\theta$
(see Fig.~\ref{Fig:BoundaryTrick}A).}
\end{quotation}
In general, these new edges $(u_{\mathrm{int}}\wt{u}_{1,2})$ constructed for neighboring inner
vertices $u_{\mathrm{int}}$, may intersect each other but it is not important for us
(one can resolve the problem of a locally self-overlapping domain by placing it on a Riemann surface).

\begin{figure}
\centering{
\begin{minipage}[b]{0.4\textwidth}
\centering{\includegraphics[width=\textwidth]{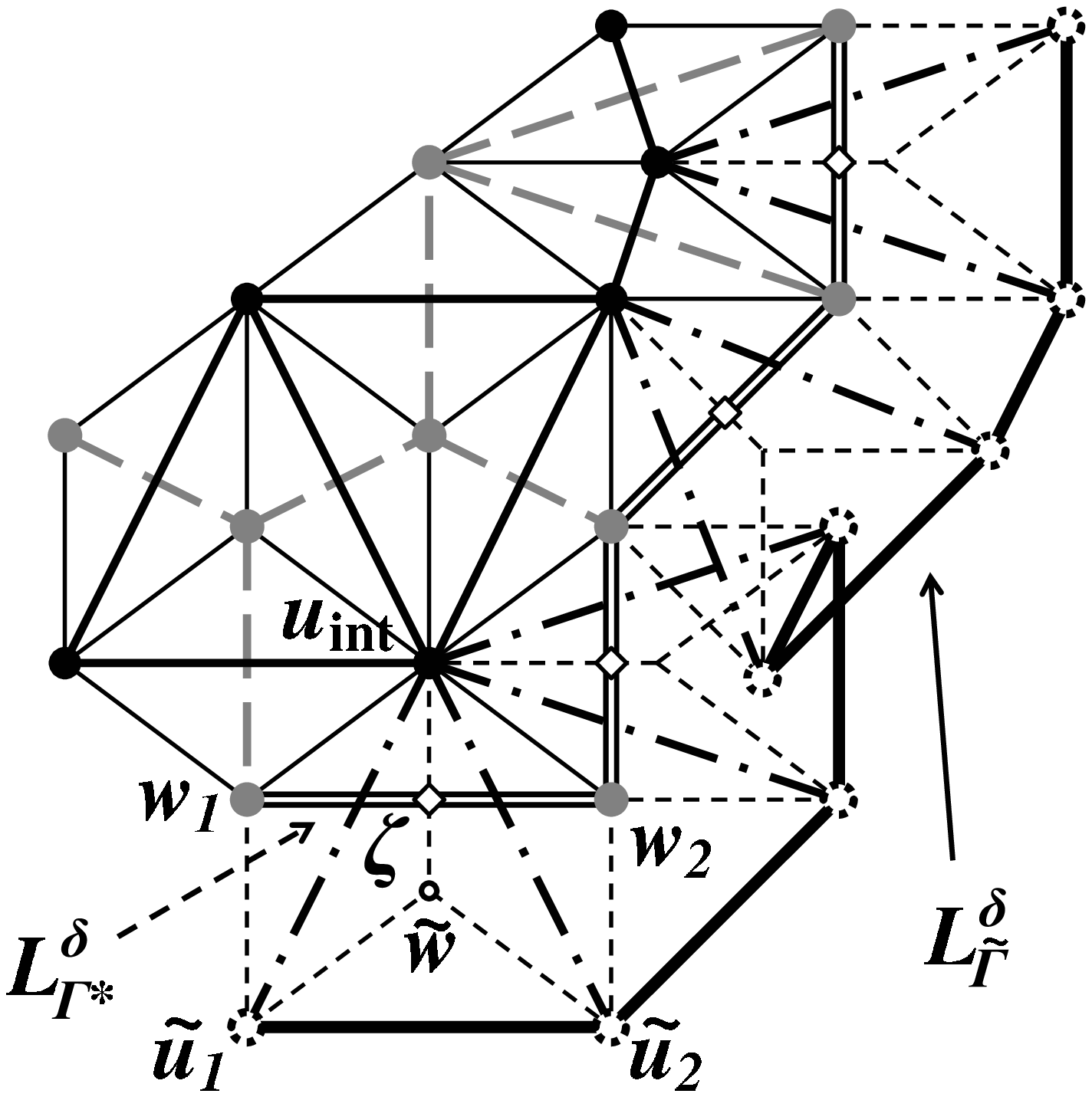}}

\bigskip\bigskip

\centering{\textsc{(A)}

\bigskip
\bigskip
$\phantom{x}$ }
\end{minipage}
\hskip 0.1\textwidth
\begin{minipage}[b]{0.33\textwidth}
\centering{\includegraphics[width=\textwidth]{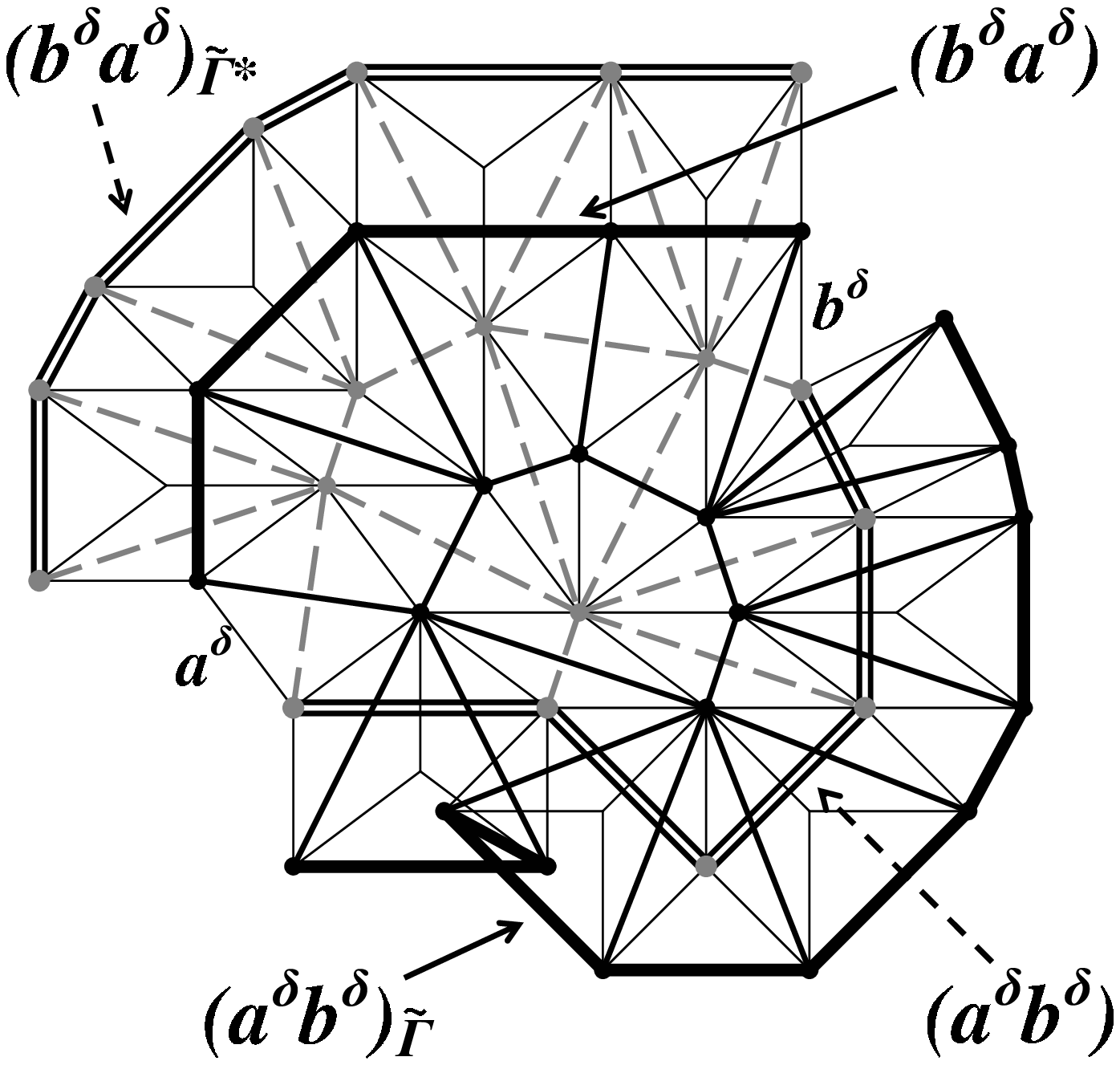}}

\smallskip

\centering{\textsc{(B)}}

\bigskip

\centering{\includegraphics[width=\textwidth]{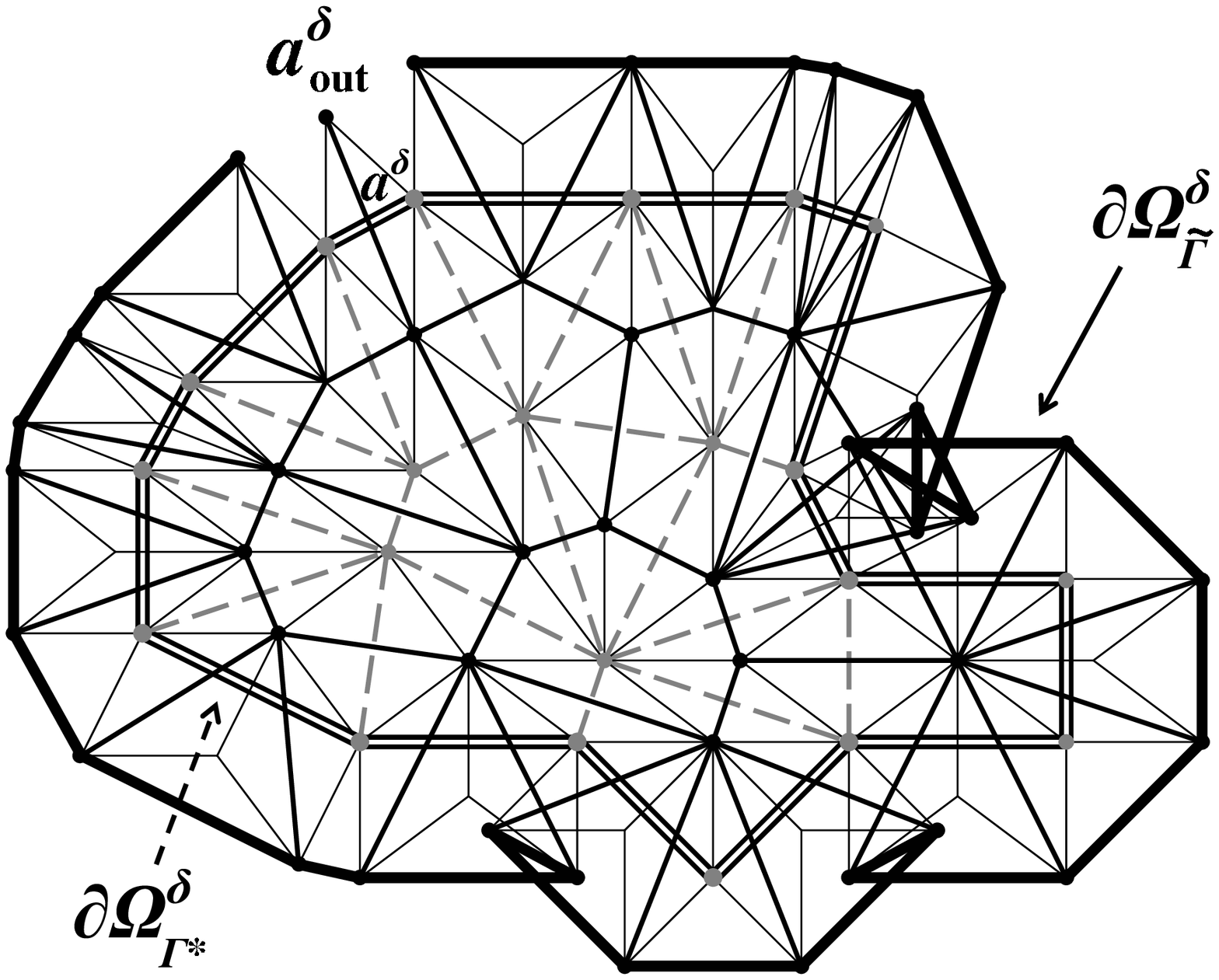}}

\smallskip

\centering{\textsc{(C)}}
\end{minipage}}
\caption{\label{Fig:BoundaryTrick} ``Boundary modification trick'': \textsc{(A)}~local
modification of the boundary; \textsc{(B)}~FK-Ising model: example of a modified domain;
\textsc{(C)}~spin-Ising model: example of a modified domain. }
\end{figure}

\begin{lemma}
%\label{BModTrick}
Let $u_{\mathrm{int}}w_1w_2$ be the half-rhombus touching $L^\mesh_{\G^*}$ and
$\z=\frac{1}{2}(w_2\!+\!w_1)$. Suppose that the function $F^\mesh$ is s-holomorphic in $\O^\mesh_\DS$
and $F^\mesh(\z)\parallel (w_2\!-\!w_1)^{-\frac{1}{2}}$. Then, if we set
\[
H^\mesh_\G(\wt{u}_2)=H^\mesh_\G(\wt{u}_1)~:=~H^\mesh_{\G^*}(w_2)=H^\mesh_{\G^*}(w_1),
\]
the function $H^\mesh_{\G}$ remains subharmonic at $u_{\mathrm{int}}$.
%Moreover, its discrete Laplacian at $u_{\mathrm{int}}$ on the modified lattice is equal to its
%discrete Laplacian on the original lattice.
\end{lemma}

\begin{proof}
Definition (\ref{HDef}) says that
\[
H^\mesh_\G(u_{\mathrm{int}}) - H^\mesh_{\G^*}(w_{1,2}) =
2\mesh|\Pr[F^\mesh(\z)\,;\,[i(w_{1,2}\!-\!u_\mathrm{int})]^{-\frac{1}{2}}]|^2=
2\mesh\cos^2\textfrac{\theta}{2}\cdot|F^\mesh(\z)|^2
\]
because $i(w_{1,2}\!-\!u_\mathrm{int})\upuparrows e^{\mp i\theta}(w_2\!-\!w_1)$ and
$F^\mesh(\z)\parallel (w_2\!-\!w_1)^{-\frac{1}{2}}$. Therefore
\[
2\tan\textfrac{\theta}{2}\cdot(H^\mesh_\G(\wt{u}_{1,2})-H^\mesh_\G(u_{\mathrm{int}})) =
-2\mesh\sin\theta\cdot |F^\mesh(\z)|^2
\]
and, if $u$ denotes the fourth vertex of the rhombus $u_{\mathrm{int}}w_1uw_2$,
\[
\tan\theta\cdot (H^\mesh_\G(u)-H^\mesh_\G(u_{\mathrm{int}}))= \tan\theta\cdot \Im
[(F^\mesh(\z))^2(u\!-\!u_{\mathrm{int}})]= -2\mesh\sin\theta\cdot|F^\mesh(\z)|^2.
\]
Thus, the standard definition of $H^\mesh_\G$ at $u$ and the new definition at $\wt{u}_{1,2}$
give the same contributions to the (unnormalized) discrete Laplacian at $u_{\mathrm{int}}$
(see (\ref{DMeshDef})).
\end{proof}

\begin{remark}
After this trick, an additive constant in the definition of $H^\mesh$ can be chosen so that
both
\[
H^\mesh_{\G^*}=0~\mathit{on}~L^\mesh_{\G^*}\quad \mathit{and}\quad H^\mesh_\G=0\
\mathit{on}~\wt{L}^\mesh_\G,
\]
where $L^\mesh_{\wt{\G}}$ denotes the set of newly constructed ``black'' vertices near
$L^\mesh_{\G^*}$.
\end{remark}

\section{Uniform convergence for the holomorphic observable in the FK-Ising model}
\setcounter{equation}{0} \label{SectFKobservable}

Recall that the discrete holomorphic fermion
$F^\mesh(z)=F^\mesh_{(\O^\mesh;a^\mesh,b^\mesh)}(z)$ constructed in Section~\ref{SectFKIsing}
satisfies the following discrete boundary value problem:

\begin{quotation} {\it
\noindent {\bf (A) Holomorphicity:} $F^\mesh(z)$ is s-holomorphic in $\O^\mesh_\DS$.

\noindent {\bf (B) Boundary conditions:} $F^\mesh(\z)\parallel (\t(\z))^{-\frac{1}{2}}$ for
$\z\in\pa \O^\mesh_\DS$, where
\begin{equation}
\label{BVPdiscr}
\begin{array}{lll}
\t(\z)=w_2(\z)\!-\!w_1(\z), & \z\in (a^\mesh b^\mesh), & w_{1,2}(\z)\in\G^*, \cr
\t(\z)=u_2(\z)\!-\!u_1(\z), & \z\in (b^\mesh a^\mesh), & u_{1,2}(\z)\in\G,
\end{array}
\end{equation}
is the ``discrete tangent vector'' to $\pa\O^\mesh_\DS$ directed from $a^\mesh$ to $b^\mesh$
on both arcs (see Fig.~\ref{Fig:FKDomain}B).

\noindent {\bf (C) Normalization at $\bm{b^\mesh}$:} $F^\mesh(b^\mesh)= \Re
F^\mesh(b^\mesh_\DS)= (2\mesh)^{-\frac{1}{2}}$.}
\end{quotation}

\begin{remark}
For each discrete domain $(\O^\mesh_\DS;a^\mesh,b^\mesh)$, the discrete boundary value problem
{(A)\&(B)\&(C)} has a unique solution.
\end{remark}
\begin{proof} Existence of the solution is given by the explicit construction of the
holomorphic fermion in the FK-Ising model. Concerning uniqueness, let $F_1^\mesh$ and
$F_2^\mesh$ denote two different solutions. Then the difference $F^\mesh:=F_1^\mesh-F_2^\mesh$
is s-holomorphic in $\O^\mesh_\DS$, thus $H^\mesh:=\Im\int^\mesh (F^\mesh(z))^2d^\mesh z$ is
well-defined (see Section \ref{SectIntF^2}, especially (\ref{HDef})). Due to condition (B),
$H^\mesh$ is constant on both boundary arcs $(a^\mesh b^\mesh)\ss\G^*$ and $(b^\mesh
a^\mesh)\ss\G$. Moreover, in view of the same normalization of $F_{1,2}^\mesh$ near $b^\mesh$,
one can fix an additive constant so that $H^\mesh_{\G^*}=0$ on $(a^\mesh b^\mesh)$ and
$H^\mesh_\G=0$ on $(b^\mesh a^\mesh)$.

The ``boundary modification trick'' described in Sect. \ref{SectBModTrick} provides us the
slight modification of $\O^\mesh_\DS$ (see Fig.~\ref{Fig:BoundaryTrick}B) such that the
Dirichlet boundary conditions $H^\mesh_\G=0$, $H^\mesh_{\G^*}=0$ hold true \emph{everywhere}
on $\pa\O^\mesh_\L$. Using sub-/super-harmonicity of $H^\mesh$ on $\G/\G^*$ and (\ref{HDef}),
we arrive at $0\ge H^\mesh_\G\ge H^\mesh_{\G^*}\ge 0$ in $\O^\mesh$. Thus, $H^\mesh\equiv 0$
and $F_1^\mesh\equiv F^\mesh_2$.
\end{proof}

Let
\begin{equation}
\label{Hmesh=Int} H^\mesh=H^\mesh_{(\O^\mesh;a^\mesh,b^\mesh)}:= \Im\int^\mesh
(F^\mesh_{(\O^\mesh;a^\mesh,b^\mesh)}(z))^2 d^\mesh z.
%,\qquad H^\mesh_\G:=H^\mesh\big|_{\O^\mesh_\G},~~H^\mesh_{\G^*}:=H^\mesh\big|_{\O^\mesh_{\G^*}}.
\end{equation}
It follows from the boundary conditions (B) that $H^\mesh$ is constant on both boundary arcs
$(a^\mesh b^\mesh)\ss\G^*$ and $(b^\mesh a^\mesh)\ss\G$. In view of the chosen normalization
(C), we have
\[
H^\mesh_\G|_{(b^\mesh a^\mesh)} -H^\mesh_{\G^*}|_{(a^\mesh b^\mesh)}=1.
\]
\begin{remark}
Due to the ``boundary modification trick'' (Section \ref{SectBModTrick}), one can fix an
additive constant so that
\begin{equation}
\label{HmeshBV}
\begin{array}{lllll}
H^\mesh_\G=0 &\mathit{on}~(a^\mesh b^\mesh)_{\wt\G}, &\quad& H^\mesh_{\G^*}=0 & \mathit{on}\
(a^\mesh b^\mesh), \cr\vphantom{\big|^|} H^\mesh_\G=1 &\mathit{on}~(b^\mesh a^\mesh), &\quad&
H^\mesh_{\G^*}=1 &\mathit{on}~(b^\mesh a^\mesh)_{\wt{\G}^*},
\end{array}
\end{equation}
where $(a^\mesh b^\mesh)_{\wt{\G}}$ (and, in the same way, $(b^\mesh a^\mesh)_{\wt{\G}^*}$)
denotes the set of newly constructed ``black'' vertices near the ``white'' boundary arc
$(a^\mesh b^\mesh)$ (see Fig.~\ref{Fig:BoundaryTrick}B).
\end{remark}

Let $f^\mesh(z)=f^\mesh_{(\O^\mesh;a^\mesh,b^\mesh)}(z)$ denote the solution of the
corresponding continuous boundary value problem inside the \emph{polygonal} domain $\O^\mesh$:

\begin{quotation}
{\it

\noindent {\bf (a) holomorphicity:} $f^\mesh$ is holomorphic in $\O^\mesh$;

\noindent {\bf (b) boundary conditions:} $f^\mesh(\z)\parallel (\t(\z))^{-\frac{1}{2}}$ for
$\z\in\pa\O^\mesh$, where $\t(\z)$ denotes the tangent vector to $\pa\O^\mesh$ oriented from
$a^\mesh$ to $b^\mesh$ (on both arcs);

\noindent {\bf (c) normalization:} the function
$h^\mesh=h^\mesh_{\O^\mesh,a^\mesh,b^\mesh}:=\Im \int (f^\mesh(\z))^2d\z$ is uniformly bounded
in $\O^\mesh$ and
\[
h^\mesh|_{(a^\mesh b^\mesh)}=0,\qquad h^\mesh|_{(b^\mesh a^\mesh)}=1.
\]
%\begin{center}
%$f^\mesh(z)\sim (\pi(b^\mesh\!-\!z))^{-\frac{1}{2}}$ as $z\to b^{\mesh}$\quad and\quad
%$f^\mesh(z)\sim (\pi(z\!-\!a^\mesh))^{-\frac{1}{2}}$ as $z\to a^{\mesh}$.
%\end{center}
}
\end{quotation}
Note that (a) and (b) guarantee that $h^\mesh$ is harmonic in $\O^\mesh$ and constant on both
boundary arcs $(a^\mesh b^\mesh)$, $(b^\mesh a^\mesh)$. In other words,
\[
f^\mesh=\sqrt{2i\pa h^\mesh},\qquad h^\mesh=\hm {\,\cdot\,}{b^\mesh a^\mesh}{\O^\mesh},
\]
where $\o$ denotes the (continuous) harmonic measure in the (polygonal) domain $\O^\mesh$.
Note that $\pa h^\mesh \ne 0$ in $\O^\mesh$, since $h^\mesh$ is the imaginary part of the
conformal mapping from $\O^\mesh$ onto the infinite strip $(-\infty,\infty)\times(0,1)$
sending $a^\mesh$ and $b^\mesh$ to $\mp\infty$, respectively. Thus, $f^\mesh$ is well-defined
(up to the sign).

\begin{theorem}[\bf convergence of FK-observable]
\label{ThmFKConvergence} The solutions $F^\mesh$ of the discrete Riemann-Hilbert boundary value
problems {(A)\&(B)\&(C)} are uniformly close in the bulk to their continuous counterpart
$f^\mesh$ defined by {(a)\&(b)\&(c)}. Namely, for all $0<r<R$ there exists
$\ve(\mesh)=\ve(\mesh,r,R)$ such that for all discrete domains
$(\O^\mesh_\DS;a^\mesh,b^\mesh)$ and $z^\mesh\in \O^\mesh_\DS$ the following holds true:
\begin{center}
if $B(z^\mesh,r)\ss\O^\mesh\ss B(z^\mesh,R)$, then $|F^\mesh(z^\mesh)-f^\mesh(z^\mesh)|\le
\ve(\mesh)\to 0$ as $\mesh\to 0$ $\vphantom{\big|^|_|}$
\end{center}
(for a proper choice of $f^\mesh$'s sign), uniformly with respect to the shape of $\O^\mesh$
and $\DS^\mesh$.
\end{theorem}

\begin{remark}
\label{RemTheSameSign} Moreover, the sign of $f^\mesh$ is the same for, at least, all
$\wt{z}^\mesh$ lying in the same connected component of the $r$-interior of $\O^\mesh$.
\end{remark}

\begin{proof}
Assume that neither $f^\mesh$ nor $-f^\mesh$ approximates $F^\mesh$ well, and so for both
signs $|F^\mesh(z^\mesh)\!\pm\!f^\mesh(z^\mesh)|\ge \ve_0\!>\!0$ for some sequence of domains
$\O^\mesh$, $\mesh\!\to\!0$. Applying translations one can without loss of generality assume
$z^\mesh=0$ for all $\mesh$'s. The set of all simply-connected domains $\O: B(0,r)\ss\O\ss
B(0,R)$ is compact in the Carath\'eodory topology (of convergence of conformal maps germs).
Thus, passing to a subsequence, we may assume that
\[
(\O^\mesh;a^\mesh,b^\mesh)~\CaraTo~(\O;a,b)\quad \mathrm{as}~~\mesh\to 0
\]
(with respect to $0\!=\!z^\mesh$). Let $h=h_{(\O;a,b)}:=\hm {\,\cdot\,}{ba}{\O}$. Note that
$h^\mesh\rra h$ as $\mesh\to 0$, uniformly on compact subsets of $\O$, since the harmonic
measure is Carath\'eodory stable. Moreover,
\[
(f^\mesh)^2= {2i\pa h^\mesh} \rra f^2= {2i\pa h}\quad \mathrm{as}~~\mesh\to 0.
\]
We are going to prove that, at the same time,
\[
H^\mesh\rra h\quad \mathrm{and}\quad (F^\mesh)^2\rra f^2\quad \mathrm{as}~~\mesh\to 0,
\]
uniformly on compact subsets of $\O$, which gives a contradiction.

It easily follows from (\ref{HmeshBV}) and the sub-/super-harmonicity of $H^\mesh$ on
$\G/\G^*$ that
\[
0\le H^\mesh \le 1~~\mathrm{everywhere~in}~~\O^\mesh_\L.
\]
In view of Theorem \ref{FboundH}, this (trivial) uniform bound implies the uniform boundedness
and the equicontinuity of functions $F^\mesh$ on compact subsets $K$ of $\O$. Thus, both
$\{H^\mesh\}$ and $\{F^\mesh\}$ are normal families on each compact subset of $\O$. Therefore,
taking a subsequence, we may assume that
\[
F^\mesh\rra F~~\mathrm{and}~~H^\mesh\rra H\quad \mathrm{for~some}~~F:\O\to\C,\
H:\O\to\R,
\]
uniformly on all compact subsets of $\O$. The simple passage to the limit in (\ref{Hmesh=Int})
gives
\[
 H(v_2)-H(v_1)=\Im \int_{[v_1;v_2]} (F(\z))^2d\z,
\]
for each segment $[v_1;v_2]\ss\O$. Thus, $F^2=2i\pa H$. %, and it is sufficient to show that $H=h$.
Being a limit of discrete subharmonic functions
$H^\mesh_\G$, as well as discrete superharmonic functions $H^\mesh_{\G^*}$, the function $H$
should be harmonic. The sub-/super-harmonicity of $H^\mesh$ on $\G/\G^*$ gives
\[
\dhm\mesh{\,\cdot\,}{(b^\mesh a^\mesh)_{\wt{\G}^*}}{\O^\mesh_{\G^*}}\le H^\mesh_{\G^*}\le
H^\mesh_\G \le \dhm\mesh{\,\cdot\,}{(b^\mesh a^\mesh)_\G}{\O^\mesh_\G}\quad \mathrm{in}~\
\O^\mesh_\L,
\]
where the middle inequality holds for any pair of neighbors $w\in\G^*$, $u\in\G$
due~to~(\ref{HDef}). It is known (see \cite{chelkak-smirnov-dca} Theorem 3.12) that both
discrete harmonic measures $\o^\mesh(\cdot)$ (as on $\G^*$, as on $\G$) are uniformly close in
the bulk to the continuous harmonic measure $\o(\cdot)=h$. Thus, $H^\mesh\rra h$ uniformly on
compact subsets of $\O$, and so $F^2={2i\pa h}$.
\end{proof}

\begin{proof}[Proof of Remark \ref{RemTheSameSign}]
Consider simply-connected domains $(\O^\mesh;a^\mesh,b^\mesh;z^\mesh,\wt{z}^\mesh)$, with
$z^\mesh, \wt{z}^\mesh$ lying in the same connected components of the $r$-interiors
$\O^\mesh_r$. Assume that we have $|F^\mesh(z^\mesh)-f^\mesh(z^\mesh)|\to 0$ but
$|F^\mesh(\wt{z}^\mesh)+f^\mesh(\wt{z}^\mesh)|\to 0$ as $\mesh\to 0$. Applying translations to
$\O^\mesh$ and taking a subsequence, we may assume that
\[
(\O^\mesh;a^\mesh,b^\mesh;\wt{z}^\mesh)\CaraTo (\O;a,b;\wt{z})\quad \mathrm{w.r.t.}~\
0=z^\mesh,
\]
for some $\wt{z}$ connected with $0$ inside the $r$-interior $\O_r$ of $\O$. As it was shown
above,
\[
F^\mesh\rra F\quad \mathrm{and}\quad f^\mesh\rra f\quad \mathrm{uniformly~on}~\O_r,
\]
where either $F\equiv f$ or $F\equiv -f$ (everywhere in $\O$), which gives a contradiction.
\end{proof}

\setcounter{equation}{0}
\section{Uniform convergence for the holomorphic observable in the spin-Ising model}
\label{SectSpinObservable}

For the spin-Ising model, the discrete holomorphic fermion
$F^\mesh(z)=F^\mesh_{(\O^\mesh;a^\mesh,b^\mesh)}(z)$ constructed in Section~\ref{SectSpinIsing}
satisfies the following discrete boundary value problem:

\begin{quotation} {\it
\noindent {\bf (A${}^\circ$) Holomorphicity:}  $F^\mesh(z)$ is s-holomorphic inside $\O^\mesh_\DS$.

\noindent {\bf (B${}^\circ$) Boundary conditions:} $F^\mesh(\z)\parallel (\t(\z))^{-\frac{1}{2}}$ for
all $\z\in\pa \O^\mesh_\DS$ except at $a^\mesh$, where $\t(\z)=w_2(\z)\!-\!w_1(\z)$ is the
``discrete tangent vector'' to $\pa\O^\mesh_\DS$ oriented in the counterclockwise direction
(see Fig.~\ref{Fig:SpinDomain}A).

\noindent {\bf (C${}^\circ$) Normalization at $\bm{b^\mesh}$:}
$F^\mesh(b^\mesh)=\cF^\mesh(b^\mesh)$, where the normalizing constants
$\cF^\mesh(b^\mesh)\parallel (\t(b^\mesh))^{-\frac{1}{2}}$ are defined in
Section~\ref{SectBoundHarnack}. }
\end{quotation}

\begin{remark}
For each discrete domain $(\O^\mesh_\DS;a^\mesh,b^\mesh)$, the discrete boundary value problem
{(A${}^\circ$)\&(B${}^\circ$)\&(C${}^\circ$)} has a unique solution.
\end{remark}
\begin{proof}
Existence is given by the holomorphic fermion in the spin-Ising model. Concerning uniqueness,
let $F^\mesh$ denote some solution. Then $H^\mesh=\int^\mesh (F^\mesh(z))^2 d^\mesh z$ is
constant on $\pa\O^\mesh_{\G^*}$, so either $F^\mesh(a^\mesh)\parallel (\t(a^\mesh))^{-1/2}$
or $F^\mesh(a^\mesh)\parallel (-\t(a^\mesh))^{-1/2}$. In the former case, using the ``boundary
modification trick'' (Section \ref{SectBModTrick}), we arrive at $H^\mesh=0$ on both
$\pa\O^\mesh_{\G^*}$ and $\pa\O^\mesh_{\wt{\G}}$. Then, sub-/super-harmonicity of $H^\mesh$ on
$\G/\G^*$ and (\ref{HDef}) imply that
$0\ge H^\mesh_\G\ge H^\mesh_{\G^*}\ge 0$ in $\O^\mesh$.
Therefore, $H^\mesh\equiv 0$ and $F^\mesh\equiv 0$, which is impossible. Thus,
$F^\mesh(a^\mesh)\parallel (-\t(a^\mesh))^{-1/2}$ (for the holomorphic fermion this follows
from the definition).

Let $F_{1,2}$ be two different solutions. Denote
\[
F^\mesh(z):=
(-\t(a^\mesh))^\frac{1}{2}\cdot[F_1^\mesh(a^\mesh)F_2^\mesh(z)-F_2^\mesh(a^\mesh)F_1^\mesh(z)].
\]
Then, $F^\mesh$ is s-holomorphic and, since
$(-\t(a^\mesh))^\frac{1}{2}F_{1,2}^\mesh(a^\mesh)\in\R$, satisfies the boundary condition (B${}^\circ$)
on $\pa\O^\mesh\setminus\{a^\mesh\}$. Moreover, we also have $F^\mesh(a^\mesh)=0\parallel
(\t(a^\mesh))^{-\frac{1}{2}}$. Arguing as above, we obtain $F^\mesh\equiv 0$.
The identity $F_1^\mesh\equiv F^\mesh_2$ then follows from (C${}^\circ$).
\end{proof}

Let
\[
H^\mesh=H^\mesh_{(\O^\mesh;a^\mesh,b^\mesh)}:=\Im\int^\mesh
(F^\mesh_{(\O^\mesh;a^\mesh,b^\mesh)}(z))^2 d^\mesh z.
\]
\begin{remark}
Using Section \ref{SectBModTrick}, one can fix an additive constant so that
\begin{equation}
\label{HmeshBV-spin}
\begin{array}{c}
H^\mesh_{\G^*}=0~\mathit{everywhere~on}~\pa\O^\mesh_{\G^*}\cr H^\mesh_\G=0\
\mathit{everywhere~on}~\pa\O^\mesh_{\wt{\G}}~\mathit{except}~a^\mesh_{\mathrm{out}},
\end{array}
\end{equation}
where $\pa\O^\mesh_{\wt{\G}}$ denotes the modified boundary (everywhere except $a^\mesh$) and
$a^\mesh_{\mathrm{out}}$ is the original outward ``black'' vertex near $a^\mesh$ (see
Fig.~\ref{Fig:BoundaryTrick}C). Then, $H^\mesh_\G\ge H^\mesh_{\G^*}\ge 0$ everywhere in
$\O^\mesh$.
\end{remark}

\subsection{Boundary Harnack principle and solution to (A${}^\circ$)\&(B${}^\circ$) in the discrete
half-plane} \label{SectBoundHarnack} We start with a version of the Harnack Lemma
(Proposition~\ref{HarnackIntF^2}) which compares the values of $H^\mesh=\Im\int^\mesh
(F^\mesh)^2(z)d^\mesh z$ in the bulk with its normal derivative at the boundary.

Let $R(s,t):=(-s;s)\ts(0;t)\ss\C$ be an open rectangle, $R^\mesh_\DS(s,t)\ss\G$ denote its
discretization, and $L^\mesh(s)$, $U^\mesh(s,t)$ and $V^\mesh(s,t)$ be the lower, upper and
vertical parts of the boundary $\pa R^\mesh_\G$ (see Fig. \ref{Fig:OsupsetR}A).

\begin{proposition}
\label{PropBoundHar} Let $t \ge \mesh$, $F^\mesh$ be an s-holomorphic function in a discrete
rectangle $R^\mesh_\DS(2t,2t)$ satisfying the boundary condition (B${}^\circ$) on the lower
boundary $L^\mesh_\DS(2t)$ and $H^\mesh=\Im\int^\mesh(F^\mesh(z))^2d^\mesh z$ be defined by
(\ref{HDef}) so that $H=0$ on $L^\mesh_{\G^*}$ and $H\ge 0$ everywhere in
$R^\mesh_{\G^*}(2t,2t)$. Let $b^\mesh\in\DS$ be the boundary vertex closest to $0$, and
$c^\mesh\in\G^*$ denote the inner face (dual vertex) containing the point $c=it$. Then,
uniformly in $t$ and $\mesh$,
\[
|F^\mesh(b^\mesh)|^2\asymp \frac{H^\mesh(c^\mesh)}{t},
\]
\end{proposition}

\begin{proof}
Let $t\ge\const\cdot \mesh$ (the opposite case is trivial). It follows from
Remark~\ref{RemHGasympG*} and Proposition~\ref{HarnackIntF^2} that all the values of $H^\mesh$
on $U^\mesh(t,\tfrac{1}{2}t)$ are uniformly comparable with $H(c^\mesh)$. Then, the
superharmonicity of $H^\mesh\big|_{\G^*}$ and simple estimates of the discrete harmonic
measure in $R^\mesh_{\G^*}(t,\frac{1}{2}t)$ (see Lemma~\ref{LemmaDhmR}) give
\[
H^\mesh(b^\mesh_{\G^*})\ge
\dhm\mesh{b^\mesh_{\G^*}}{U^\mesh_{\G^*}}{R^\mesh_{\G^*}(t,\tfrac{1}{2}t)} \cdot
\min\nolimits_{w^\mesh\in U^\mesh_{\G^*}}H^\mesh(w^\mesh)\ge \const\cdot {\mesh}/{t}\cdot
H^\mesh(c^\mesh),
\]
where $b^\mesh_{\G^*}\in\G^*$ denotes the inner dual vertex closest to $b^\mesh$ (see
Fig.~\ref{Fig:OsupsetR}A). Therefore, $|F^\mesh|^2\ge\const\cdot H^\mesh(c^\mesh)/t$ for a
neighbor of $b^\mesh$. Due to the s-holomorphicity of $F^\mesh$, this is sufficient to
conclude that
\[
|F^\mesh(b^\mesh)|^2\ge \const\cdot H^\mesh(c^\mesh)/t.
\]

On the other hand, since $H=0$ on $L^\mesh(2t)$, one has $H^\mesh\le\const\cdot
H^\mesh(c^\mesh)$ \mbox{\emph{everywhere}} in $R^\mesh_\G(t,\frac{1}{2}t)$ (the proof mimics
the proof of Proposition~\ref{HarnackIntF^2}: if $H^\mesh(v)\gg H^\mesh(c^\mesh)$ at some
$v\in R^\mesh_\G(t,\frac{1}{2}t)$, then, since $H^\mesh\equiv 0$ on $L^\mesh(2t)$, the same
holds true along some path running from $v$ to $U^\mesh(2t,2t)\cup V^\mesh(2t,2t)$, which
gives a contradiction). Thus, estimating the discrete harmonic measure in
$R^\mesh_\G(t,{\textstyle \frac{1}{2}t})$ from the inner vertex $b^\mesh_\G\in \G$ closest to
$b$ (see Fig.~\ref{Fig:OsupsetR}A), we arrive at
\[
H^\mesh(b^\mesh_{\G})\le \dhm\mesh{b^\mesh_{\G}}{U^\mesh_{\G}\cup
V^\mesh_{\G}}{R^\mesh_{\G}(t,\tfrac{1}{2}t)} \cdot \max\nolimits_{u^\mesh\in U^\mesh_{\G}\cup
B^\mesh_{\G}}H^\mesh(u^\mesh) \le \const\cdot \mesh/t\cdot H^\mesh(c^\mesh).
\]
Since $H^\mesh(b^\mesh_{\G})\asymp \mesh\cdot |F^\mesh(b^\mesh)|^2$, this means
$|F^\mesh(b^\mesh)|^2\le \const\cdot H^\mesh(c^\mesh)/t$.
\end{proof}

Now we are able to construct a special solution $\cF^\mesh$ to the discrete boundary value
problem (A${}^\circ$)\&(B${}^\circ$) in the discrete half-plane. The value
$\cF^\mesh(b^\mesh)$ will be used later on for the normalization of the spin-observable at the
target point $b^\mesh$.

\begin{theorem}
\label{ThmCFexists} Let $\H^\mesh$ denote some discretization of the upper half-plane (see
Fig.~\ref{Fig:OsupsetR}A). Then, there exist a unique \mbox{s-holomorphic} function
$\cF^\mesh:\H^\mesh_\DS\to\C$ satisfying boundary conditions $\cF^\mesh(\z)\parallel
(\t(\z))^{-\frac{1}{2}}$ for $\z\in\pa\H^\mesh_\DS$, such that
\[
\cF^\mesh(z)=1+O(\mesh^{\frac{1}{2}}\cdot(\Im z)^{-\frac{1}{2}}),
%\quad\mathit{and}\quad \cH^\mesh(v)=\Im v + O(\mesh),
\]
uniformly with respect to $\DS^\mesh$. Moreover, $|\cF^\mesh|\asymp 1$ on the boundary
$\pa\H^\mesh_\DS$.
\end{theorem}

\begin{remark}
If $\pa\H^\mesh_\DS$ is a straight line (e.g., for the proper oriented square or
triangular/hexagonal grids), then $\cF^\mesh\equiv 1$ easily solves the problem.
\end{remark}

\begin{proof}
\emph{Uniqueness.} Let $F^\mesh_1$, $F^\mesh_2$ be two different solutions. Clearly,
$F^\mesh:=F_1^\mesh\!-\!F_2^\mesh$ is \mbox{s-holomorphic} and satisfies the same boundary
conditions on $\pa\H^\mesh_\DS$. Thus we can set $H^\mesh:=\Im\int^\mesh(F^\mesh(z))^2d^\mesh
z$, where $H^\mesh=0$ on both $\pa \H^\mesh_{\G^*}$ and $\pa\H^\mesh_{\wt{\G}}$ (see
Section~\ref{SectBModTrick}).

Since $(F^\mesh(z))^2=O(\mesh\cdot(\Im z)^{-1})$, the integration over ``vertical'' paths
gives
\[
H^\mesh(v)=O(\mesh\cdot \log(\mesh^{-1}\Im v))\quad \mathrm{as}\quad \Im v\to\infty,
\]
so $H^\mesh$ grows sublinearly as $\Im v\to\infty$ which is impossible. Indeed, using simple
estimates of the discrete harmonic measure (see Lemma \ref{LemmaDhmR}) in big rectangles
$R(2n,n)$, $n\to\infty$, and sub-/super-harmonicity of $H^\mesh$ on $\G$/$\G^*$ together with
the Dirichlet boundary conditions on the boundary $\pa\H^\mesh_\L$, we conclude that
\[
H^\mesh_\G(u)\le \lim_{n\to \infty} O(\mesh\cdot \log(\mesh^{-1}n))\cdot (\Im
u\!+\!2\mesh)n^{-1}=0\quad\mathrm{for~any}~~u\in\H^\mesh_\G,
\]
and, similarly, $H^\mesh_{\G^*}(w)\ge 0$ for any $w\in\H^\mesh_{\G^*}$. Thus, $H^\mesh\equiv
0$, and so $F^\mesh_1\equiv F^\mesh_2$.

\smallskip

\noindent \emph{Existence.} We construct $\cF^\mesh$ as a (subsequential) limit of holomorphic
fermions in \emph{increasing} discrete rectangles. Let \emph{$\mesh$ be fixed},
$b^\mesh\in\DS$ denote the closest to $0$ boundary vertex, and $R^\mesh_n$ denote
discretizations (see Fig.~\ref{Fig:OsupsetR}A) of the rectangles
\[
R_n=R(4n,2n):=(-4n;4n)\ts (0;2n).
\]
Let $F_n^\mesh:R^\mesh_{n,\DS}\to\C$ be the discrete s-holomorphic fermion solving the
boundary value problem (A${}^\circ$)\&(B${}^\circ$) in $R^\mesh_n$ with $a^\mesh_n$ being the
discrete approximations of the points $2ni$. For the time being, we normalize $F_n^\mesh$ by
the condition
\[
|F^\mesh(b^\mesh_n)|=1.
\]
Having this normalization, it follows from the discrete Harnack principle
(Propositions~\ref{HarnackIntF^2} and \ref{PropBoundHar}) that $H_n^\mesh\asymp n$ everywhere
near the segment $[-2n\!+\!in;2n\!+\!in]$. Moreover, since $H=0$ on the lower boundary, one
also has $H\le \const \cdot n$ everywhere in the smaller rectangle $R^\mesh_\DS(2n,n)$ (the
proof mimics the proof of Proposition~\ref{HarnackIntF^2}). Thus, estimating the discrete
harmonic measure of $U^\mesh(2n,n)_\G\cup V^\mesh_\G(2n,n)$ in $R^\mesh_\G(2n,n)$ from any
\emph{fixed} vertex $v^\mesh\in \H^\mesh_\L$ and using the subharmonicity of $H\big|_\G$, one
obtains
\[
H^\mesh_n(v^\mesh)\le \frac{\Im v^\mesh\!+\!2\mesh}{n}\cdot \const\cdot n\le\const\cdot (\Im
v^\mesh\!+\!2\mesh),
\]
if $n=n(v^\mesh)$ is big enough. Moreover, since $H^\mesh\big|_{\G^*}$ is superharmonic, one
also has the inverse estimate for $v^\mesh$ near the imaginary axis $i\R_+$:
\[
H^\mesh_n(v^\mesh)\ge \const\cdot (\Im v^\mesh\!+\!2\mesh),\quad\mathrm{if~}|\Re
v^\mesh|\le\mesh.
\]
Further, Theorem~\ref{FboundH} applied in $(\Re z-\tfrac{1}{2}\Im z;\Re z+\tfrac{1}{2}\Im
z)\ts (\tfrac{1}{2}\Im z;\tfrac{3}{2}\Im z)$ gives $|F^\mesh_n(z^\mesh)|\le\const$ for any
$z^\mesh\in H^\mesh_\DS$, if $n=n(z^\mesh)$ is big enough.

Note that there are only countably many points $v^\mesh\in\H^\mesh_\L$ and
$z^\mesh\in\H^\mesh_\DS$. Since for any fixed vertex the values $H^\mesh_n(v^\mesh)$ and
$F^\mesh_n(z^\mesh)$ are bounded, we may choose a subsequence $n=n_k\to\infty$ so that
\[
H_n^\mesh(v^\mesh)\to\cH^\mesh(v^\mesh)~~\mathrm{and}~~F_n^\mesh(z^\mesh)\to
\cF^\mesh(z^\mesh)\quad \mathrm{for\
each}~~v^\mesh\in\H^\mesh_\L~~\mathrm{and}~~z^\mesh\in\H^\mesh_\DS,
\]
It's clear that $\cF^\mesh:\H^\mesh_\DS\to\C$ is s-holomorphic,
$\cH^\mesh=\Im\int^\mesh(\cF^\mesh(z))^2d^\mesh z\ge 0$, $\cH^\mesh$ and $\cF^\mesh$ satisfy
the same boundary conditions as $H^\mesh_n$, $F^\mesh_n$, and $\cF^\mesh(b^\mesh)=1$.
Moreover,
\begin{equation}
\label{xHFestim}
\begin{aligned}
\cH^\mesh(v^\mesh)& =O(\Im v^\mesh\!+\!2\mesh),\quad\cF^\mesh(z^\mesh)=O(1)
~~\mathrm{uniformly~in}~\H^\mesh, \cr &\mathrm{and}~~\cH^\mesh(v^\mesh)\asymp (\Im
v^\mesh\!+\!2\mesh)~~\mathrm{for}~~v^\mesh~~\mathrm{near}~~i\R_+.
\end{aligned}
\end{equation}

Now we are going to improve this estimate and show that, uniformly in $H^\mesh_\L$,
$\cH^\mesh(v)=\mu\cdot(\Im v+ O(\mesh))$ for some $\mu>0$. For this purpose, we \emph{re-scale
our lattice and functions} by a small factor $\ve\to 0$. Let
\[
v^{\ve\mesh}:=\ve v^\mesh\,,\quad z^{\ve\mesh}:=\ve z^\mesh\,,\quad \mathrm{and}\quad
\cH^{\ve\mesh}(v^{\ve\mesh}):=\ve\cH^\mesh(v^\mesh)\,,\quad
\cF^{\ve\mesh}(z^{\ve\mesh}):=\cF^\mesh(z^\mesh).
\]
Note that the uniform estimates (\ref{xHFestim}) remains valid for the re-scaled functions.
Therefore, Theorem~\ref{FboundH} guarantees that the functions $\cH^{\ve\mesh}$ and
$\cF^{\ve\mesh}$ are uniformly bounded and equicontinuous on compact subsets of $\H$, so,
taking a subsequence, we may assume
\[
\cH^{\ve\mesh}(v)\rra h(v)\quad \mathrm{uniformly~on~compact~subsets~of}~\H.
\]
Being a limit of discrete subharmonic functions $H^{\ve\mesh}_\G$ as well as discrete
superharmonic functions $H^{\ve\mesh}_{\G^*}$, the function $h$ is harmonic. Moreover, it is
nonnegative and, due to (\ref{xHFestim}), has zero boundary values everywhere on $\R$. Thus,
$h(v)\equiv \mu v$ for some $\mu>0$ (the case $\mu=0$ is excluded by the uniform double-sided
estimate of $\cH^{\ve\mesh}$ near $i\R_+$).

Thus, for any fixed $s\gg t>0$ one has
\[
\cH^{\ve\mesh}\rra \mu t~~~\mathrm{on~~~}[-s\!+\!it;s\!+\!it]~~~\mathrm{as}~~~\ve\to 0.
\]
For the original function $H^\mesh$, this means
\[
\cH^\mesh(v^\mesh)= (\mu+o_{k\to\iy}(1))\cdot \Im v^\mesh~~~
\mathrm{uniformly~~~on~~~}U^\mesh(ks,kt)~~~\mathrm{as}~~~k=\ve^{-1}\to\infty.
\]

Estimating the discrete harmonic measure in (big) rectangles $R^\mesh(ks,kt)$ from a fixed
vertex $v^\mesh$ (see Lemma~\ref{LemmaDhmR}) and using subharmonicity of $H\big|_\G$ and
superharmonicity of $H\big|_{\G^*}$, we obtain
\[
\cH^\mesh(v^\mesh)=\lt[\frac{\Im v^\mesh+O(\mesh)}{kt}+O\lt(\frac{|v^\mesh|\cdot kt}{(ks)^2}
\rt)\rt]\cdot(\mu+o_{k\to\iy}(1))\cdot kt+O\lt(\frac{|v^\mesh|\cdot kt}{(ks)^2}\rt)\cdot
O(kt),
\]
where the $O$-bounds are uniform in $v^\mesh$, $t$ and $s$, if $k=k(v^\mesh)$ is big enough.
Passing to the limit as $k\to\infty$, we arrive at
\[
\cH^\mesh(v^\mesh)=\mu\cdot(\Im v^\mesh+O(\mesh))+O(|v^\mesh|\cdot t^2/s^2),
\]
where the $O$-bound is uniform in $v^\mesh$ and $t,s$. Therefore,
\[
\cH^\mesh(v^\mesh)=\mu\cdot(\Im v^\mesh+O(\mesh))~~~\mathrm{uniformly~~in}~~\H^\mesh.
\]

It follows from (\ref{xHFestim}) and Theorem~\ref{FboundH} that both $\cF^\mesh$ and
$(\cF^\mesh)^2$ are uniformly Lipschitz in each strip $\b\le \Im \z\le 2\b$ with the Lipschitz
constant bounded by $O(\b^{-1})$. Taking some $v\in\H^\mesh_\L$ near $z$ and
$v'\in\H^\mesh_\L$ such that $|v'-v|\asymp \mesh^{1/2}(\Im z)^{1/2}$, we obtain
\[
\cH^\mesh(v')-\cH^\mesh(v)=\Im\int^\mesh_{[v,v']} (F^\mesh(\z))^2 d^\mesh \z =
\Im[(\cF^\mesh(z))^2(v'\!-\!v)] +O\lt(\frac{|v'\!-\!v|^2}{\Im z}\rt),
\]
i.e., $\Im[(\cF^\mesh(z))^2(v'\!-\!v)]=\mu\cdot\Im(v'\!-\!v)+O(\mesh)$ for all $v'$. Thus,
\[
(\cF^\mesh(z))^2=\mu+O({\mesh^{\frac{1}{2}}}\cdot {(\Im
z)^{-\frac{1}{2}}})~~\mathrm{uniformly~in}~\H^\mesh_\DS.
\]
Since $\cF^\mesh$ is Lipschitz with the Lipschitz constant bounded by $O((\Im z)^{-1})$ (see
above), this allows us to conclude that
\[
\pm\cF^\mesh(z)=\mu^{\frac{1}{2}}+O(\mesh^{\frac{1}{2}}\cdot(\Im
z)^{-\frac{1}{2}})~~\mathrm{uniformly~in}~\H^\mesh_\DS
\]
(for some choice of the sign). Thus, the function $\wt{\cF}^\mesh:=\pm\mu^{-1/2}\cF^\mesh$
satisfies the declared asymptotics and boundary conditions. Moreover,
$|\wt{\cF}^\mesh(b^\mesh)|=\mu^{-1/2}\asymp 1$. Since such a function $\wt{\cF}^\mesh$ is
unique, all other values $|\wt{\cF}^\mesh|$ on the boundary $\pa\H^\mesh_\DS$ are $\asymp 1$
too.
\end{proof}

\begin{figure}
\centering{\begin{minipage}[b]{0.4\textwidth}
\centering{\includegraphics[width=\textwidth]{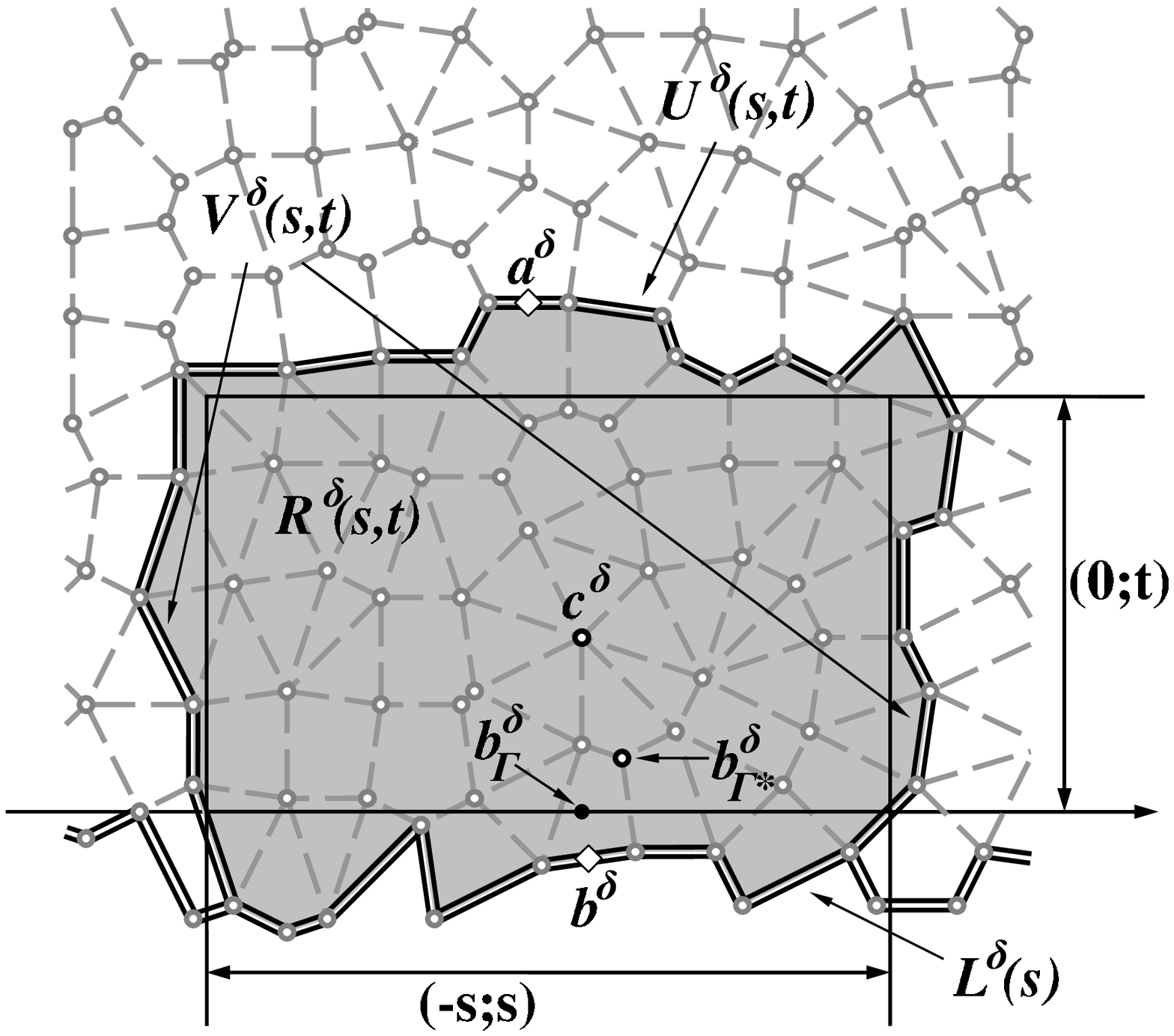}}

\bigskip

\textsc{(A)}
\end{minipage}
\hskip 0.05\textwidth
\begin{minipage}[b]{0.5\textwidth}
\centering{\includegraphics[width=\textwidth]{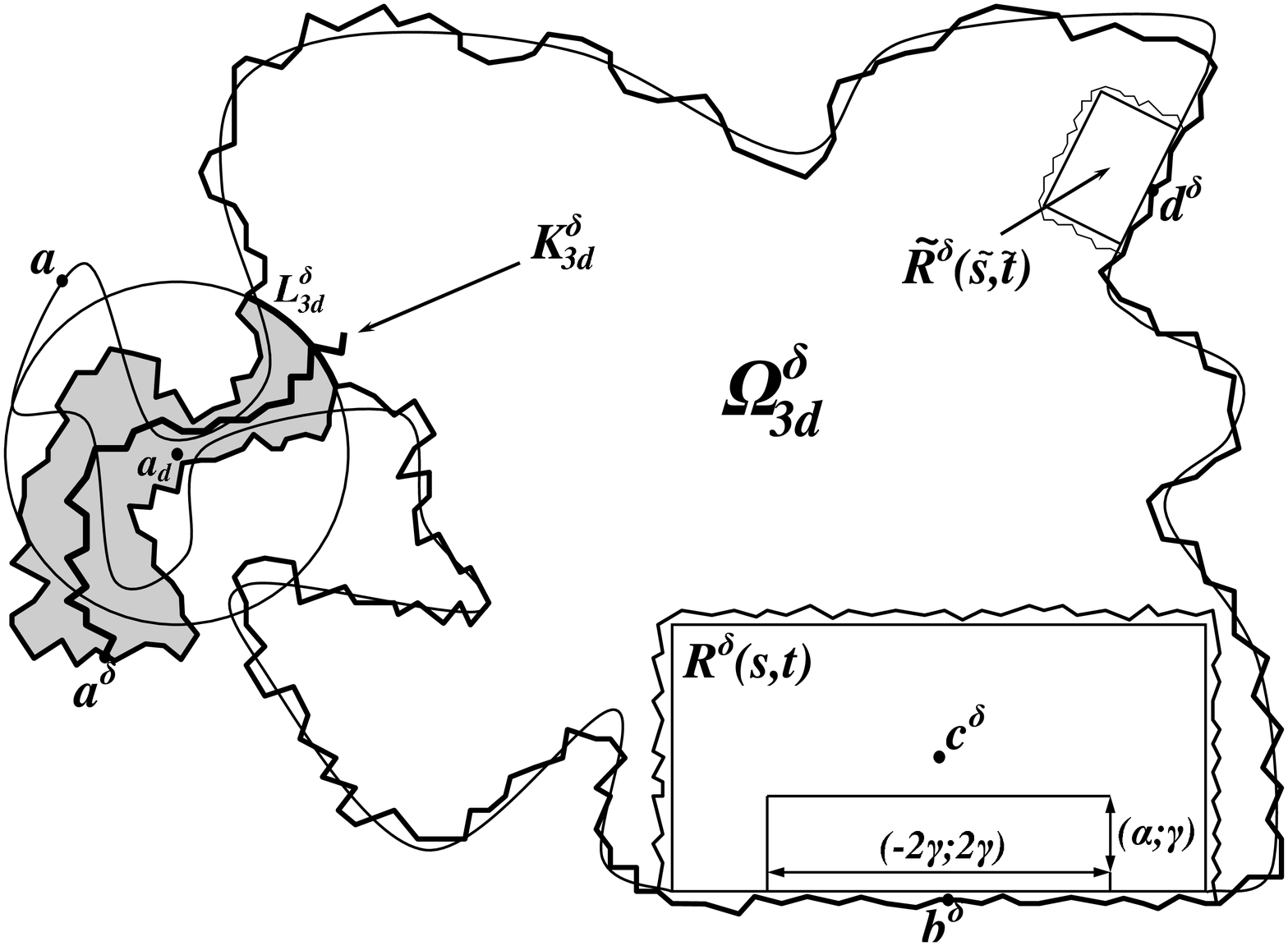}}

\bigskip

\textsc{(B)}
\end{minipage}}
\caption{\label{Fig:OsupsetR} \textsc{(A)} Discretizations of the upper half-plane $\H$ and
the rectangle $R(s,t)=(-s;s)\ts(0;t)$. Boundary points $a^\mesh,b^\mesh\in\DS$ approximate the
points $i\,t$ and $0$, while $c^\mesh\in\G^*$ approximate the center $\tfrac{1}{2}it$. We
denote by $L^\mesh(s)$, $U^\mesh(s,t)$ and $V^\mesh(s,t)$ the lower, upper and vertical parts
of $\pa R^\mesh(s,t)$, respectively. \textsc{(B)}~To~perform the passage to the limit under
the normalization condition at $b^\mesh$, we assume that $\O^\mesh_\DS\supset
R^\mesh_\DS(s,t)$ and $\pa\O^\mesh_\DS\setminus\{a^\mesh\}\supset L^\mesh_\DS(s)$. For $\mesh$
small enough, the discrete harmonic measure from $c^\mesh$ of any path $K^\mesh=K^\mesh_{3d}$
going from $\O^\mesh_{3d}$ to $a^\mesh$ is uniformly bounded from below. We can similarly use
any point $d^\mesh$ lying on the ``straight'' part of the boundary for the (different)
normalization of the observable.}
\end{figure}

\subsection{Main convergence theorem.}

To handle the normalization at $b^\mesh$ of our discrete observable, from now on we assume
that, for some $s,t>0$,
\begin{equation}
\label{OcontainsR}
\begin{array}{l}
\O^\mesh_\DS~\mathit{contains~the~discrete~rectangle}~R^\mesh_\DS(s,t),\cr
\pa\O^\mesh_\DS\setminus\{a^\mesh\}~\mathit{contains~the~lower~side}\ L^\mesh_\DS(s)\
\mathit{of}\ R^\mesh_\DS(s,t), \vphantom{\big|_|^|}\cr
\mathit{and}~b^\mesh~\mathit{is~the~closest\
to}~0~\mathit{vertex~of}~\pa\O^\mesh_\DS~~\mathit{(see~Fig.~\ref{Fig:OsupsetR})}.
\end{array}
\end{equation}
Some assumption of a kind is certainly necessary: one can imagine continuous
domain with such an irregular approach to $b$, that any approximation is forced
to have many ``bottlenecks,'' ruining the estimates.

Let $f^\mesh(z)=f^\mesh_{(\O^\mesh;a^\mesh,b^\mesh)}(z)$ denote the solution of the following
boundary value problem inside the \emph{polygonal} domain $\wt{\O}^\mesh$ (here the tilde
means that {\it we slightly modify the original polygonal domain $\O^\mesh$ near $b^\mesh$,
replacing the polyline $L^\mesh(s)$ by the straight real segment $[-s;s]$}, cf.
\cite{chelkak-smirnov-dca} Theorem 3.20):

\begin{quotation}
{\it \noindent {\bf (a${}^\circ$) holomorphicity:} $f^\mesh$ is holomorphic in $\wt{\O}^\mesh$;

\noindent {\bf (b${}^\circ$) boundary conditions:} $f^\mesh(\z)\parallel (\t(\z))^{-\frac{1}{2}}$ for
$\z\in\pa\wt{\O}^\mesh$, where $\t(\z)$ is the tangent to $\pa\wt{\O}^\mesh$ vector oriented
in the counterclockwise direction, $f^\mesh$ is bounded away from $a^\mesh$;

\noindent {\bf (c${}^\circ$) normalization:} the function
$h^\mesh=h^\mesh_{\O^\mesh,a^\mesh,b^\mesh}:=\Im \int (f^\mesh(\z))^2d\z$ is nonnegative in
$\wt{\O}^\mesh$, bounded away from $a^\mesh$, and
\[
f^\mesh(0)=[{\pa_y}h^\mesh(0)]^{1/2}=1.
\]}
\end{quotation}

As in Section {\ref{SectFKobservable}}, (a${}^\circ$) and (b${}^\circ$) guarantee that
$h^\mesh$ is harmonic in $\O^\mesh$ and constant on $\pa\wt{\O}^\mesh$. Thus,
\[
f^\mesh=\sqrt{2i\pa h^\mesh},\quad \mathrm{where}~~h^\mesh=P_{(\O^\mesh;a^\mesh,0)},
\]
is the Poisson kernel in $\O^\mesh$ having mass at $a^\mesh$ and normalized at $0$. In other
words, $h^\mesh$ is the imaginary part of the conformal mapping (normalized at $0$) from
$\O^\mesh$ onto the upper half-plane $\H$ sending $a^\mesh$ and $0$ to $\infty$ and $0$,
respectively. Note that $\pa h^\mesh \ne 0$ everywhere in $\O^\mesh$, thus $f^\mesh$ is
well-defined in $\O^\mesh$ up to a sign, which is fixed by
%Having the straight segment $[-\tfrac{1}{2}s;\tfrac{1}{2}t]\ss\pa\wt{\O}^\mesh$ near $0$
$f^\mesh(0)=+1$.

\begin{theorem}[\bf convergence of the spin-observable]
\label{ThmSpinConvergence} The discrete solutions of the Riemann-Hilbert boundary value
problems {(A${}^\circ$)\&(B${}^\circ$)\&(C${}^\circ$)} are uniformly close in the bulk to their continuous counterparts
$f^\mesh$, defined by {(a${}^\circ$)\&(b${}^\circ$)\&(c${}^\circ$)}. Namely, there exists $\ve(\mesh)=\ve(\mesh,r,R,s,t)$
such that for all discrete domains $(\O^\mesh_\DS;a^\mesh,b^\mesh)\ss B(0,R)$ satisfying
(\ref{OcontainsR}) and for all $z^\mesh\in \O^\mesh_\DS$ lying in the same connected component
of the $r$-interior of $\O^\mesh$ as the neighborhood of $b^\mesh$ (see
Fig.~\ref{Fig:OsupsetR}) the following holds true:
\[
|F^\mesh(z^\mesh)-f^\mesh(z^\mesh)|\le \ve(\mesh)\to 0\quad \mathit{as}\quad \mesh\to 0
\]
(uniformly with respect to the shape of $\O^\mesh$ and the structure of $\DS^\mesh$).
\end{theorem}

Moreover, it is easy to conclude from this theorem that the convergence also should hold true
at any boundary point $d^\mesh$ such that $\pa\O^\mesh_\DS$ has a ``straight'' part
near~$d^\mesh$. Namely, as in (\ref{OcontainsR}), let (see Fig.~\ref{Fig:OsupsetR}B)

\begin{quotation}
\noindent \emph{$d^\mesh\in\pa\O^\mesh_\DS$ be a boundary point, the boundary
$\pa\O^\mesh_\DS$ is ``straight'' near $d^\mesh$ and oriented in the (macroscopic) direction
$\t_d:|\t_d|=1$, i.e.,}

\noindent \emph{$\O^\mesh_\DS$ contains the discretization $\wt{R}^\mesh_\DS(\wt{s},\wt{t})$
of the rectangle $d+\t_d\cdot R(\wt{s},\wt{t})$ and $\pa\O^\mesh_\DS\setminus\{a^\mesh\}$
contains the discretization $\wt{L}^\mesh_\DS(\wt{s})$ of the segment
$d+\t_d\cdot[-\wt{s};\wt{s}]$.}
\end{quotation}

\noindent Further, let $\wt{\cF}^\mesh$ denotes the solution of the boundary value problem
(A${}^\circ$)\&(B${}^\circ$) in the discrete half-plane $(d+\t_d\cdot \H)_\DS$ which is
asymptotically equal to $(\t_d)^{-1/2}$ (again, $\wt{\cF}\equiv (\t_d)^{-1/2}$, if one deals
with, e.g., the properly oriented square grid).

\begin{corollary}[\bf convergence of spin-observable on the boundary]
\label{CorBoundarySpinConv} If centers of $R^\mesh_\DS(s,t)$ and
$\wt{R}^\mesh_\DS(\wt{s},\wt{t})$ are connected in the $r$-interior of $\O^\mesh\ss B(0,R)$,
then
\[
|F^\mesh(d^\mesh) - [(\t_d)^{1/2}\wt{\cF}^\mesh(d^\mesh)]\cdot f^\mesh(d)|\le \ve(\mesh)\to
0\quad \mathit{as}\quad \mesh\to 0
\]
(one should replace both $L^\mesh_{\G^*}(s)$ and $\wt{L}^\mesh_{\G^*}(\wt{s})$ by the
corresponding straight segments to define properly the value of the ``continuous'' solution
$f^\mesh(d)$ on the boundary).
\end{corollary}

\begin{proof}[\bf Proof of Corollary \ref{CorBoundarySpinConv}]
Let $c^\mesh$ be a discrete approximation of the point $c:=\frac{1}{2}it$ and $\wt{c}^\mesh$
be a discrete approximation of the point $\wt{c}:=d+\frac{1}{2}i\t_d\wt{t}$. It follows from
the discrete Harnack principle (Proposition~\ref{PropBoundHar} and
Proposition~\ref{PropHarnack}) that
\[
{|F^\mesh(d^\mesh)|}\asymp {|H^\mesh(\wt{c}^\mesh)|} \asymp {|H^\mesh(c^\mesh)|} \asymp
|F^\mesh(b^\mesh)|\asymp 1,
\]
uniformly in $\O^\mesh$ and $\mesh$, if all the parameters $r,R,s,t,\wt{s},\wt{t}$ are fixed.
Denote by $\wt{F}^\mesh$ the discrete observable $F^\mesh$ renormalized at $d^\mesh$:
\[
\wt{F}^\mesh:= \frac{\wt{\cF}^\mesh(d^\mesh)}{F^\mesh(d^\mesh)}\cdot F^\mesh,
\]
and by $\wt{f}^\mesh$ the corresponding continuous function renormalized at $d$:
\[
\wt{f}^\mesh:=\frac{(\t_d)^{-1/2}}{f^\mesh(d)}\cdot f^\mesh.
\]
Again, $|f^\mesh(d)|\asymp |h^\mesh(\wt{c})|\asymp |h^\mesh(c)|\asymp |f^\mesh(0)|=1$ due to
the Harnack principle. Moreover, since $h^\mesh$ is equal to the imaginary part of the
conformal mapping from $\O^\mesh$ onto the upper half-plane, the Koebe Distortion Theorem
gives
\[
|f^\mesh(c^\mesh)|^2= 2|\pa h^\mesh(c^\mesh)|\asymp |h^\mesh(c^\mesh)|\asymp 1.
\]
One needs to prove that the ratio
\[
\frac{F^\mesh(d^\mesh)}{[(\t_d)^{1/2}\wt{\cF}^\mesh(d^\mesh)]\cdot f^\mesh(d)}=
\frac{F^\mesh}{\wt{F}^\mesh}\cdot \frac{\wt{f}^\mesh}{f^\mesh}=
\frac{F^\mesh(c^\mesh)}{f^\mesh(c^\mesh)}\cdot
\frac{\wt{f}^\mesh(c^\mesh)}{\wt{F}^\mesh(c^\mesh)}
\]
is uniformly close to $1$. This follows from Theorem~\ref{ThmSpinConvergence}, since
$F^\mesh(c^\mesh)$ is uniformly close to $f^\mesh(c^\mesh)$, $\wt{F}^\mesh(c^\mesh)$ is
uniformly close to $\wt{f}^\mesh(c^\mesh)$, and
$|\wt{f}^\mesh(c^\mesh)|\asymp|f^\mesh(c^\mesh)|\asymp 1$.
\end{proof}

\begin{proof}[\bf Proof of Theorem~\ref{ThmSpinConvergence}]
Assume that
\[
|F^\mesh(z^\mesh)-f^\mesh(z^\mesh)|\ge \ve_0>0
\]
for some sequence of domains $\O^\mesh$ with $\mesh\to 0$. Passing to a subsequence, we may
assume that $z^\mesh\to z$. The set of all simply-connected domains $B(z,r)\ss\O\ss B(0,R)$ is
compact in the Carath\'eodory topology, % (of convergence of germs of conformal maps),
so, passing to a subsequence once more, we may assume that
\[
(\O^\mesh;a^\mesh,b^\mesh)~\CaraTo~(\O;a,b)\quad \mathrm{with~respect~to~~} z^\mesh\to
z\in\O~~ \mathrm{as}~~\mesh\to 0.
\]
Note that $\O\supset R(s,t)=(-s;s)\ts(0;t)$, $\pa\O\supset[-s;s]$, and $b^\mesh\to b=0$. Let
$h=h_{(\O;a,b)}$ be the continuous Poisson kernel in $\O$ having mass at $a$ and normalized at
$0$ (i.e., the imaginary part of the properly normalized conformal mapping from $\O$ onto
$\H$). Then,
\[
h^\mesh\rra h~~\mathrm{as}~~\mesh\to 0,
\]
uniformly on compact subsets of $\O$, since this kernel can be easily constructed as a
pullback of the Poisson kernel in the unit disc. Moreover, it gives
\[
f^\mesh= \sqrt{2i\pa h^\mesh}\rra f=\sqrt{2i\pa h}\quad \mathrm{as}~~\mesh\to 0
\]
uniformly on compact subsets of $\O$ (here and below the sign of the square root is chosen
so that $f^\mesh(0)=f(0)=+1$). We are going to prove that, at the same time,
\[
H^\mesh\rra h\quad \mathrm{and}\quad F^\mesh\rra \sqrt{2i\pa h}\quad \mathrm{as}~~\mesh\to
0.
\]

We start with the proof of the {\it uniform boundedness of $H^\mesh$ away from $a^\mesh$}.
Denote by $c:=\tfrac{1}{2}it$ the center of the rectangle $R(s,t)$ and by $c^\mesh\in\G^*$ the
dual vertex closest to~$c$ (see Fig.~\ref{Fig:OsupsetR}B). Let $d\!>\!0$ be small enough and
$\g_d^a\ss B(a_d,\frac{1}{2}d)$ be some crosscut in $\O$ separating $a$ from $c$ in $\O$.
Further, let $L_{3d}^\mesh\ss \O^\mesh\cap \pa B(a_d,3d)$ be an arc separating $a^\mesh$ from
$c^\mesh$ in $\O^\mesh$ (such an arc exists, if $\mesh$ is small enough), and $\O^\mesh_{3d}$
denote the connected components of $\O^\mesh\setminus L^\mesh_{3d}$, containing $c^\mesh$.

The Harnack principle (Propositions~\ref{PropBoundHar} and \ref{HarnackIntF^2}) immediately
give
\[
H^\mesh(c^\mesh)\asymp 1~~~\mathrm{uniformly~~in~~}\mesh,
\]
if $s$ and $t$ are fixed. Let
\[
M^\mesh_{3d}:=\max\{H^\mesh_\G(u^\mesh),~u^\mesh\in(\O^\mesh_{3d})_\G\}.
\]
Because of the subharmonicity of $H\big|_\G$, $ M^\mesh_{3d}= H^\mesh_\G(u_0^\mesh)\le
H^\mesh_\G(u_1^\mesh)\le H^\mesh_\G(u_2^\mesh)\le\dots $ for some path of consecutive
neighbors $K^\mesh_\G=\{u_0^\mesh\sim u_1^\mesh\sim u_2^\mesh\sim\dots\}\ss\G$. Since the
function $H^\mesh_\G$ vanishes everywhere on $\pa\O^\mesh_{\wt{\G}}$ except $a^\mesh$, this
path necessarily ends at $a^\mesh$. Taking on the dual graph a close path
$K^\mesh_{\G^*}=\{w_0^\mesh\sim w_1^\mesh\sim w_2^\mesh\sim\dots\}\ss\G^*$ starting near
$u_0^\mesh$ and ending near $a^\mesh$, we deduce from Remark~\ref{RemHGasympG*} that
\[
H^\mesh_{\G^*}(w_k^\mesh)\ge\const\cdot M^\mesh_{3d}\,.
\]
Then,
\[
H^\mesh_{\G^*}(c^\mesh)\ge \dhm\mesh{c^\mesh}{K^\mesh_{\G^*}}{\O^\mesh_{\G^*}\setminus
K^\mesh_{\G^*}}\cdot \const\cdot M^\mesh_{3d} \ge \const((\O;a),d) \cdot M^\mesh_{3d},
\]
since $H\big|_{\G^*}$ is superharmonic and
\[
\dhm\mesh{c^\mesh}{K^\mesh_{\G^*}}{\O^\mesh_{\G^*}\setminus K^\mesh_{\G^*}} \ge
\textfrac{1}{2}\hm{c^\mesh}{K^\mesh_{\G^*}}{\O^\mesh\setminus K^\mesh_{\G^*}} \ge
\const((\O;a),d)>0
\]
for all sufficiently small $\mesh$'s (see Fig.~\ref{Fig:OsupsetR}B and
\cite{chelkak-smirnov-dca} Lemma 3.14).

Thus, the functions $H^\mesh$ are uniformly bounded away from $a^\mesh$. Due to
Theorem~\ref{FboundH}, we have
\begin{equation}
\label{FuniBound} F^\mesh=O(1)~\mathrm{uniformly~on~compact~subsets~of}~\O.
\end{equation}
Moreover, using uniform estimates of the discrete harmonic measure in rectangles (Lemma
\ref{LemmaDhmR}) exactly in the same way as in the proof of Theorem~\ref{ThmCFexists}, we
arrive at
\begin{equation}
\label{HFuniBound1} H^\mesh(v)=O(\Im v),~~F^\mesh=O(1)~~\mathrm{uniformly~in}\
R^\mesh(\tfrac{1}{2}s,\tfrac{1}{2}t).
\end{equation}

\smallskip

Taking a subsequence, we may assume that
\[
F^\mesh\rra F~~\mathrm{and}~~H^\mesh\rra H\quad \mathrm{for~some}~~F:\O\to\C,\
H:\O\to\R,
\]
uniformly on all compact subsets of $\O$. The simple passage to the limit in (\ref{Hmesh=Int})
gives $H(v_2)-H(v_1)=\Im \int_{[v_1;v_2]} (F(\z))^2d\z,$ for each segment $[v_1;v_2]\ss\O$.
So, $F^2=2i\pa H$, and it is sufficient to show that $H=P_{(\O;a,b)}$. Being a limit of
discrete subharmonic functions $H^\mesh_\G$, as well as discrete superharmonic functions
$H^\mesh_{\G^*}$, the function $H$ is harmonic. The next step is the {\it
identification of the boundary values of $H$}.

Let $u\in\O$ and $d\!>\!0$ be so small that $u\in\O^\mesh_{4d}$. Recall that the functions
$H^\mesh\big|_\G$ are subharmonic, uniformly bounded away from $a^\mesh$, and $H^\mesh_\G=0$
on the (modified) boundary $\pa\O^\mesh_{\wt{\G}}$, except at $a^\mesh$. Thus, the weak
Beurling-type estimate of the discrete harmonic measure (Lemma \ref{WeakBeurling}) easily
gives
\begin{align*}
H(u)=\lim_{\mesh\to 0} H^\mesh_\G(u)&\le \const(\O,d)\cdot \lim_{\mesh\to 0
}\lt[\frac{\dist(u\,;\pa\O^\mesh_{3d}\setminus\pa
B^\mesh_\G(a_d,3d))}{\dist_{\O^\mesh_\G}(u\,;\pa B^\mesh_\G(a_d,3d))}\rt]^\b\cr & \le
\const(\O,d)\cdot (\dist(u\,;\pa\O_{3d}\setminus B(a_d,3d))^\b\quad
\mathrm{for~all}~u\in\O_{5d}
\end{align*}
%Passing to the limit as $\mesh\to 0$, we obtain
%\[
%H(u)\le \const(\O,d)\cdot(\dist(u;\pa\O^\mesh_{3d}))^{-\b}.
%\]
(since, if $\mesh$ is small enough, $u\in\O^\mesh_{4d}$). Thus, for each $d>0$, $H(u)\to 0$ as
$u\to\pa\O$ inside $\O_{5d}$, i.e., $H=0$ on $\pa\O\setminus\{a\}$. Clearly, $H$ is
nonnegative because $H^\mesh$ are nonnegative. Therefore, $H$ should be proportional to the
Poisson kernel in $\O$ having mass at $a$, i.e.,
\[
H=\mu^2P_{(\O;a,b)}\quad \mathrm{and}\quad F=\mu\sqrt{2i\pa P_{(\O;a,b)}}~~\mathrm{for~some}\
\mu\in \R.
\]
Note that $|\mu|$ is uniformly bounded from $\infty$ and $0$, since $H^\mesh(c^\mesh)\asymp 1$
uniformly in $\mesh$.

\smallskip

To finish the proof, we need to show that $\mu=1$. For each $0\!<\!\a\!\ll\!\g\!\ll\!t$, we
have
\[
F^\mesh(z)\rra \mu\cdot(1\!+\!O(\g))\quad \mathrm{uniformly~for}~z\in[-2\g,2\g]\ts[\a,\g],
\]
as $\mesh\to 0$. Recall that $\O^\mesh_\DS\supset R^\mesh_\DS(s,t)$ and
$\pa\O^\mesh_\DS\supset L^\mesh_\DS(s)$ for all $\mesh$ (see (\ref{OcontainsR}) and
Fig.~\ref{Fig:OsupsetR}B). Set $F_0^\mesh:=F^\mesh-\mu\cF^\mesh$, where the function
$\cF^\mesh$ is defined in Theorem~\ref{ThmCFexists}. Then $F_0^\mesh$ is s-holomorphic in
$R^\mesh_\DS(s,t)$, satisfies the boundary condition (B${}^\circ$) on the lower boundary, and
\[
F^\mesh_0(z)\rra O(\g)\quad  \mathrm{uniformly~for}~z\in[-2\g,2\g]\ts[\a,\g],
\]
since $\cF^\mesh\rra 1$. Moreover, due to (\ref{FuniBound}), we have $F_0^\mesh(z)=O(1)$
everywhere in the rectangle $R^\mesh_{\DS}(\frac{1}{2}s,\frac{1}{2}t)$. Let
$H^\mesh_0:=\int^\mesh (F^\mesh_0(z))^2 d^\mesh z$, where the additive constant is chosen so
that $H^\mesh_0=0$ on the boundary $L^\mesh(s)$. Then
\[
H^\mesh_0=O(\a\!+\!\g^3)+o_{\mesh\to 0}(1)\quad
\mathrm{uniformly~~on~~the~~boundary~~of}~~R^\mesh_\G(2\g,\g).
\]
Since the subharmonic function $H^\mesh_0\big|_\G$ vanishes on $\wt{L}^\mesh_{\G}(s)$, Lemma
\ref{LemmaDhmR} gives
\[
H^\mesh_0(b_{\mathrm{int}}^\mesh)\le O(\mesh\g^{-1})\cdot [O(\a\!+\!\g^3)+o_{\mesh\to
0}(1)]=\mesh\cdot O(\g^{-1}[\a+o_{\mesh\to 0}(1)]\!+\!\g^2),
\]
where $b^\mesh_{\mathrm{int}}\in\G$ denotes the inner vertex near $b^\mesh$. On the other
hand,
\[
H^\mesh_0(b_{\mathrm{int}}^\mesh)\asymp\mesh
|F^\mesh_0(b^\mesh)|^2=\mesh|(1\!-\!\mu)\cF^\mesh(b^\mesh)|^2\asymp \mesh|1\!-\!\mu|^2
\]
which doesn't depend on $\a$ and $\g$. Successively passing to the limit as $\mesh\to 0$,
$\a\to 0$ and $\g\to 0$, we obtain $\mu\!=\!1$. Thus, $F^\mesh\rra \sqrt{2i\pa P_{(\O;a,b)}}$
as $\mesh=\mesh_k\to 0$.
\end{proof}

\section{$4$-point crossing probability for the  FK-Ising model}
\setcounter{equation}{0}

\label{SectFK4points}

Let $\O^\mesh_\DS\!\ss\!\DS$ be a discrete quadrilateral, i.e. simply-connected discrete
domain composed of inner rhombi \mbox{$z\!\in\!\Int\O^\mesh_\DS$} and boundary half-rhombi
$\z\in\pa\O^\mesh_\DS$, with four marked boundary points $a^\mesh,b^\mesh,c^\mesh,d^\mesh$ and
alternating Dobrushin-type boundary conditions (see Fig.~\ref{Fig:FK4points}): $\pa\O^\mesh_\DS$ consists
of two ``white'' arcs $a^\mesh_{\mathrm{w}}b^\mesh_{\mathrm{w}}$,
$c^\mesh_{\mathrm{w}}d^\mesh_{\mathrm{w}}$ and two ``black'' arcs
$b^\mesh_{\mathrm{b}}c^\mesh_{\mathrm{b}}$, $d^\mesh_{\mathrm{b}}a^\mesh_{\mathrm{b}}$.
%and four edges $[a^\mesh_{\mathrm{b}}a^\mesh_{\mathrm{w}}]$,
%$[b^\mesh_{\mathrm{w}}b^\mesh_{\mathrm{b}}]$, $[c^\mesh_{\mathrm{b}}c^\mesh_{\mathrm{w}}]$,
%$[d^\mesh_{\mathrm{w}}d^\mesh_{\mathrm{b}}]$ of $\DS$.
In the random cluster language it means that the four arcs are wired/free/wired/free,
and in the loop representation this creates two interfaces that end at the four marked points
and can connect in two possible ways.
As in Section~\ref{SectFKIsing}, we assume that
\mbox{$b^\mesh_{\mathrm{b}}-b^\mesh_{\mathrm{w}}=i\mesh$.}

\begin{figure}
%\centering{\includegraphics[width=\textwidth]{FK4points.ps}}
\centering{
\begin{minipage}[b]{0.4\textwidth}
\centering{\includegraphics[width=\textwidth]{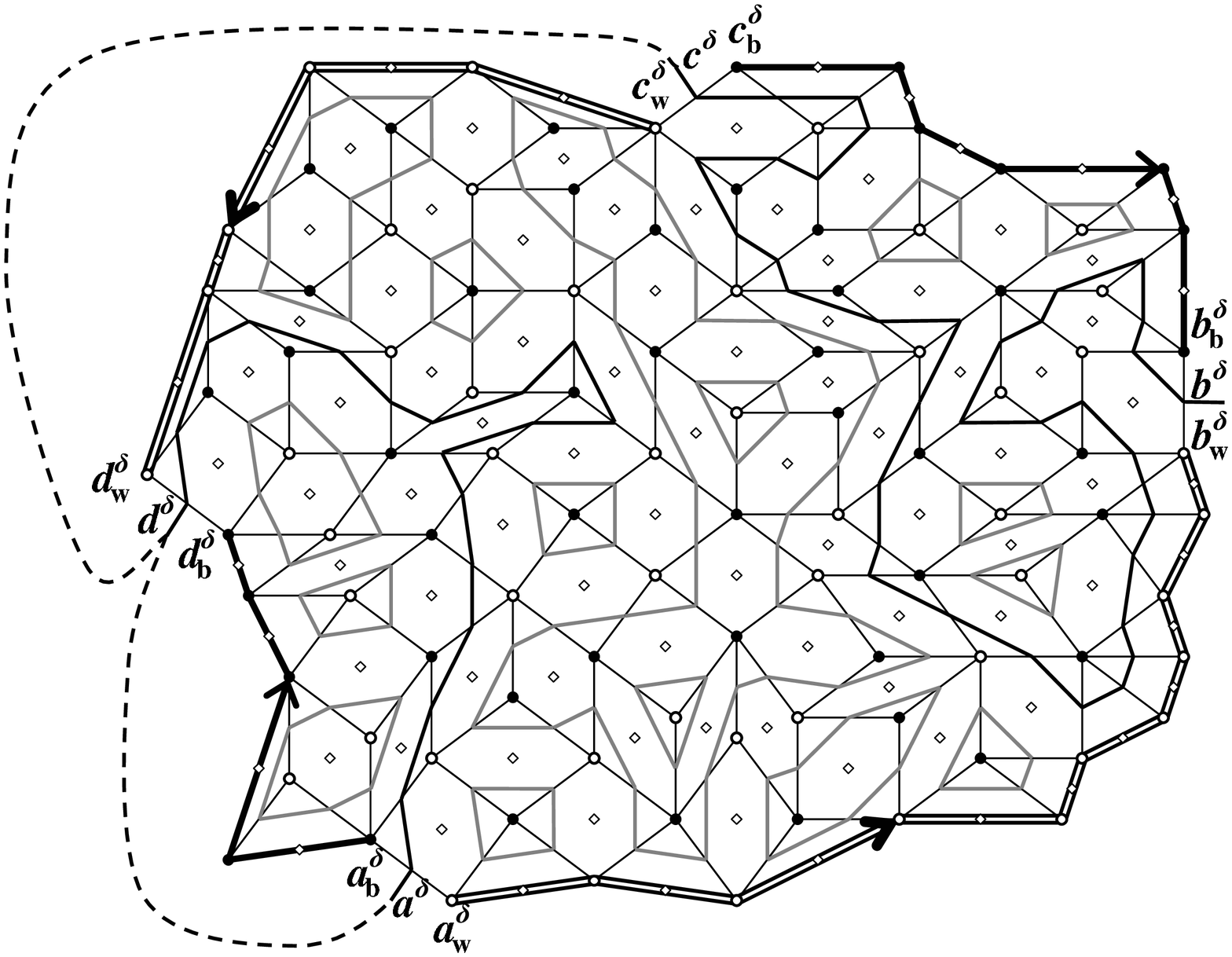}}

\centering{\textsc{(A)}}
\bigskip

\centering{
\begin{minipage}[b]{0.42\textwidth}
\centering{\includegraphics[width=\textwidth]{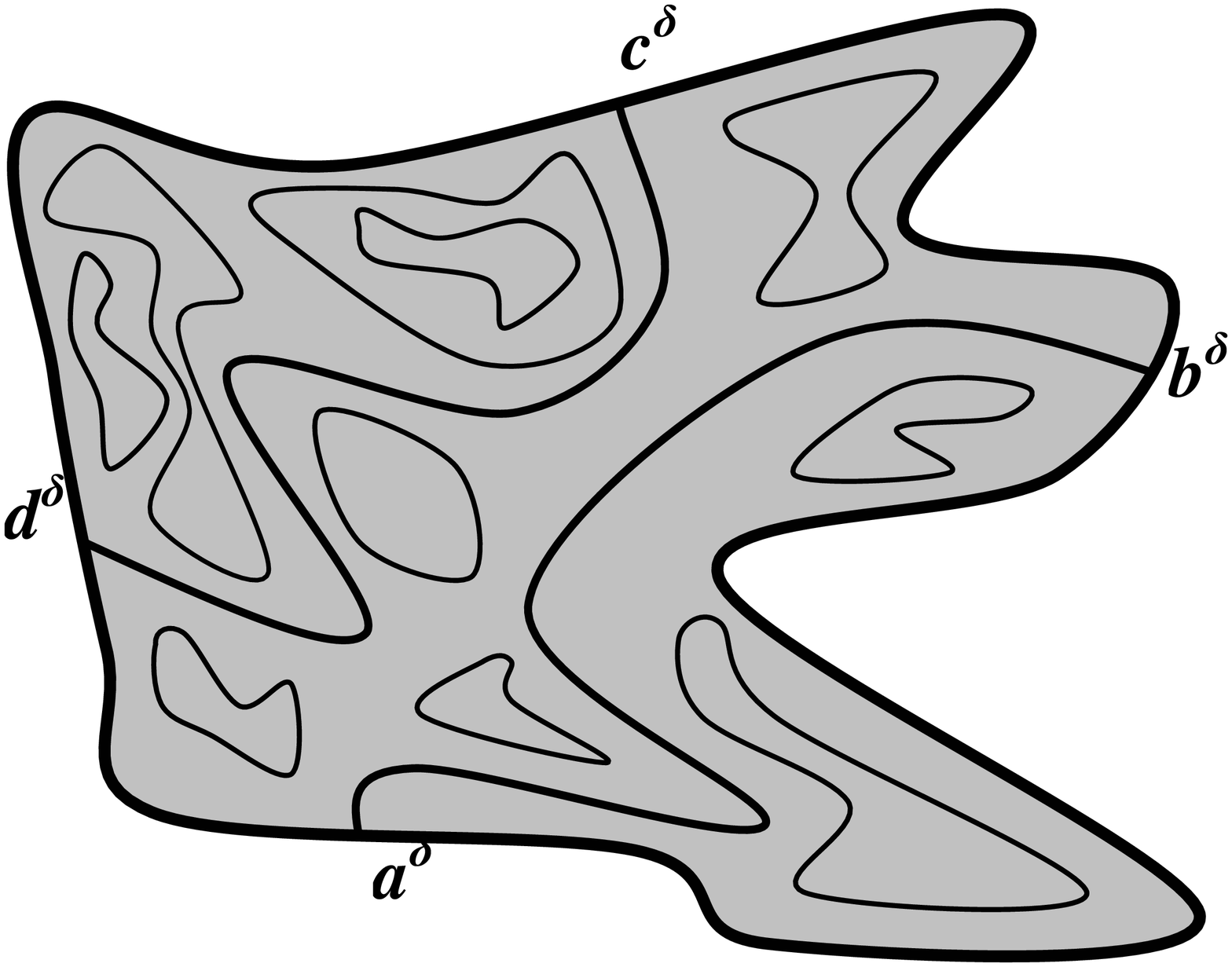}}

\smallskip

$\rP^\mesh$
\end{minipage}
\textsc{vs.}
\begin{minipage}[b]{0.42\textwidth}
\centering{\includegraphics[width=\textwidth]{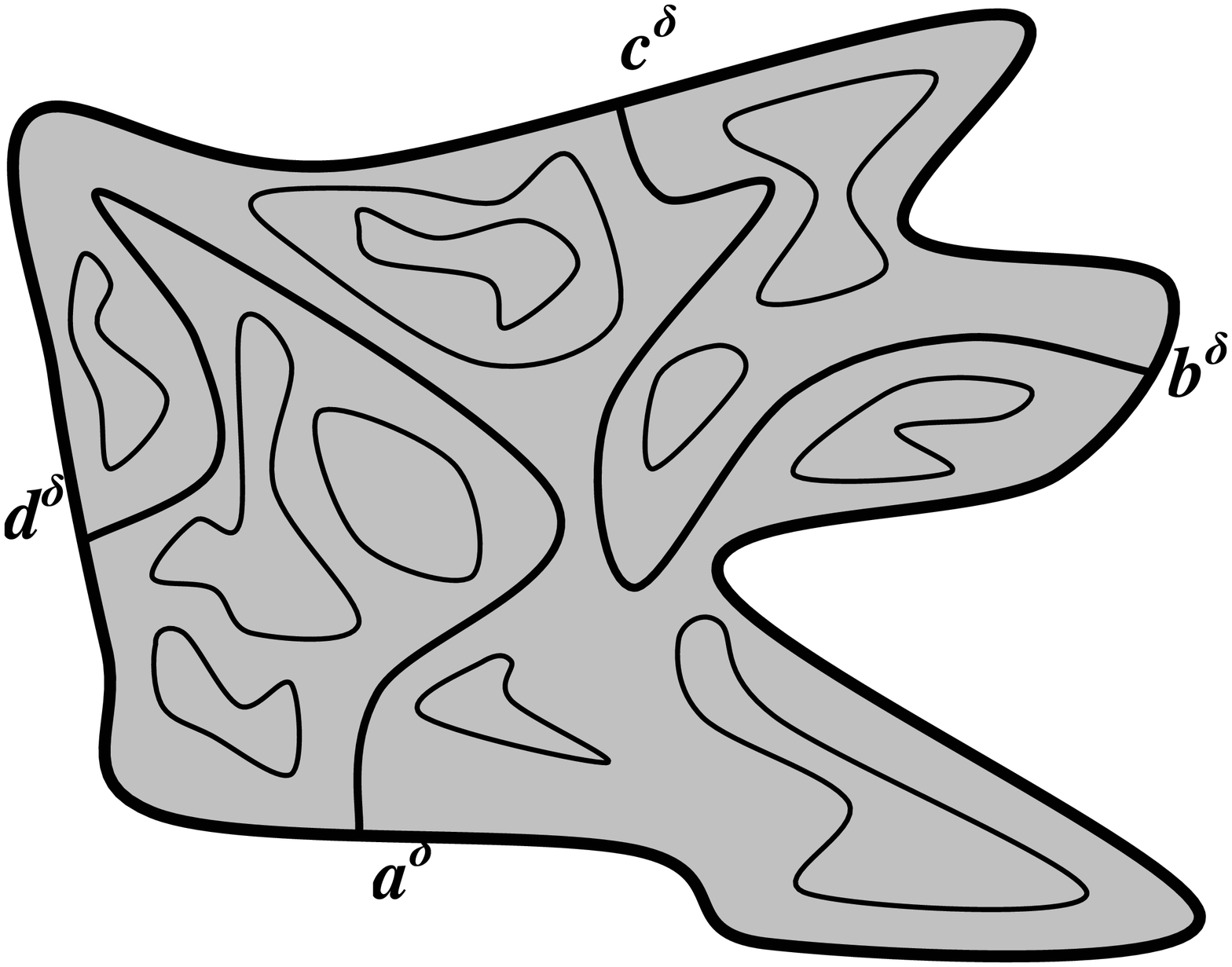}}

\smallskip

$\rQ^\mesh$
\end{minipage}}

\bigskip

\centering{\textsc{(B)}}
\end{minipage}
\hskip 0.03\textwidth
\begin{minipage}[b]{0.55\textwidth}
\centering{
\begin{minipage}[b]{0.42\textwidth}
\centering{\includegraphics[width=\textwidth]{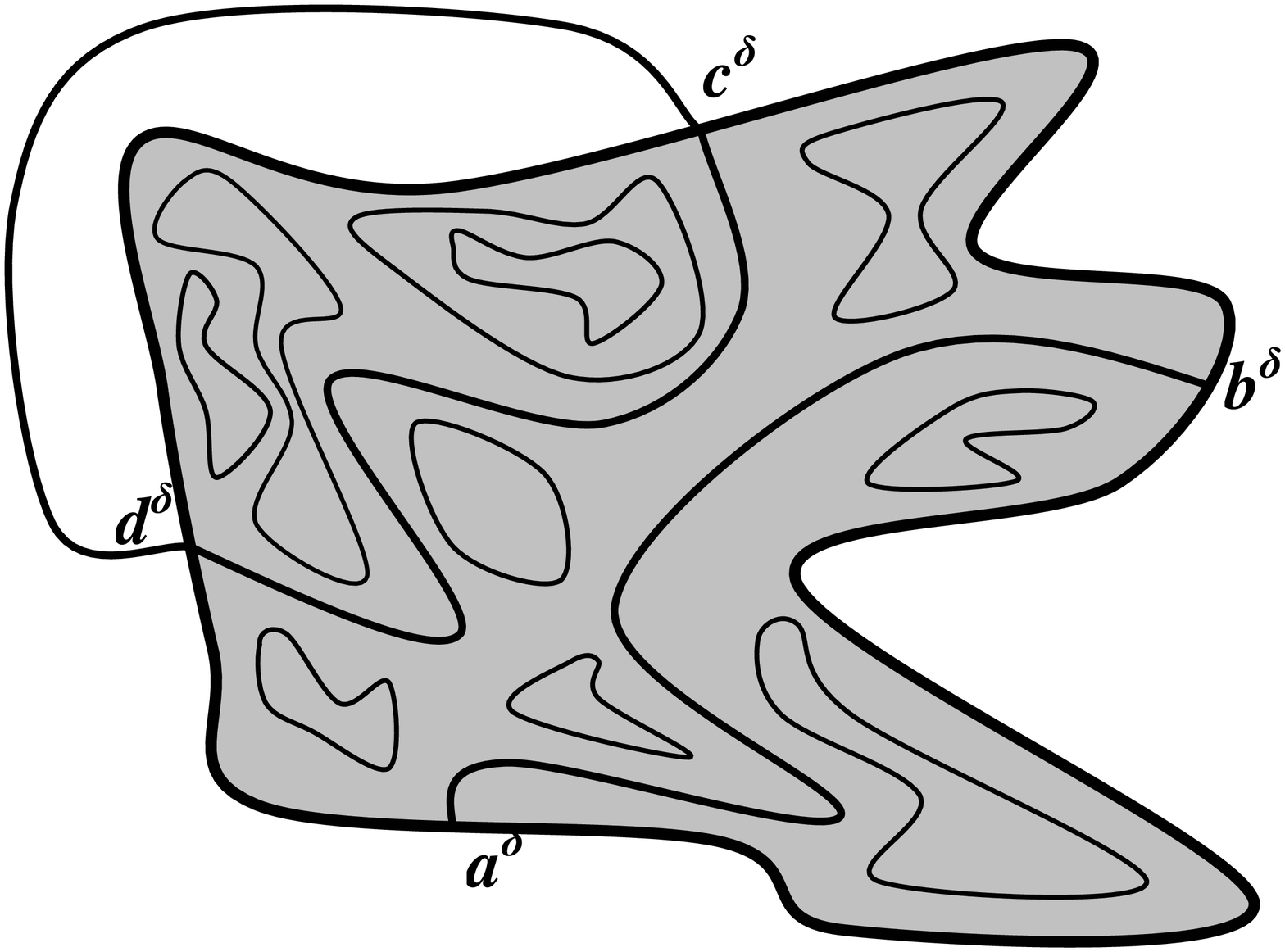}}

\smallskip

$\frac{\sqrt{2}\rP^\mesh}{\sqrt{2}\rP^\mesh+\rQ^\mesh}$
\end{minipage}
\textsc{vs.}
\begin{minipage}[b]{0.42\textwidth}
\centering{\includegraphics[width=\textwidth]{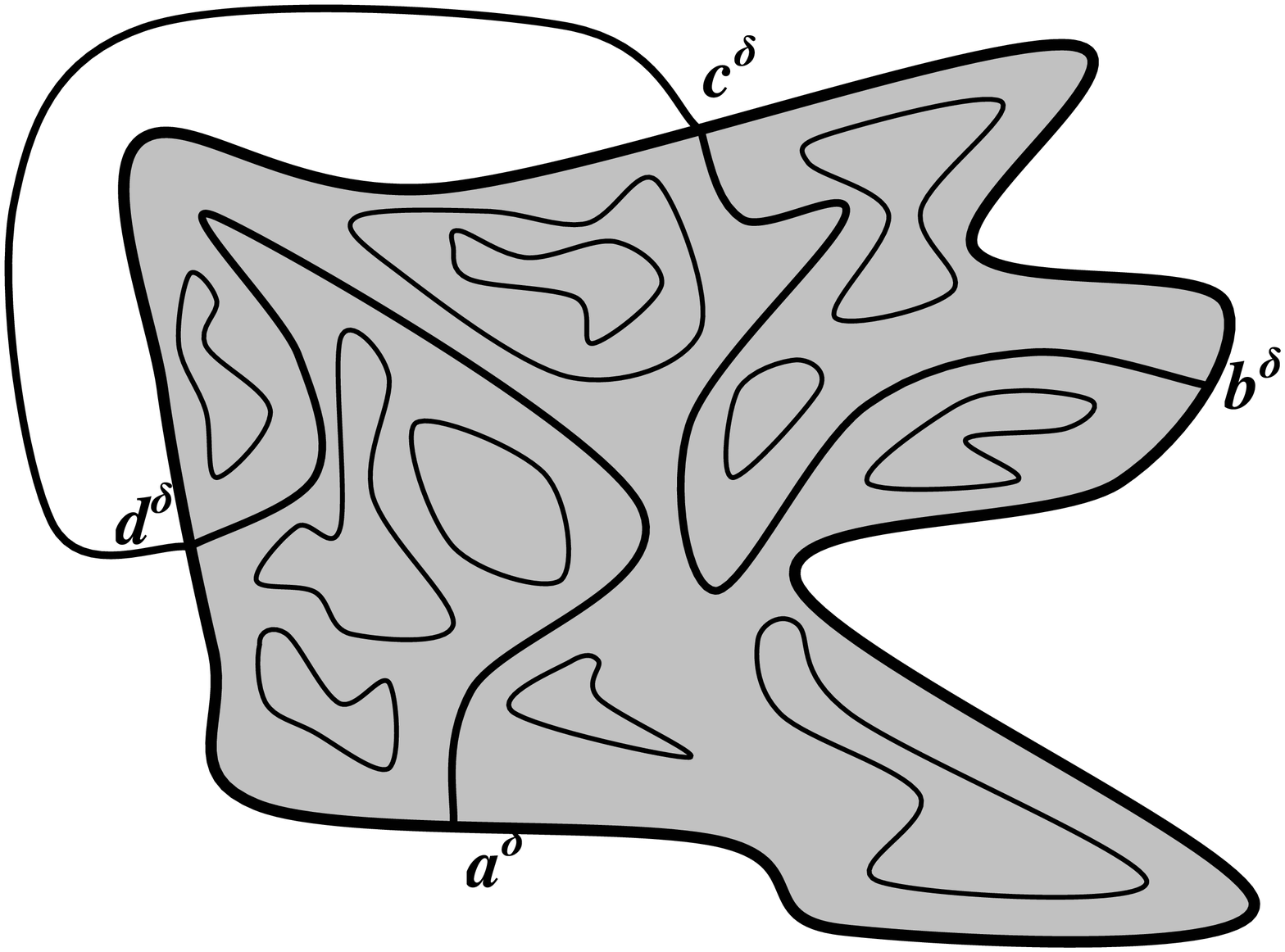}}

\smallskip

$\frac{\rQ^\mesh}{\sqrt{2}\rP^\mesh+\rQ^\mesh}$
\end{minipage}

\bigskip
\centering{\textsc{(C)}}
\bigskip\bigskip

\begin{minipage}[b]{0.42\textwidth}
\centering{\includegraphics[width=\textwidth]{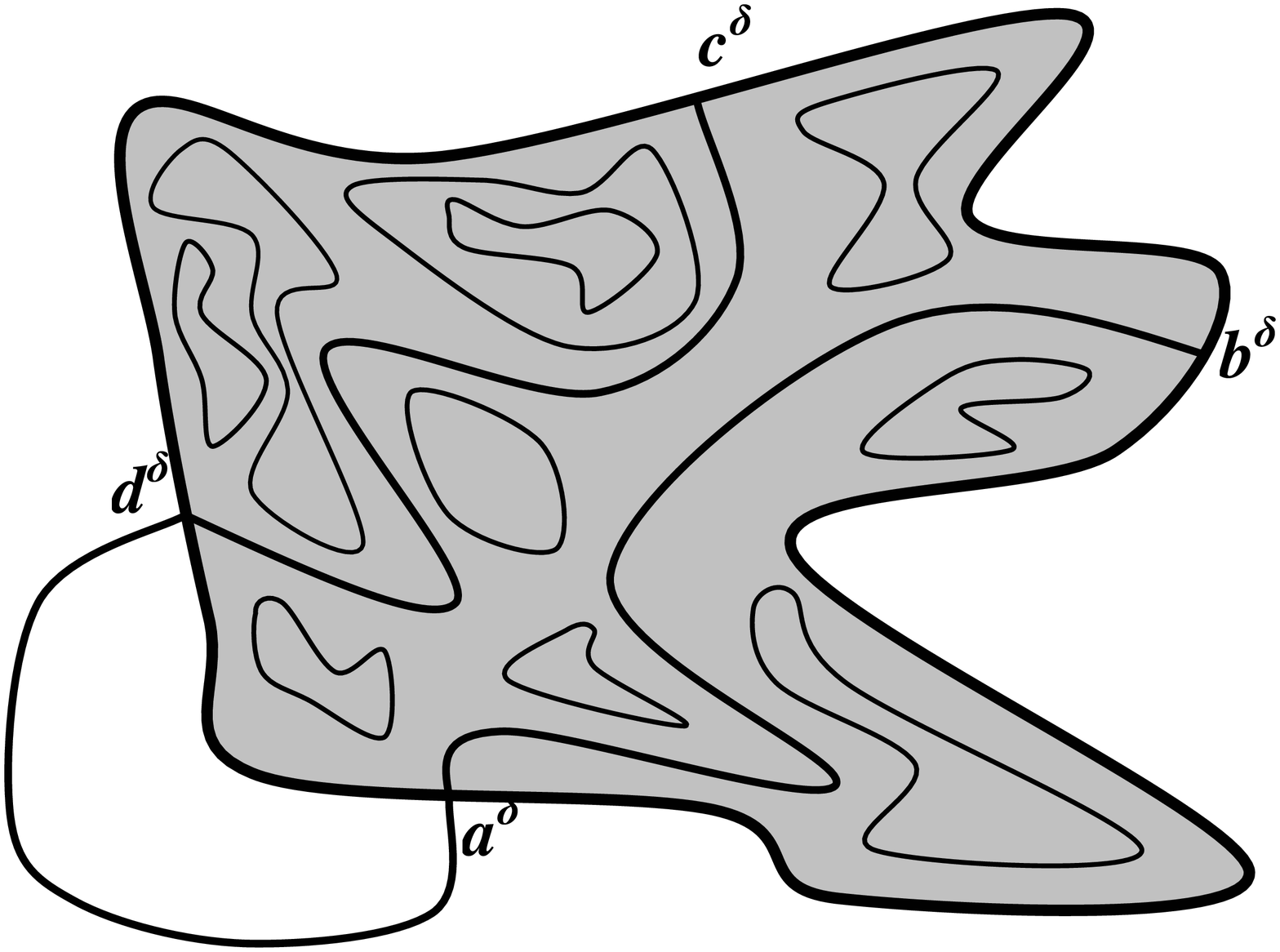}}

\smallskip

$\frac{\rP^\mesh}{\rP^\mesh+\sqrt{2}\rQ^\mesh}$
\end{minipage}
\textsc{vs.}
\begin{minipage}[b]{0.42\textwidth}
\centering{\includegraphics[width=\textwidth]{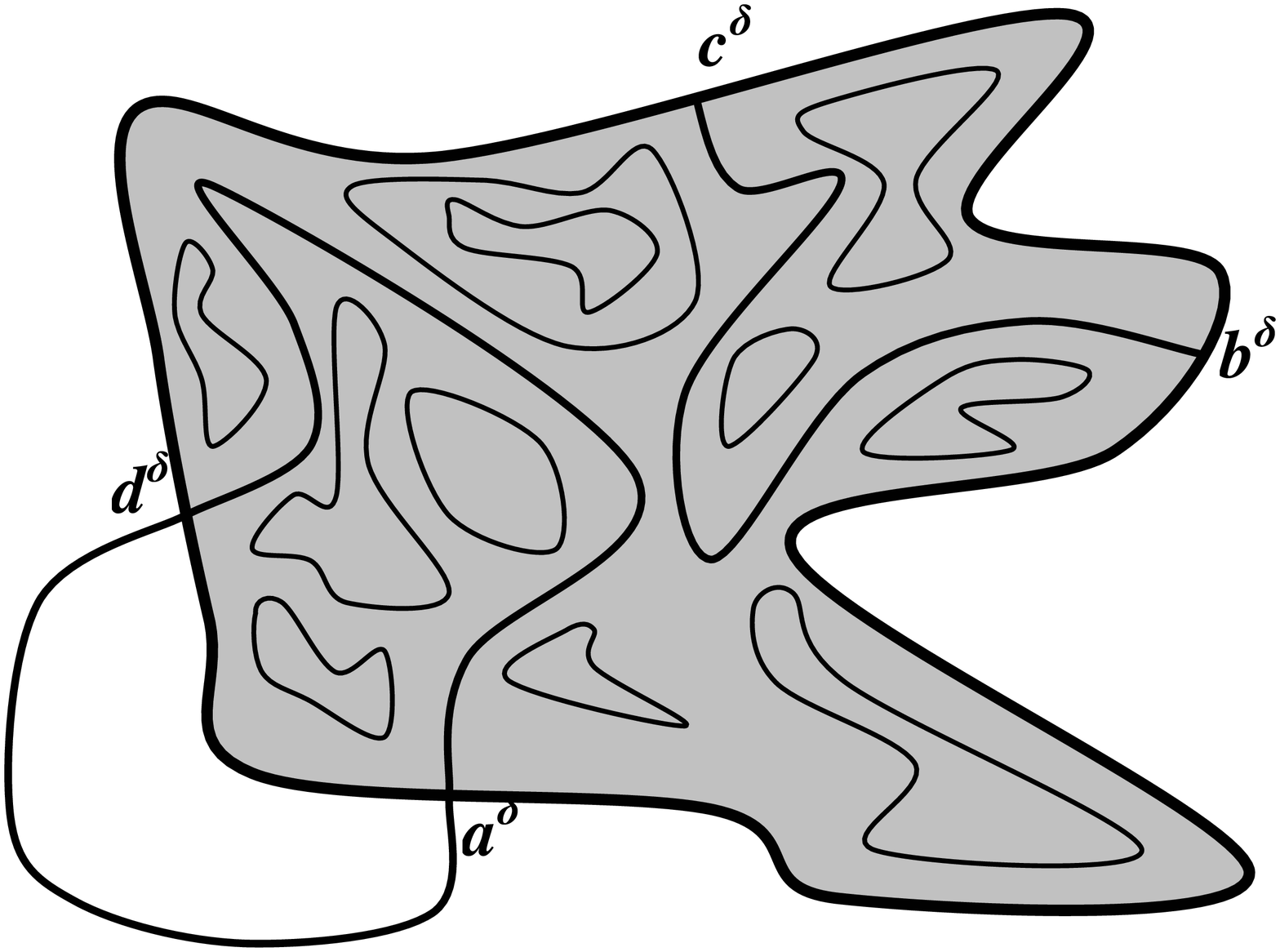}}

\smallskip

$\frac{\sqrt{2}\rQ^\mesh}{\rP^\mesh+\sqrt{2}\rQ^\mesh}$
\end{minipage}}

\bigskip
\centering{\textsc{(D)}}

\end{minipage}}

\caption{\label{Fig:FK4points} \textsc{(A)} Discrete quadrilateral $\O^\mesh_\DS$ with four
marked boundary points and Dobrushin-type boundary conditions. If we draw the additional
external edge $[c^\mesh d^\mesh]$ (or $[a^\mesh d^\mesh]$), there is one interface $\g^\mesh$,
going from $a^\mesh$ (or $d^\mesh$, respectively) to $b^\mesh$. Since $\g^\mesh$ has black
vertices to the left and white vertices to the right, its winding on all boundary arcs is
defined uniquely due to topological reasons. \textsc{(B)} Besides loops, there are two
interfaces: either $a^\mesh\lra b^\mesh$ and $c^\mesh\lra d^\mesh$, or $a^\mesh\lra d^\mesh$
and $b^\mesh\lra c^\mesh$. We denote the probabilities of these events by $\rP^\mesh$ and
$\rQ^\mesh$, respectively. \textsc{(C)}~The external edge $[c^\mesh d^\mesh]$ changes the
probabilities: there is one additional loop (additional factor $\sqrt{2}$), if $b^\mesh$ is
connected directly with $a^\mesh$. \textsc{(D)}~The external edge $[a^\mesh d^\mesh]$ changes
the probabilities differently.}
\end{figure}

Due to Dobrushin-type boundary conditions, each configuration (besides loops) contains
\emph{two interfaces}, either connecting \mbox{$a^\mesh$ to $b^\mesh$} and \mbox{$c^\mesh$ to
$d^\mesh$}, or vice versa. Let
\[
\rP^\mesh=\rP^\mesh(\O^\mesh_\DS;a^\mesh, b^\mesh, c^\mesh, d^\mesh):=\P(a^\mesh\lra
b^\mesh;c^\mesh\lra d^\mesh)
\]
denote the probability of the first event, and $\rQ^\mesh=1\!-\!\rP^\mesh$.

\begin{theorem}
\label{ThmFK4points} For all $r,R,t>0$ there exists $\ve(\mesh)=\ve(\mesh,r,R,t)$ such that
\begin{quotation}
if $B(0,r)\ss\O^\mesh\ss B(0,R)$ and either both $\hm{0}{\O^\mesh}{a^\mesh b^\mesh},
\hm{0}{\O^\mesh}{c^\mesh d^\mesh}$ or both $\hm{0}{\O^\mesh}{b^\mesh c^\mesh},
\hm{0}{\O^\mesh}{d^\mesh a^\mesh}$ are $\ge t$ (i.e., quadrilateral $\O^\mesh$ has no
neighboring small arcs), then
\[
|\rP^\mesh(\O^\mesh_\DS;a^\mesh, b^\mesh, c^\mesh, d^\mesh) -
p(\O^\mesh;a^\mesh,b^\mesh,c^\mesh,d^\mesh)|\le\ve(\mesh)\to 0~\mathit{as}~\mesh\to 0
\]
\end{quotation}
(uniformly with respect to the shape of $\O^\mesh$ and $\DS^\mesh$), where $p$ depends only on
the conformal modulus of the quadrilateral $(\O^\mesh;a^\mesh,b^\mesh,c^\mesh,d^\mesh)$. In
particular, for $u\in [0,1]$,
\[ %\begin{equation}\label{CrossP}
p(\H;0,1\!-\!u,1,\infty)= %\Pi(\H;[1\!-\!u,1]\lra[\infty,0]) =
\frac{\sqrt{1-\sqrt{1\!-\!u}}}{\sqrt{1\!-\!\sqrt{u}}+\sqrt{1\!-\!\sqrt{1\!-\!u}}}\,.
%\lt[1+\frac{\sin(\frac{\pi}{4}-\frac{\phi}{2})}{\sin(\frac{\phi}{2})}\rt]^{-1}=.
\]%\end{equation}
\end{theorem}
\begin{remark}
\label{RemFK4pointsDisc} This formula is a special case of a hypergeometric formula
for crossings in a general FK model. In the Ising case it becomes algebraic and
furthermore can be rewritten in several ways.
It has an especially simple form for the crossing
probabilities
\[
p(\phi):=p(\dD;-e^{i\phi},e^{-i\phi},e^{i\phi},-e^{-i\phi})\quad\mathit{and}\quad
p(\textfrac{\pi}{2}-\phi)=%p(\dD;-e^{i\phi},e^{-i\phi},e^{i\phi},-e^{-i\phi})=
1-p(\phi)
\]
in the unit disc $\dD$ (clearly, the cross-ratio $u$ is equal to $\sin^2\phi$). Namely,
\[
\frac{p(\phi)} {p(\frac{\pi}{2}\!-\!\phi)} =
\frac{\sin\frac{\phi}{2}}{\sin(\frac{\pi}{4}-\frac{\phi}{2})}\quad \mathit{for}~\
\phi\in[0,\textfrac{\pi}{2}].
\]
Curiously, this macroscopic formula formally coincides with the relative weights
corresponding to two different possibilities of crossings inside microscopic rhombi (see
Fig.~\ref{Fig:FKDomain}A) in the critical FK-Ising model on isoradial graphs.
\end{remark}

\begin{proof}
We start with adding to our picture the ``external'' edge connecting $c^\mesh$ and $d^\mesh$
(see Fig.~\ref{Fig:FK4points}). Then, exactly as in Sect. \ref{SectFKIsing},
(\ref{FKfermionXiDef}) and (\ref{FKfermionDef}) allow us to define the s-holomorphic in
$\O^\mesh_\DS$ function $F^\mesh_{[cd]}:\O^\mesh_\DS\to\C$ such that
\begin{equation}
\label{4pointsBC} F^\mesh_{[cd]}(\z)\parallel (\t(\z))^{-\frac{1}{2}}
\end{equation}
on $\pa\O^\mesh_\DS$, where
\begin{equation}
\label{4pointsBCtau}
\begin{array}{lll}
\t(\z)=w_2(\z)\!-\!w_1(\z), & \z\in (a^\mesh b^\mesh)\cup(c^\mesh d^\mesh), &
w_{1,2}(\z)\in\G^*, \cr \t(\z)=u_2(\z)\!-\!u_1(\z), & \z\in (b^\mesh c^\mesh)\cup (d^\mesh
a^\mesh), & u_{1,2}(\z)\in\G,
\end{array}
\end{equation}
is the ``discrete tangent vector'' to $\pa\O^\mesh_\DS$ oriented from $a^\mesh/c^\mesh$ to
$b^\mesh/d^\mesh$ (see Fig.~\ref{Fig:FK4points}).

Note that $F^\mesh_{[cd]}(b^\mesh)= (2\mesh)^{-\frac{1}{2}}$ and $F^\mesh_{[cd]}(a^\mesh)=
(2\mesh)^{-\frac{1}{2}}\cdot e^{-\frac{i}{2}\wind(b^\mesh\rsa a^\mesh)}$, but
\[
F^\mesh_{[cd]}(d^\mesh) =
(2\mesh)^{-\frac{1}{2}}\cdot\frac{\rQ^\mesh}{\sqrt{2}\cdot\rP^\mesh\!+\!\rQ^\mesh} \cdot
e^{-\frac{i}{2}\wind(b^\mesh\rsa (c^\mesh \rsa d^\mesh))}
\]
(and similarly for $F^\mesh_{[cd]}(c^\mesh)$, since the interface passes through $d^\mesh$ if
and only if $b^\mesh$ is connected with $c^\mesh$, see Fig.~\ref{Fig:FK4points}).

Similarly, we can add an external edge $[a^\mesh d^\mesh]$ and construct another s-holomorphic
in $\O^\mesh_\DS$ function $F^\mesh_{[ad]}$ satisfying the same boundary conditions
(\ref{4pointsBC}). Arguing in the same way, we deduce that $F^\mesh_{[ad]}(b^\mesh)=
(2\mesh)^{-\frac{1}{2}}$, $F^\mesh_{[ad]}(c^\mesh)=(2\mesh)^{-\frac{1}{2}}\cdot
e^{-\frac{i}{2}\wind(b^\mesh\rsa c^\mesh)}$, and
\[
F^\mesh_{[ad]}(d^\mesh)=
(2\mesh)^{-\frac{1}{2}}\cdot\frac{\rP^\mesh}{\rP^\mesh\!+\!\sqrt{2}\cdot\rQ^\mesh} \cdot
e^{-\frac{i}{2}\wind(b^\mesh\rsa (a^\mesh \rsa d^\mesh))}
\]
(and similarly for $F^\mesh_{[ad]}(a^\mesh)$). Note that
\begin{equation}
\label{xWind} e^{-\frac{i}{2}\wind(b^\mesh\rsa (a^\mesh \rsa d^\mesh))} =
-e^{-\frac{i}{2}\wind(b^\mesh\rsa (c^\mesh\rsa d^\mesh))}\,.
\end{equation}

Let
\[
F^\mesh:=\frac{\rP^\mesh(\sqrt{2}\rP^\mesh\!+\!\rQ^\mesh)\cdot F^\mesh_{[cd]} +
\rQ^\mesh(\rP^\mesh\!+\!\sqrt{2} \rQ^\mesh)\cdot F^\mesh_{[ad]}}
{\rP^\mesh(\sqrt{2}\rP^\mesh\!+\!\rQ^\mesh) + \rQ^\mesh(\rP^\mesh\!+\!\sqrt{2} \rQ^\mesh)}\,.
\]
Then, $F^\mesh$ also satisfies boundary conditions (\ref{4pointsBC}), (\ref{4pointsBCtau})
and, in view of (\ref{xWind}),
\begin{equation}
\label{F4Pbdac=}
\begin{array}{lll}
F^\mesh(a^\mesh)= (2\mesh)^{-\frac{1}{2}}\cdot\rA^\mesh\cdot e^{-\frac{i}{2}\wind(b^\mesh\rsa
a^\mesh)}\,, &~~& F^\mesh(b^\mesh)= (2\mesh)^{-\frac{1}{2}}\,, \vphantom{\big|_|}\cr
F^\mesh(c^\mesh)= (2\mesh)^{-\frac{1}{2}}\cdot\rC^\mesh\cdot e^{-\frac{i}{2}\wind(b^\mesh\rsa
c^\mesh)}\,, && F^\mesh(d^\mesh)=0\,, \vphantom{\big|^|}
\end{array}
\end{equation}
where
\[
\rA^\mesh=\frac{\rP^\mesh(\sqrt{2}\rP^\mesh\!+\!\rQ^\mesh) + \rQ^\mesh\rP^\mesh}
{\rP^\mesh(\sqrt{2}\rP^\mesh\!+\!\rQ^\mesh) + \rQ^\mesh(\rP^\mesh\!+\!\sqrt{2}
\rQ^\mesh)}\,,\quad \rC^\mesh=\frac{\rP^\mesh\rQ^\mesh + \rQ^\mesh(\rP^\mesh\!+\!\sqrt{2}
\rQ^\mesh)} {\rP^\mesh(\sqrt{2}\rP^\mesh\!+\!\rQ^\mesh) + \rQ^\mesh(\rP^\mesh\!+\!\sqrt{2}
\rQ^\mesh)}\,.
\]
%(note that $(\rA^\mesh)^2+(\rC^\mesh)^2=1$).

Since $F^\mesh$ is s-holomorphic and satisfies the boundary conditions (\ref{4pointsBC}),
(\ref{4pointsBCtau}), we can define $H^\mesh:=\int^\mesh(F^\mesh(z))^2d^\mesh z$ and use the
``boundary modification trick'' (see Section~\ref{SectBModTrick}). Then, (\ref{F4Pbdac=})
implies
\begin{equation}
\begin{array}{lllll}
H^\mesh_\G = 0 & \mathrm{on}~~(a^\mesh b^\mesh)_{\wt\G}, && H^\mesh_{\G^*} = 0 &
\mathrm{on}~~(a^\mesh b^\mesh), \cr H^\mesh_\G = 1 & \mathrm{on}~~(b^\mesh c^\mesh), &&
H^\mesh_{\G^*} = 1 & \mathrm{on}~~(b^\mesh c^\mesh)_{\wt\G^*}, \cr H^\mesh_\G = \vk^\mesh &
\mathrm{on}~~(c^\mesh d^\mesh)\cup(d^\mesh a^\mesh)_{\wt\G}, &\phantom{}& H^\mesh_{\G^*} =
\vk^\mesh & \mathrm{on}~~(c^\mesh d^\mesh)_{\wt\G^*}\cup (d^\mesh a^\mesh), \cr
\end{array}
\end{equation}
where
\begin{equation}
\label{Vk=P} \vk^\mesh=(\rA^\mesh)^2=1-(\rC^\mesh)^2=
\lt[\frac{(t^\mesh)^2\!+\!\sqrt{2}t^\mesh}{(t^\mesh)^2\!+\!\sqrt{2}t^\mesh\!+\!1}\rt]^2,\qquad
t^\mesh=\frac{\rP^\mesh}{\rQ^\mesh}=\frac{\rP^\mesh}{1\!-\!\rP^\mesh}\,.
\end{equation}

\smallskip

Suppose that $|\rP^\mesh(\O^\mesh_\DS)-p(\O^\mesh)|\ge\ve_0>0$ for some domains
$\O^\mesh_\DS=(\O^\mesh_\DS;a^\mesh,b^\mesh,c^\mesh,d^\mesh)$ with $\mesh\to 0$. Passing to a
subsequence (exactly as in Section~\ref{SectFKobservable}), we may assume that
\[
(\O^\mesh;a^\mesh,b^\mesh,c^\mesh,d^\mesh)\CaraTo (\O;a,b,c,d),\quad \vk^\mesh\to\vk\in [0,1],
\]
\[
H^\mesh\rra H,\qquad \mathrm{and}\qquad F^\mesh\rra F=\sqrt{2i\pa H},
\]
uniformly on compact subsets, for some harmonic function $H:\O\to\R$. It follows from our
assumptions that $B(0,r)\ss\O\ss B(0,R)$ and either $a\ne b$, $c\ne d$ or $b\ne c$, $d\ne a$.

We begin with the main case, when the limiting quadrilateral $(\O;a,b,c,d)$ is non-degenerate
and $0<\varkappa<1$. As in Section \ref{SectFKobservable}, we see that
\begin{equation}
\label{FK4Hbc} H = 0~\mathrm{on}~(ab),\quad H=1~\mathrm{on}~(bc) ~~\mathrm{and}~\
H=\vk~\mathrm{on}~(cd)\cup(da).
\end{equation}
Consider the conformal mapping $\Phi$ from $\O$ onto the slit strip $[\R\ts(0;1)]\setminus
(i\vk-\infty;i\vk]$ such that $a$ is mapped to ``lower'' $-\infty$, $b$ to $+\infty$ and $c$
to ``upper'' $-\infty$ (note that such a mapping is uniquely defined). Then, the imaginary
part of $\Phi$ is harmonic and has the same boundary values as $H$, so we conclude that
$H=\Im\Phi$. We prove that
\begin{quotation}
\begin{center}{\it $d$ is mapped exactly to the tip $i\vk$}.\end{center}
\end{quotation}
Then, $\vk$ can be uniquely determined from the conformal modulus of $(\O;a,b,c,d)$.

\begin{figure}
\centering{\begin{minipage}[b]{0.45\textwidth}
\centering{\includegraphics[width=\textwidth]{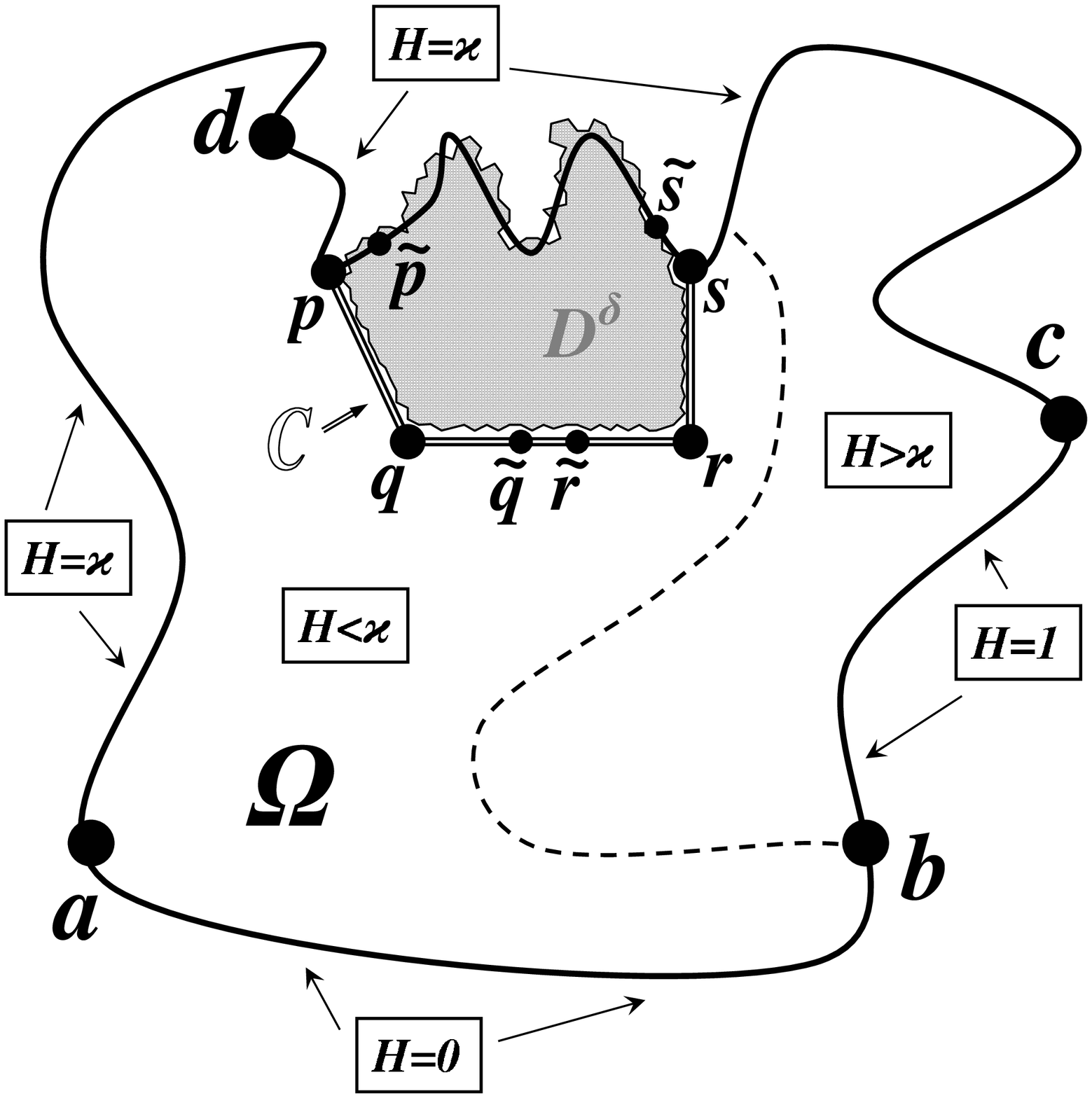}}

\bigskip

\textsc{(A)}
\end{minipage}
\hskip 0.09\textwidth
\begin{minipage}[b]{0.45\textwidth}
\centering{\includegraphics[width=\textwidth]{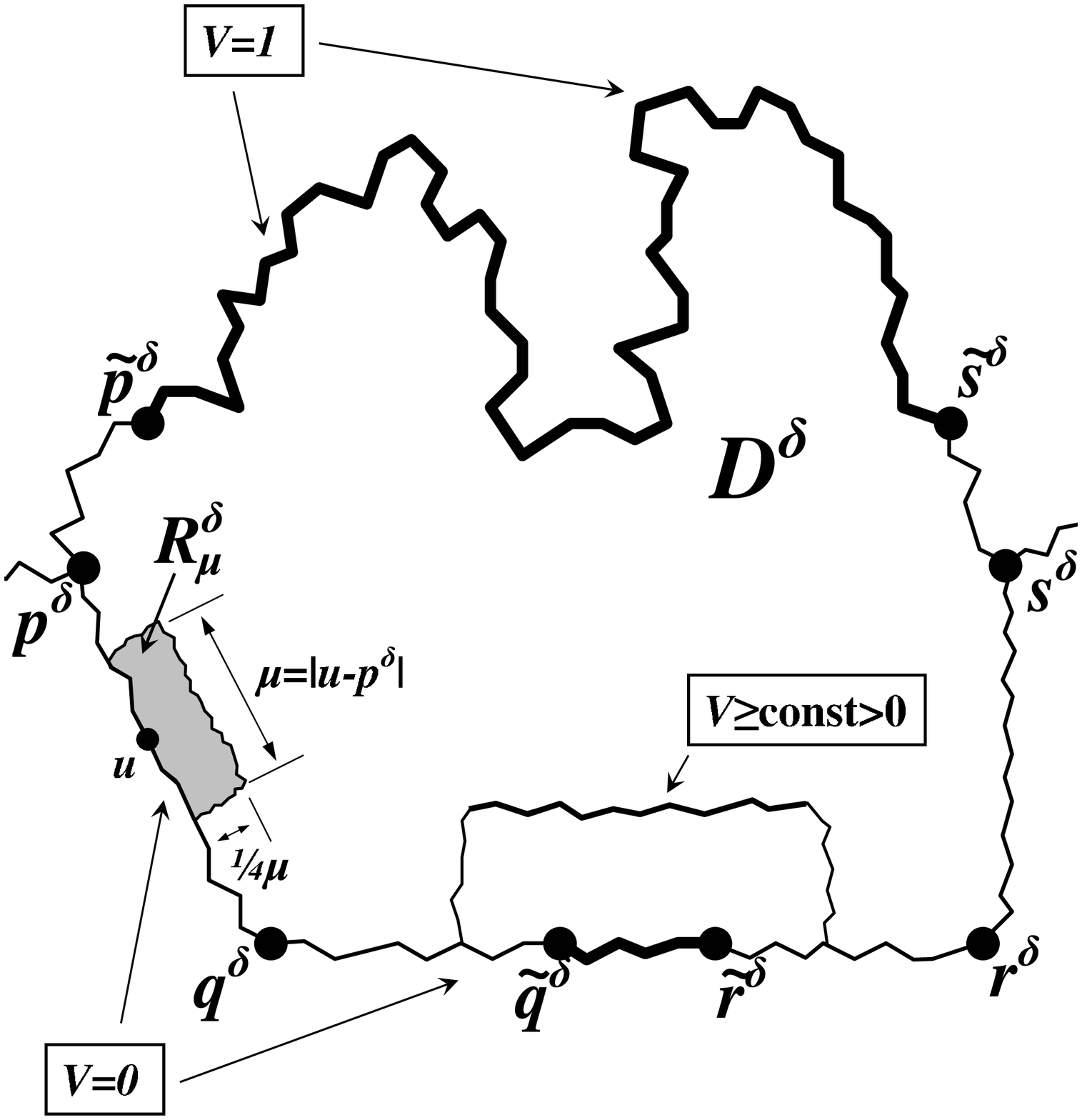}}

\bigskip

\textsc{(B)}
\end{minipage}}
\caption{\label{Fig:FK4lim} \textsc{(A)} If $d$ is mapped on the lower bank of the cut, then
\mbox{$H<\vk$} somewhere near $(cd)$. The contour $C=[p;q]\cup[q;r]\cup [r;s]$ is chosen so
that $H<\vk$ on $C$, $\dist(q;\pa\O)=|q\!-\!p|$ and $\dist(r;\pa\O)=|r\!-\!s|$.
\textsc{(B)}~Since $V^\mesh=0$ on $C^\mesh$, we have $\pa_n^\mesh V^\mesh\le 0$ everywhere on
$C^\mesh$. Moreover, $\pa^\mesh_n V^\mesh\le\const <0$ on $(\wt{q}^\mesh\wt{r}^\mesh)$ and
$\pa^\mesh_n V^\mesh =O([\dist(u;\pa\O^\mesh)]^{\b-1})$ near $p^\mesh$,~$s^\mesh$ (here
$\pa_n^\mesh$ denotes the discrete derivative in the \emph{outer} normal direction).}
\end{figure}

Suppose that above is not the case, and  $d$ is mapped, say, on the lower bank of the cut. It
means that $H<\vk$ near some (close to $d$) part of the boundary arc $(cd)$. Then, there
exists a (small) contour $C=[p;q]\cup[q;r]\cup[r;s]\ss\O$ such that $H<\vk$ everywhere on $C$,
\[
s,p\in(cd)\ss\pa\O,\quad \dist(q;\pa\O)=|q\!-\!p|~~\mathrm{and}~~\dist(r;\pa\O)=|r\!-\!s|
\]
(see Fig.~\ref{Fig:FK4lim}). Denote by $D\ss\O$ the part of $\O$ lying inside~$C$. For
technical purposes, we also fix some intermediate points $\wt{s},\wt{p}\in (sp)\ss(cd)$ and
$\wt{q},\wt{r}\in [q;r]$ (see Fig.~\ref{Fig:FK4lim}). For sufficiently small $\mesh$, we can
find discrete approximations $s^\mesh,\wt{s}^\mesh,\wt{p}^\mesh,p^\mesh\in (c^\mesh
d^\mesh)\ss~\pa\O^\mesh_{\wt{\G}}$ and
$q^\mesh,\wt{q}^\mesh,\wt{r}^\mesh,r^\mesh\in\O^\mesh_\G$ to these points such that the
contour $[p^\mesh;q^\mesh]\cup [q^\mesh;r^\mesh]\cup[r^\mesh;s^\mesh]$ approximates $C$.
Denote by $D^\mesh\ss\O^\mesh$ the part of $\O^\mesh$ lying inside $C$ and by
$D^\mesh_\G\ss\O^\mesh_\G$ the set of all ``black'' vertices lying in $D^\mesh$ and their
neighbors.

Let $\wt H^\mesh_\G:=H^\mesh_\G-\vk^\mesh$ and
$V^\mesh:=\dhm\mesh{\,\cdot\,}{(\wt{s}^\mesh\wt{p}^\mesh)}{D^\mesh}$. Since $H^\mesh_\G$ is
subharmonic and $V^\mesh\ge 0$ is harmonic, the discrete Green's formula gives
\[
\sum_{u\in\pa D^\mesh} [(\wt{H}^\mesh_\G(u)\!-\!\wt{H}^\mesh_\G(u_{\mathrm{int}}))V^\mesh(u) -
(V^\mesh(u)\!-\!V^\mesh(u_{\mathrm{int}}))\wt{H}^\mesh_\G(u)] \tan\theta_{uu_{\mathrm{int}}}
\ge 0.
\]
Note that $\wt{H}^\mesh_\G\equiv 0$ on $(s^\mesh p^\mesh)_{\wt{\G}}$ and $V^\mesh(u)\equiv 0$
on $C^\mesh:=\pa D^\mesh_\G \setminus (s^\mesh p^\mesh)_{\wt{\G}}$. Thus,
\begin{equation}
\label{4Px0} \sum_{u\in C^\mesh} V^\mesh(u_{\mathrm{int}})\wt{H}^\mesh_\G(u)
\tan\theta_{uu_{\mathrm{int}}} \ge\!\!\!\! \sum_{u\in (s^\mesh p^\mesh)}
(\wt{H}^\mesh_\G(u_{\mathrm{int}})\!-\!\wt{H}^\mesh_\G(u))V^\mesh(u)
\tan\theta_{uu_{\mathrm{int}}} \ge 0,
\end{equation}
since $H^\mesh_\G(u_{\mathrm{int}})\!\ge\!H^\mesh_\G(u)$ everywhere on $(s^\mesh p^\mesh)$ due
to the boundary conditions (\ref{4pointsBC}),~(\ref{4pointsBCtau}). On the other hand, on most
of $C^\mesh$, we have $\wt{H}^\mesh_\G<0$ (since $H\!-\!\vk=\lim_{\mesh\to0}\wt{H}^\mesh_\G
<0$ on $C$ by assumption), and $V^\mesh(u_{\mathrm{int}})\!\ge\!0$ everywhere on $C^\mesh$,
which gives a contradiction. Unfortunately, we cannot immediately claim that
$\wt{H}^\mesh_\G<0$ near the boundary, so one needs to prove that the neighborhoods of
$p^\mesh$ and $s^\mesh$ cannot produce an error sufficient to compensate this difference of
signs.

More accurately, it follows from the uniform convergence $V^\mesh\rra
\hm{\,\cdot\,}{(\wt{s}\wt{p})}{D}>0$ on compacts inside $D$ and Lemma~\ref{LemmaDhmR} that
$V^\mesh(u_{\mathrm{int}})\ge \const(D)\cdot \mesh$ everywhere on
$(\wt{q}^\mesh\wt{r}^\mesh)\ss C^\mesh$, so, for small enough $\mesh$,
\begin{equation}
\label{4Px1} \sum_{u\in (\wt{q}^\mesh\wt{r}^\mesh)}
V^\mesh(u_{\mathrm{int}})\wt{H}^\mesh_\G(u) \tan\theta_{uu_{\mathrm{int}}}\le -\const(D,H)<0.
\end{equation}
Thus, it is sufficient to prove that the neighborhoods of $p^\mesh$ and $s^\mesh$ cannot
compensate this negative amount which is independent of $\mesh$. Let $u\in (p^\mesh
q^\mesh)\ss C^\mesh$ and $\m\!=\!\dist(u;p^\mesh)\!=\!\dist(u;\pa\O^\mesh)$ be small. Denote
by $R^\mesh_\m$ the discretization of the $\m\ts\textfrac{1}{4}\m$ rectangle near $u$ (see
Fig.~\ref{Fig:FK4lim}). Due to Lemma \ref{LemmaDhmR}, we have
\[
\dhm\mesh{u_{\mathrm{int}}}{\pa R^\mesh_\m\setminus (p^\mesh q^\mesh)}{R^\mesh_\m} = O(\mesh
\m^{-1}).
\]
Furthermore, for each $v\in \pa R^\mesh_\m\setminus (p^\mesh q^\mesh)$, Lemma
\ref{WeakBeurling} gives
\[
\dhm\mesh{v}{(\wt{s}^\mesh\wt{p}^\mesh)}{D^\mesh} = O(\m^\b)\quad \mathrm{uniformly~on~}\
\pa R^\mesh_\m.
\]
Hence,
\[
V^\mesh(u_{\mathrm{int}})=\dhm\mesh{u_{\mathrm{int}}}{(\wt{s}^\mesh \wt{p}^\mesh)} {D^\mesh}
\le \const(D)\cdot \mesh \m^{-(1-\b)}.
\]
Recalling that $H^\mesh=O(1)$ by definition
and summing, for any $p^\mesh_\m\in (p^\mesh
q^\mesh)\ss C^\mesh$ sufficiently close to $p^\mesh$ we obtain
\begin{align}
\label{4Px2} \sum_{u\in (p^\mesh p^\mesh_\m)} V^\mesh(u_{\mathrm{int}})\wt{H}^\mesh_\G(u)
\tan\theta_{uu_{\mathrm{int}}}&\le \const(D)\cdot\!\!\!\sum_{u\in (p^\mesh p^\mesh_\m)}
\mesh\cdot (\dist(u;p^\mesh))^{-(1-\b)}\cr&\le \const(D) \cdot (\dist(p^\mesh_\m;p^\mesh))^\b
\end{align}
(uniformly with respect to $\mesh$). The same estimate holds near $s^\mesh$. Taking into
account $\wt{H}^\mesh_\G(u)<0$ which holds true (if $\mesh$ is small enough) for all $u\in
C^\mesh$ lying $\m$-away from $p^\mesh$, $s^\mesh$, we deduce from (\ref{4Px0}), (\ref{4Px1})
and (\ref{4Px2}) that $0<\const(D,H)\le \const(D)\cdot \m^\b$ for any $\m>0$ (and sufficiently
small $\mesh\le\mesh(\m)$), arriving at a contradiction.

\smallskip

All ``degenerate'' cases can be dealt with in the same way:
\begin{itemize}
\item if the quadrilateral $(\O;a,b,c,d)$ is non-degenerate, then
\begin{itemize}
\item if $\vk=1$, then $H<\vk$ near some part of $(cd)$, which is impossible;
\item if $d$ is mapped onto the upper bank of the slit or $\vk=0$, then $H>\vk$ near some part
of $(ad)$, which leads to a contradiction via the same arguments as above;
\end{itemize}
\item if $b=c$ (and so $c\ne d$), then $\vk=0$ since otherwise $H<\vk$ everywhere near $(cd)$ due to
boundary conditions (\ref{FK4Hbc});
\item if $d=a$ (and so $c\ne d$, $a\ne b$), then again $\vk=0$ since otherwise
(\ref{FK4Hbc}) implies $H<\vk$ near some part of $(cd)=(ca)$ close to $a$;
\item finally, $a=b$ or $c=d$ lead to $\vk=1$ (otherwise $H>\vk$ near some part of
$(da)$).
\end{itemize}
Thus, $d$ is mapped to the tip and so $\vk=\vk(\O;a,b,c,d)$ is uniquely determined by the
conformal modulus of the quadrilateral ($\vk$ is either $0$ or $1$ in degenerate cases).
Recall that $\vk^\mesh=\xi(P^\mesh)$, where the bijection $\xi:[0,1]\to [0,1]$ is given by
(\ref{Vk=P}). Let
\[
p(\O;a,b,c,d):=\xi^{-1}(\vk(\O;a,b,c,d)).
\]
Then (since $\vk(\cdot)$ is Carath\'eodory stable) both
$\rP^\mesh(\O^\mesh_\DS)=\xi^{-1}(\vk^\mesh)$ and $p(\O^\mesh)$ tend to $p(\O)=\xi^{-1}(\vk)$
as $\mesh=\mesh_k\to 0$, which contradicts to
$|\rP^\mesh(\O^\mesh_\DS)-p(\O^\mesh)|\ge\ve_0>0$.

\smallskip

Finally, the simple calculation for the half-plane $(\H;0,1-u,1,\infty)$ gives
\[
H(z)\equiv u+\tfrac{1}{\pi}(-\arg [z-(1\!-\!u)] + u\arg z + (1\!-\!u)\arg[z-1]),\quad z\in\H.
\]
Hence, $\vk(\H;0,1-u,1,\infty)=u$ and $p=\xi^{-1}(u)$ which coincides with (\ref{CrossP}).
\end{proof}

\begin{remark}
In fact, above we have shown that the ``$(\t(z))^{-\frac{1}{2}}$'' boundary condition
(\ref{RH-1/2bc}) reformulated in the form $\pa_n^\mesh H^\mesh\le 0$ remains valid in the
limit as $\mesh\to 0$. Namely,
\begin{quotation}
let a sequence of discrete domains $\O^\mesh$ converge to some limiting $\O$ in the
Carath\'eodory topology, while s-holomorphic functions $F^\mesh$ defined on $\O^\mesh$ satisfy
(\ref{RH-1/2bc}) on arcs $(s^\mesh p^\mesh)$ converging to some boundary arc
\mbox{$(sp)\ss\pa\O$}. Let their integrals $H^\mesh=\Im\int^\mesh (F^\mesh(z))^2d^\mesh z$ are
defined so that they are uniformly bounded near $(s^\mesh p^\mesh)$ and their (constant)
values $\vk^\mesh$ on $(s^\mesh p^\mesh)$ tends to some $\vk$ as $\mesh\to 0$. Then, if
$H^\mesh$ converge to some (harmonic) function $H$ inside $\O$, one has $\pa_n H^\mesh\le 0$
on $(sp)$ in the following sense: there is no point $\zeta\in(sp)$ such that $H<\vk$ in a
neighborhood of $\zeta$.
\end{quotation}
The proof mimics the corresponding part of the proof of Theorem~\ref{ThmFK4points}. Since the
boundary conditions (\ref{RH-1/2bc}) are typical for holomorphic observables in the critical
Ising model, this statement eventually can be applied to all such observables.
\end{remark}

\renewcommand{\thesection}{A}
\section{Appendix}
\setcounter{equation}{0}
\renewcommand{\theequation}{A.\arabic{equation}}

\subsection{Estimates of the discrete harmonic measure}
\label{SectDHMestim} Here we formulate uniform estimates for the discrete harmonic measure on
isoradial graphs which were used above.

\begin{lemma}[\bf exit probabilities in the disc]
\label{HmInDisc} Let $u_0\!\in\!\G$, $r\ge \mesh$ and $a\in\pa B^\mesh_\G(u_0,r)$. Then,
\[
\dhm{\mesh}{u_0}{\{a\}}{B^\mesh_\G(u_0,r)}\asymp {\mesh}/{r}\,.
\]
\end{lemma}

\begin{proof}
See \cite{buckling-2008} (or \cite{chelkak-smirnov-dca} Proposition A.1). The proof is based on the asymptotics
(\ref{Gasympt}) of the free Green's function.
\end{proof}

\begin{lemma}[\bf weak Beurling-type estimate]
\label{WeakBeurling} There exists an absolute constant $\b>0$ such that for any simply
connected discrete domain $\O^\mesh_\G$, point $u\in\Int\O^\mesh_\G$ and  some part of the
boundary $E\subset\pa\O^\mesh_\G$ we have
\[
\dhm{\mesh}{u}{E}{\O^\mesh_\G}\le \const\cdot
\lt[\frac{\dist(u;\pa\O^\mesh_\G)}{\dist_{\O^\mesh_\G}(u;E)}\rt]^\b.
\]
%\quad \mathit{and} \quad \dhm{\mesh}{u}{E}{\O^\mesh_\G}\le \const\cdot \lt[\frac{\diam
%E}{\dist_{\O^\mesh_\G}(u;E)}\rt]^\b.
%\]
%Above we set $\diam E:=\mesh$, if $E$ consists of one point.
Here $\dist_{\O^\mesh_\G}$ denotes the distance inside $\O^\mesh_\G$.
\end{lemma}

\begin{proof}
See \cite{chelkak-smirnov-dca} Proposition 2.11. The proof is based on the uniform bound of
the probability that the random walk on $\G$ crosses the annulus without making the full turn
inside.
\end{proof}
Finally, let $R^\mesh_\G(s,t)\ss\G$ denote the discretization of the open rectangle
\[R(s,t)=(-s;s)\ts(0;t)\ss\C,\quad s,t>0;\] $b^\mesh\in\pa\H^\mesh_\G$ be the boundary vertex
closest to $0$; and $L^\mesh_\G(s)$, $U^\mesh_\G(s,t)$, $V^\mesh_\G(s,t)$ be the lower, upper
and vertical parts of the boundary $\pa R^\mesh_\G(s,t)$, respectively.

\begin{lemma}[\bf exit probabilities in the rectangle]
\label{LemmaDhmR} Let $s\ge 2t$ and $t\ge 2\mesh$. Then, for any $v^\mesh=x\!+\!iy\in
R^\mesh_\G(s,t)$, one has
\[
\frac{y\!+\!2\mesh}{t\!+\!2\mesh}\ge \dhm\mesh{v^\mesh}{U^\mesh_\G (s,t)}{R^\mesh_\G(s,t)}\ge
\frac{y}{t\!+\!2\mesh}\,-\,\frac{x^2+(y\!+\!2\mesh)(t\!+\!2\mesh\!-\!y)}{s^2}
\]
and
\[
\dhm\mesh{v^\mesh}{V^\mesh_\G (s,t)}{R^\mesh_\G(s,t)}\le
\frac{x^2+(y\!+\!2\mesh)(t\!+\!2\mesh\!-\!y)}{s^2}\,.
\]
\end{lemma}
\begin{proof}
See \cite{chelkak-smirnov-dca} Lemma 3.17. The claim easily follows from the maximum principle
for discrete harmonic functions.
\end{proof}

\subsection{Lipschitzness of discrete harmonic and discrete holomorphic functions}

\begin{proposition}[\bf discrete Harnack Lemma]
\label{PropHarnack} Let $u_0\in\G$ and $H:B^\mesh_\G(u_0,R)\to\R$ be a nonnegative discrete
harmonic function. Then,

\noindent (i)\phantom{i} for any $u_1,u_2\in B^\mesh_\G(u_0,r)\ss\Int B^\mesh_\G(u_0,R)$,
\[
\exp\lt[-\const\cdot\frac{r}{R-r}\rt] \le \frac{H(u_2)}{H(u_1)}\le
\exp\lt[\const\cdot\frac{r}{R-r}\rt];
\]
(ii) for any $u_1\sim u_0$,
\[
|H(u_1)-H(u_0)|\le \const\cdot{\mesh H(u_0)}/{R}\,.
\]
\end{proposition}

\begin{proof}
See \cite{buckling-2008} (or \cite{chelkak-smirnov-dca} Proposition 2.7). The proof is based
on the asymptotics (\ref{Gasympt}) of the free Green's function.
\end{proof}

\begin{corollary}[\bf Lipschitzness of harmonic functions]
\label{CorHarnack} Let $H$ be discrete harmonic in $B^\mesh_\G(u_0,R)$ and $u_1,u_2\in
B^\mesh_\G(u_0,r)\ss\Int B^\mesh_\G(u_0,R)$. Then
\[
|H(u_2)-H(u_1)|\le \const\cdot\frac{M|u_2\!-\!u_1|}{R-r},\quad \mathit{where}\quad
M=\max_{B^\mesh_\G(u_0,R)}|H(u)|.
\]
\end{corollary}

In order to formulate the similar result for discrete holomorphic functions we need some
preliminary definitions. Let $F$ be defined on some part of $\DS$. Taking the real and
imaginary parts of $\dopa F$ (see (\ref{DopaDef})), it is easy to see that $F$ is holomorphic
if and only if both functions
\[
[\cB F](z):=\Pr\left[F(z) \,;\,
\ol{u_1(z)\!-\!u_2(z)}\,\right],\quad [\cW F](z):=\Pr\left[F(z) \,;\,
\ol{w_1(z)\!-\!w_2(z)}\,\right]
\]
are holomorphic, where $u_{1,2}(z)\in\G$ and $w_{1,2}(z)\in\G^*$ are the black and  white
neighbors of $z\in\DS$, respectively (note that $F=\cB F +\cW F$).

\begin{figure}
\centering{\begin{minipage}[b]{0.4\textwidth}
\centering{\includegraphics[width=\textwidth]{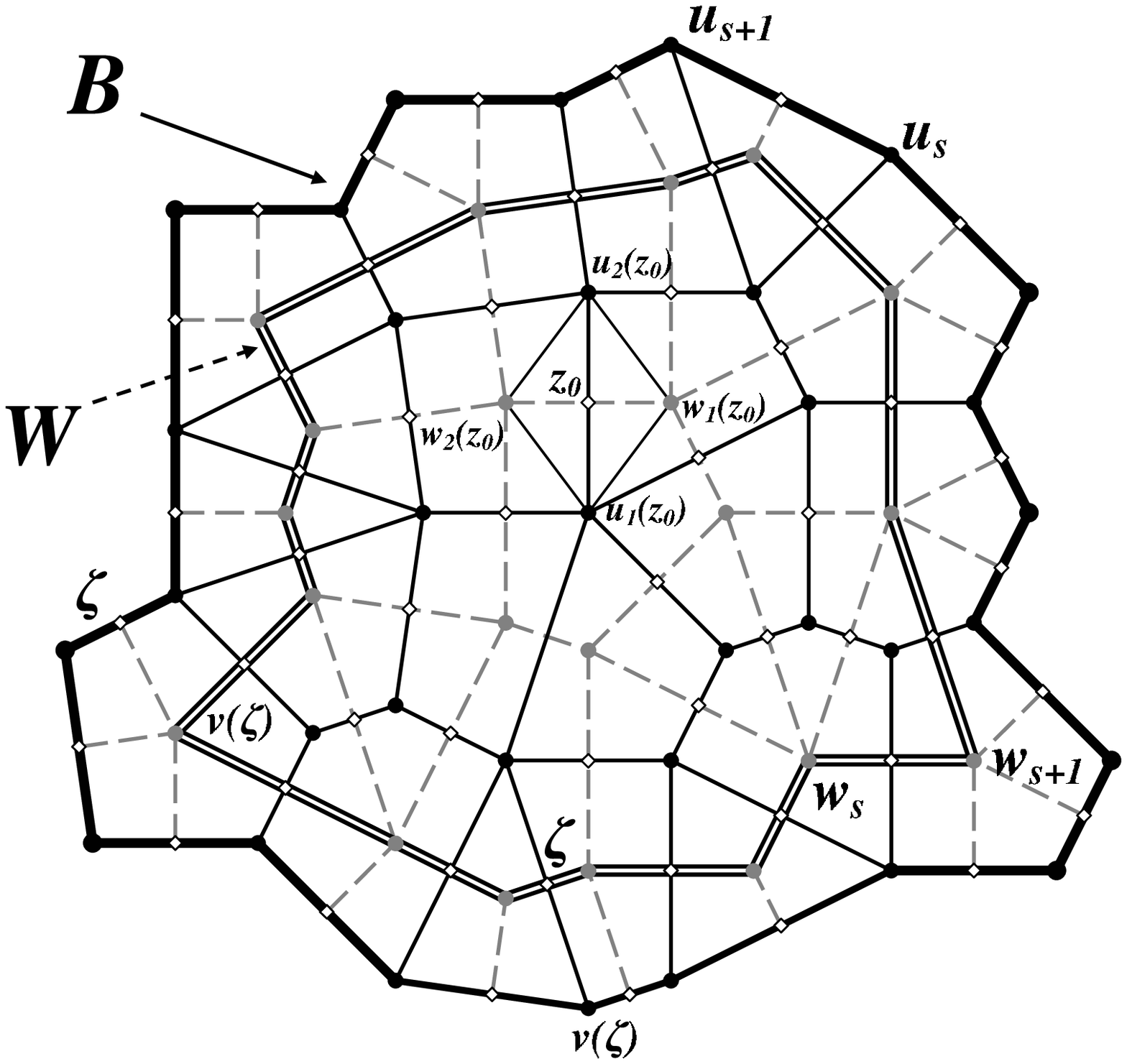}}
\end{minipage}
\hskip 0.1\textwidth
\begin{minipage}[b]{0.44\textwidth}
Let $\O^\d_\G$ be a bounded simply connected discrete domain. We denote by $B=u_0u_1..u_n$,
$u_s\in\G$, its closed polyline boundary enumerated in the counterclockwise order, and by
$W=w_0w_1..w_m$, $w_s\in\G^*$, the closed polyline path passing through the centers of all
faces touching $B$ enumerated in the counterclockwise order. In order to write down the
discrete Cauchy formula (see Lemma~\ref{CauchyFormula}), one needs to ``integrate'' over both
$B$ and $W$.
\end{minipage}}
\caption{\label{Fig:Cauchy} Discrete Cauchy formula, notations (see Lemma
\ref{CauchyFormula}).}
\end{figure}

Let $\O^\d_\G$ be a bounded simply connected discrete domain. For a function $G$ defined on
both ``boundary contours'' $B$,$W$ (see Fig.~\ref{Fig:Cauchy}), we set
\[
\oint^\mesh_{B\cup W} G(\z)d^\mesh \z := \sum_{s=0}^{n-1}
G\left({\textstyle\frac{1}{2}}(u_{s+1}\!+\!u_s)\right)(u_{s+1}\!-\!u_s)
 + \sum_{s=0}^{m-1} G\left({\textstyle\frac{1}{2}}(w_{s+1}\!+\!w_s)\right)(w_{s+1}\!-\!w_s).
\]

\begin{lemma}[\bf Cauchy formula]
\label{CauchyFormula} (i) There exists a function (discrete Cauchy kernel)
$K(\,\cdot\,;\,\cdot\,):\L\ts\DS\to\C$, $K(v,z)=O(|v\!-\!z|^{-1})$, such that for any discrete
holomorphic function $F:\O^\mesh_\DS\to\C$ and $z_0\in\O^\d_\DS\setminus(B\cup W)$ the
following holds true:
\[
F(z_0)=\frac{1}{4i}\oint^\mesh_{B\cup W} K(v(\z);z_0)F(\z)d^\mesh \z,
\]
where $\z\sim v(\z)\in W$, if $\z\in B$, and $\z\sim v(\z)\in B$, if $\z\in W$ (see
Fig.~\ref{Fig:Cauchy}B).

\smallskip

\noindent (ii) Moreover, if $F=\cB F$, then
\[
F (z_0) = \Pr\lt[\frac{1}{2\pi i}\oint^\mesh_{B\cup W}\frac{F(\z)d^\mesh \z}{\z\!-\!z_0}~;\
\ol{u_1(z_0)\!-\!u_2(z_0)}\,\rt] + O\lt(\frac{M\mesh L} {d^2}\rt),
\]
where $d=\dist(z_0,W)$, $M=\max_{z\in B\cup W}|F(z)|$ and $L$ is the length of $B\cup W$. The
similar formula (with ${w_1(z_0)\!-\!w_2(z_0)}$ instead of ${u_1(z_0)\!-\!u_2(z_0)}$) holds
true, if $F=\cW F$.
\end{lemma}

\begin{proof}
See \cite{chelkak-smirnov-dca} Proposition 2.22 and Corollary 2.23. The proof is based on the
discrete integration by parts and asymptotics of the discrete Cauchy kernel
$K(\,\cdot\,;\,\cdot\,)$ proved by R.Kenyon in \cite{kenyon-operators}.
\end{proof}

\begin{corollary}[\bf Lipschitzness of holomorphic functions]
\label{LipHol} Let $F:B^\mesh_\DS(z_0,R)\to\C$ be discrete holomorphic. Then there exist
$A,B\in\C$ such that
\[
F(z)=\Pr[\,A\,;\,\ol{u_1(z)-u_2(z)}]+\Pr[\,B\,;\,\ol{w_1(z)-w_2(z)}]+O(Mr/(R\!-\!r)),
\]
where $M=\max_{z\in B^\mesh_\DS(z_0,R)}|F(z)|$, for any $z$ such that $|z\!-\!z_0|\le r<R$.
\end{corollary}

\begin{proof}
Namely,
\[
A=\frac{1}{2\pi i}\oint^\mesh_{B\cup W}\frac{[\cB F](\z)d^\mesh \z}{\z\!-\!z_0}\quad
\mathrm{and} \quad B=\frac{1}{2\pi i}\oint^\mesh_{B\cup W}\frac{[\cW F](\z)d^\mesh
\z}{\z\!-\!z_0}.\qedhere
\]
\end{proof}

\subsection{Estimates of the discrete Green's function}

Here we prove two technical lemmas which were used in Section~\ref{SectRegSHol}. Recall that the
Green's function $G_{\O^\mesh_\G}(\cdot;u):\O^\mesh_\G\to\R$, $u\in\O^\mesh_\G\ss\G$, is the
(unique) discrete harmonic in $\O^\mesh_\G\setminus\{u\}$ function such that
$G_{\O^\mesh_\G}=0$ on the boundary $\pa\O^\mesh_\G$ and $\weightG{u}\cdot[\D^\mesh
G_{\O^\mesh_\G}](u)=1$. Clearly,
\[
G_{\Omega^\mesh_\G}^{~}= G_\G-G_{\Omega^\mesh_\G}^*,
\]
where $G_\G$ is the free Green's function and $G_{\Omega^\mesh_\G}^*$ is the unique discrete
harmonic in $\O^\mesh_\G$ function that coincides with $G_\G$ on the boundary
$\pa\O^\mesh_\G$. It is known that $G_\G$ satisfies (\cite{kenyon-operators}, see also \cite{chelkak-smirnov-dca}
Theorem 2.5) the asymptotics
\begin{equation}
\label{Gasympt} G_\G(v;u)= \frac{1}{2\pi}\log |v\!-\!u| +
O\lt(\frac{\d^2}{|v\!-\!u|^2}\rt),\quad v\ne u.
\end{equation}

\begin{lemma}
\label{Gestimate} Let $B^\mesh_\G=B^\mesh_\G(z_0,r)\ss\G$, $r\ge \const\cdot\mesh$, be the
discrete disc, $u\in B^\mesh_\G$ be such that $|u\!-\!z_0|\le\frac{3}{4}r$ and
$G=G_{B^\mesh_\G}(\cdot;u):B^\mesh_\G\to \R$ be the corresponding discrete Green's function.
Then
\[
\|G\|_{1,B^\mesh_\G}={\sum_{v\in B^\mesh_\G}}\weightG{v}|G(v)|\ge \const \cdot r^2.
\]
\end{lemma}

\begin{proof}
It immediately follows from (\ref{Gasympt}) that
\begin{equation}
\label{xEstim} (2\pi)^{-1}\log(\tfrac{1}{4}r) + O({\mesh^2}/{r^2}) \le
G^*_{B^\mesh_\G}(\cdot;u) \le (2\pi)^{-1}\log(\tfrac{7}{4}r) + O(\mesh^2/r^2)
\end{equation}
on the boundary $\pa B^\mesh_\G$, and so inside $B^\mesh_\G$. For $v\in B^\mesh_\G$ such that
$|v\!-\!u|\le\frac{1}{8}r$, this gives
\[
G(v) = G_\G(v;u) - G^*_{B^\mesh_\G}(v;u) \le -(2\pi)^{-1}\log 2+ O({\mesh^2}/{r^2})\le
-\const.
\]
Thus, $\|G\|_{1,B^\mesh_\G}\ge \|G\|_{1,B^\mesh_\G(u,\frac{1}{8}r)}\ge \const \cdot r^2$.
\end{proof}

\begin{lemma}
\label{Ggradestimate} Let $B^\mesh_\G=B^\mesh_\G(z_0,r)\ss\G$, $r\ge \const\cdot\mesh$, be the
discrete disc, $u\in B^\mesh_\G$ and $G=G_{B^\mesh_\G}(\cdot;u):B^\mesh_\G\to \R$ be the
corresponding discrete Green's function. Then
\[
\|\dpa G\|_{1,B^\mesh_\DS(z_0,\frac{2}{3}r)}={\sum_{z\in
B^\mesh_\DS(z_0,\frac{2}{3}r)}}\weightDS{z}|[\dpa G](z)|\le \const \cdot r.
\]
\end{lemma}

\begin{proof}
Let $|u-z_0|\le \frac{3}{4}r$. It easily follows from (\ref{Gasympt}) and
Corollary~\ref{CorHarnack} applied in the disc $B^\mesh_\DS(z,\frac{1}{2}|z\!-\!u|)$ that
$|[\pa^\mesh G_\G](z;u)|\le \const\cdot |z\!-\!u|^{-1}$. Therefore,
\[
\|\dpa G_{\G}\|_{1,B^\mesh_\DS(z_0,\frac{2}{3}r)} \le \|\dpa G_{\G}\|_{1,B^\mesh_\DS(u,2r)}\le
\const\cdot r.
\]
Furthermore, double-sided bound (\ref{xEstim}) and Corollary~\ref{CorHarnack} imply
\[
[\dpa G^*_{B^\mesh_\G}](z) = \dpa [G^*_{B^\mesh_\G}-(2\pi)^{-1}\log r](z) = O(\const\cdot
r^{-1}),\quad |z\!-\!z_0|\le\tfrac{2}{3}r.
\]
Thus,
\[
\|\dpa G^*_{B^\mesh_\G}\|_{1,B^\mesh_\DS(z_0,\frac{2}{3}r)} \le \const\cdot r.
\]

Otherwise, let $|u-z_0|>\frac{3}{4}r$. We have
\[
G^*_{B^\mesh_\G}(\cdot;u) \le (2\pi)^{-1}\log(2r) + O(1)
\]
on $\pa B^\mesh_\G$, and so on the boundary of the smaller disc
$B^\mesh_\G(z_0,\frac{17}{24}r)$ which still contains $B^\mesh_\DS(z_0,\frac{2}{3}r)$. At the
same time,
\[
G_\G(\cdot;u)\ge (2\pi)^{-1}\log \tfrac{1}{24}r + O(1)~~\mathrm{on}~~\pa
B^\mesh_\G(z_0,\tfrac{17}{24}r).
\]
Thus,
\[
0\ge G=G_\G(\cdot;u)-G^*_{B^\mesh_\G}(\cdot;u)\ge -\const
\]
on the boundary, and so inside $B^\mesh_\G(z_0,\frac{17}{24}r)$. Due to
Corollary~\ref{CorHarnack}, this gives
\[
[\dpa G](z) = O(\const\cdot r^{-1}),~~\mathrm{for}~|z\!-\!z_0|\le\tfrac{2}{3}r.
\]
Hence, $\|\dpa G\|_{1,B^\mesh_\DS(z_0,\frac{2}{3}r)} \le \const\cdot r$.
\end{proof}

\end{document}